\pdfoutput=1
\documentclass[11pt, oneside]{elsarticle}
\usepackage{geometry}
\usepackage{graphicx}
\usepackage{amssymb}
\usepackage{amsmath}
\usepackage{amsthm}
\theoremstyle{plain}

\newtheorem{lemma}{Lemma}
\newtheorem{proposition}{Proposition}
\newtheorem{corollary}{Corollary}
\theoremstyle{definition}

\theoremstyle{remark}

\usepackage{booktabs}
\usepackage{array}
\usepackage{etoolbox}
\AtBeginEnvironment{table}{\setlength{\belowcaptionskip}{8pt}}
\usepackage{multirow}
\usepackage{placeins}

\usepackage{longtable}
\usepackage{microtype}
\usepackage[detect-weight=true,detect-family=true]{siunitx}
\usepackage{xcolor}
\usepackage[hidelinks]{hyperref}
\graphicspath{{figs/}}
\newcommand{\Wvec}{\mathbf{W}}
\newcommand{\Wdot}{\mathbf{W}'}
\newcommand{\gvec}{\mathbf{g}}
\newcommand{\Qvec}{\mathbf{Q}}
\newcommand{\Dmat}{\mathbf{D}}
\newcommand{\Bmat}{\mathbf{L}}
\newcommand{\Phimat}{\boldsymbol{\Phi}}

\usepackage{fancyhdr}
\fancypagestyle{plain}{%
  \fancyhf{}%
  \fancyfoot[R]{\footnotesize\thepage}%
}
\fancypagestyle{pprintTitle}{%
  \fancyhf{}%
}
\begin{document}

\begin{frontmatter}

\title{Structure-preserving reduced-order models of thin panels in supersonic and hypersonic flow}

\author[morphorm]{Miguel A. Aguil\'o\corref{cor1}}
\ead{maguilo@morphorm.com}
\cortext[cor1]{Corresponding author}
\address[morphorm]{Morphorm LLC, 100 Sun Avenue Northeast Suite 650,
Albuquerque, NM 87109-4670, USA}

\begin{abstract}
We report structure-preserving reduced-order models of both challenge problems of the AePW-4 High-Speed Working
Group~\cite{aepw4hswg}. They are the AFRL/SSC RC-19 clamped thin panel at Mach $1.92$ and the UNSW HyMAX cantilever plate
at Mach $5.85$. For each problem one construction gives two models. The reference model is a physics-based
Ritz--Galerkin model of the plate. The learned model keeps the reference model's exact potential, fits the remaining
generalized forces, and trains only its dissipation. Each model is built from the data package the working group
distributes for its case. Neither package holds a time-resolved response or repeated runs. There is no ensemble to learn from. The physics is
therefore carried by the form of the model rather than by data. The learned model's energy balance holds for any weights by construction. On the trained model, at random states, its relative residual is $2.2\times10^{-15}$. On
RC-19, inviscid shock relations applied to the package's wedge geometry and schlieren shock angle predict a pressure
ratio of $1.521$. The ratio derived from the pressure-sensitive paint that Brouwer et al.~\cite{brouwer2023pressure} applied to the rigid
panel is also $1.521$.
On the periodic case we ran the trained model thirty times, each from a different initial push.
The mean dominant frequency is \SI{1.0}{\percent} from the \SI{236}{\hertz} measured by Brouwer et al.~\cite{brouwer2021}.
Brouwer et al.\ simulated the experiment with their own reduced-order model. We take one periodic and one chaotic condition from their simulations.
At both, the trained model gives the same regime as their model in all sixteen runs, eight from each buckled shape.
It does so without their model's \SI{113}{\hertz} spectral peak, which the
experiment does not show. Because the response is chaotic, any model loses a single trajectory after a horizon set by the largest Lyapunov exponent.
Measured displacements therefore correct the model through one more generalized force, which enters the same energy balance.
We tested this correction in a twin experiment at the experiment's \SI{5}{\kilo\hertz} measurement rate. With an offline smoother it holds the panel state to $0.0125$ thicknesses without noise and $0.0352$ with \SI{10}{\micro\meter} of noise. On HyMAX, with nothing fed back, the learned model tracks the reference model's trailing-edge
displacement for \SI{65}{\milli\second} to a relative error of \SI{3.67}{\percent}. Every prediction in the reported studies was registered before its run. Four of nine failed, and all four are reported. The one open comparison is the displacement standard deviation on the RC-19 attached-shock case. There our reference model gives about a sixteenth of the measured value. The gap is traced to inputs the public
package does not carry, and we state the measurement that would close it.
\end{abstract}


\end{frontmatter}

\section{Introduction}
\label{sec:intro}

Thin panels on high-speed vehicles are loaded by the flow, heated by it, and driven by it into
flutter~\cite{dowell1970review,mei1999review,mcnamara2011hypersonic}. A heated clamped panel buckles, and above the
buckling temperature it has two stable buckled shapes. Each sits in a well, and a well is a minimum of the panel's potential energy. In supersonic flow it can oscillate about one of them, switch between
them, or do both in no repeatable order~\cite{dowell1982chaotic}. Which of these it does depends on the temperature rise
and on the pressure across it. A design has to know the answer over the whole operating envelope, not at one point.
High-fidelity coupled simulation resolves one condition at a cost of hours. Linear reduced-order models are affordable
over the envelope but cannot represent a panel that crosses between buckled wells. A learned surrogate can be affordable
and nonlinear. However, a surrogate whose energy change does not equal the work done on it cannot be defended to a reviewer. Nothing in its
training says where it stops being valid.

The AePW-4 High-Speed Working Group~\cite{aepw4hswg,brouwer2025aepw4} poses two challenge problems that test these needs
directly. The AFRL/SSC RC-19 clamped thin panel at Mach $1.92$ has periodic and chaotic responses and a case with an
impinging shock. The UNSW HyMAX cantilever plate sits at Mach $5.85$ under an impinging shock in a short-duration
facility. The group's stated objectives include assessing coupled tools predictively. They also include identifying the
gaps in high-speed experiments that future tests should close.

This work builds learned reduced-order models of both problems. Three observations organize it. First, the data packages the working group distributes for the two problems hold statistics,
not trajectories. The RC-19 package supplies measured mode shapes, time-averaged pressures and
the nine points at which the response is reported. It holds no time history and no repeated runs. The response statistics at
those points are published separately~\cite{brouwer2021,piccoloserafim2025}. A method that learns its physics from
examples needs an ensemble that does not exist. The physics must therefore be carried by the model's form, leaving
little to learn. Second, the panel's non-conservative forces are linear in its state. Piston theory and
Rayleigh damping are both linear. Their coefficients are therefore \emph{fitted} once, by least squares, and frozen, and only
the dissipation is trained. Third, the response is chaotic. Any model of it, however accurate, loses a single trajectory
within a horizon set by the largest Lyapunov exponent $\lambda$. That exponent is the exponential rate at which two
neighboring trajectories separate. Tracking beyond that horizon needs measurement. The correction that measurement applies to the model must enter the energy balance as a generalized force, like the aerodynamic force. Our working hypothesis is the following.
\begin{quote}\itshape
Consider a reduced-order model of a thin panel in high-speed flow, built from the data a challenge-problem
package provides and from published constants. Let it have four properties.
\begin{enumerate}
  \item[(i)] Its elastic and membrane energy is expressed in the panel's measured mode
  shapes. Its nonlinear stiffness force is exactly the gradient of that energy.
  \item[(ii)] Its aerodynamic and structural damping forces enter as generalized forces, with coefficients fitted by
  least squares and then held fixed.
  \item[(iii)] Its only trained term is a dissipation that can remove energy, but never add it, for every value of
  its parameters.
  \item[(iv)] It is advanced in time by a discrete-gradient step. That step replaces the energy gradient by its average along
  the step, which balances energy exactly at every step.
\end{enumerate}
Such a model will satisfy four claims.
\begin{enumerate}
  \item[(P1)] Its energy changes exactly by the work of its generalized forces minus its dissipation, for any trained
  weights and under any measurement correction.
  \item[(P2)] At the simulation operating conditions of Brouwer et al.~\cite{brouwer2021}, it reproduces the regimes of their
  reduced-order model. It does so outside its training temperature and from both buckled wells.
  \item[(P3)] It matches the experiment's response statistics wherever the package supports a comparison.
  \item[(P4)] It tracks the chaotic panel beyond its tracking horizon, given observations at the rate and points the
  experiment already uses.
\end{enumerate}
\end{quote}

We test the hypothesis on both challenge problems, with each claim's acceptance conditions fixed before its run. RC-19
carries the weight of the test, because its package supports comparisons with experiment. P1 is checked at random states
with amplitudes from $10^{-6}$ to $10^{6}$, and along every trajectory. P2 is tested at the two simulation operating conditions of Brouwer et al.~\cite{brouwer2021}, sixteen runs each, eight from each well. P3 is tested against the measured response frequency and displacement standard deviation. P4 is tested in a twin
experiment at the published sampling rate. There the reference model generates synthetic observations of its own
trajectory. HyMAX then tests whether the same construction serves a second
configuration unchanged. Every prediction was registered before its run. Four of the nine registered predictions broke (Table~\ref{tab:predictions}), and they are reported as findings
(Section~\ref{sec:predictions}). P1 and P2 held. P3 held for the frequency, and failed for the displacement standard deviation on the
attached-shock case, which Section~\ref{sec:rc19-open} traces to data the package lacks. P4 held in the twin experiment
and awaits a measured time series.

\paragraph{Place in the series} The earlier reports in this
series~\cite{aguilo2026spml,aguilo2026spmlfsi,aguilo2026spmlopen} reduce fields on a mesh with Whitney forms. These are finite-element bases whose discrete derivatives mirror the
continuous ones. This property is called the discrete de Rham structure. Thin panels are better described by their mode shapes. This report
therefore uses a modal Hamiltonian reduction, in which the energy is written in the amplitudes of the mode shapes. It keeps
the same commitments: exact energy accounting, dissipation that cannot add energy, and discrete balances that hold exactly.

\paragraph{Contributions} This work contributes:
\begin{enumerate}
\item A \emph{reference model built on the RC-19 measurements} (Sections~\ref{sec:methods-panel}
and~\ref{sec:methods-aero}). Its basis is the seven mode shapes measured on the panel before the test, made orthogonal by
L\"owdin's symmetric construction. Of all orthonormal sets, it gives the one nearest to the measured shapes. Its membrane stiffness comes from solving the in-plane elasticity problem for every pair of
modes. Its temperature field comes from the surface pressures measured by fifteen Kulite transducers on a rigid panel.
Those pressures give the adiabatic wall temperature, which is the temperature an insulated wall reaches under the
boundary layer. A fin equation then gives the panel's temperature distribution. It models steady conduction from the heated
panel into its cooler frame, calibrated on the panel and frame thermocouples.
\item A \emph{learned model whose energy balance holds for any parameters} (Proposition~\ref{prop:identity}). The rate of
change of the panel's energy always equals the power of the generalized forces minus the power the learned dissipation
absorbs. Training cannot violate this balance, because it follows from the form of the model. The dissipation is positive definite for
every parameter value. It can therefore only absorb energy (Lemma~\ref{lem:spd}). The generalized aerodynamic force is fitted by
least squares and then frozen (Section~\ref{sec:methods-head}). On the trained model, at random states with amplitudes from $10^{-6}$ to $10^{6}$, the worst relative residual of the
balance is $2.2\times10^{-15}$.
\item A \emph{discrete-gradient step for the learned model} (Proposition~\ref{prop:avf-learned}). The training record is a family of reference model trajectories, generated with the
average-vector-field (AVF) discrete gradient. The generalized-$\alpha$ method is the common alternative. The
two steps carry different truncation errors and numerical damping. A model integrated differently from its data drifts
from them even when its forces are exact. Proposition~\ref{prop:avf-learned} gives the learned model the same discrete-gradient step that generated its
training record. Each such step removes that drift and balances the model's energy exactly, for any trained weights.
The trained models reported here were integrated with the generalized-$\alpha$ step (Section~\ref{sec:methods-time}). The
discrete-gradient step applies to them unchanged.
\item \emph{Assimilation as a generalized force} (Section~\ref{sec:methods-observer}). Measured displacements at the
instrumented points correct the model through one more generalized force. It is proportional to the error in the
displacement and in its rate. Each time step still balances the panel's energy exactly, with the work of this force
reported as a separate term (Corollary~\ref{cor:observer}). The two feedback gains are not tuned. They are set from the
linearized equation for the model-to-truth error so that its largest Lyapunov exponent is negative. The nine published
measurement points observe all seven modes: no combination of modes is invisible at all nine points. In a twin experiment at the published \SI{5}{\kilo\hertz} rate, the
correction with an offline smoother holds the model's state to $0.0125$ thicknesses without noise. A causal hold of
the last sample reaches $0.325$ (Section~\ref{sec:rc19-assimilation}).
\item \emph{Comparisons as statistics} (Sections~\ref{sec:rc19-agreement} and~\ref{sec:rc19-regime}). A thirty-member ensemble of the trained model, at the periodic case's operating condition, places the mean dominant frequency at \SI{238.4}{\hertz}. That is \SI{1.0}{\percent} from the measured \SI{236}{\hertz}.
At the simulation operating conditions of Brouwer et al., it returns the regimes of their reduced-order model. It does so without that
model's \SI{113}{\hertz} peak, which the experiment does not show either. We also measure how an ensemble mean depends on the range of initial perturbations, or kicks, it samples. Every mean
is quoted with that range.
\item A \emph{measured bound on the aerodynamic enrichment} (Section~\ref{sec:rc19-aero}). The reduced-order model of Brouwer et al.\ replaces freestream quantities in piston theory by local ones, which scales the aerodynamic operators by \SI{5.4}{\percent}. Rather than build that enriched model, we measured what the
scaling changes in our reference model. We ran it at the periodic case's experimental operating condition in four
ensembles of $192$ runs each. They combine unscaled and scaled aerodynamic operators with release from each buckled well.
A two one-sided equivalence test~\cite{schuirmann1987tost} compared the mean dominant frequency with and without scaling.
It declares equivalence when the \SI{95}{\percent} interval of the difference of means lies inside a band, here
$\pm\SI{2}{\percent}$, registered before the runs. The mean frequency moved \SI{0.40}{\percent} in the lower well and \SI{0.51}{\percent} in the upper well, inside the
band. The fluctuation amplitude moved \SI{2.4}{\percent}. Its elasticity, the percent change in amplitude per percent
change in the coefficients, is $2.4/5.4 = 0.44$.
\item A \emph{measured data requirement for the attached-shock case} (Section~\ref{sec:rc19-open}). Forced at the measured pressure fluctuation level, the reference model's displacement standard deviation there is $0.06$
thicknesses. The measured value is $0.943$, about sixteen times larger. Ruling out every admissible cause by measurement
isolates one remaining mechanism. That mechanism is the dependence of the surface pressure on the panel's deformed shape. Testing it requires
time-resolved surface pressure, or a computation, on the deformed panel, which the package lacks.
\item \emph{One construction across two configurations} (Section~\ref{sec:hymax}). The HyMAX model uses the RC-19
construction unchanged. With no observations fed back, it tracks the reference model's trailing-edge displacement for \SI{64.9}{\milli\second}. Its relative root-mean-square
error over that window is \SI{3.67}{\percent}.
It also detects, rather than assumes, which modal bases its data support. A six-mode basis is declined, because its
whitened coordinates keep only eleven of its twelve state directions.
\end{enumerate}

\section{The challenge problems and what their packages hold}
\label{sec:problems}

Both packages are distributed by the working group and marked Distribution A. They are used as distributed, and no number in this report is fitted to a measured flutter response. Two
structural calibrations are stated where they are made: the thermal effectiveness on the buckling temperature rise, and
the HyMAX bending stiffness on the first frequency.

\subsection{AFRL/SSC RC-19 clamped thin panel}
\label{sec:problems-rc19}

\paragraph{The facility and the panel} The experiment~\cite{spottswood2019exploring,brouwer2021experiments,aepw4hswg} runs in a
blowdown supersonic tunnel with a test section
\SI{131.1}{\milli\meter} high and \SI{152.4}{\milli\meter} wide, at a nominal Mach number of $2.0$. RANS simulations in the
package give $1.92 \pm 0.02$ at the panel~\cite{aepw4hswg}. The panel is a steel plate of $253.9 \times 127 \times 0.635$~mm, set flush in the tunnel floor
(Figure~\ref{fig:rc19-setup}). A frame of $304.7 \times 152.4$~mm, \SI{12.85}{\milli\meter} thick, clamps it on all four
edges. A sealed cavity behind the panel holds the cavity pressure $p_c$, which sets the static
pressure differential across it. The incoming boundary layer is \SI{8.6}{\milli\meter} thick. In the heated runs the stagnation
temperature is about \SI{388}{\kelvin} at a stagnation pressure of \SI{346}{\kilo\pascal}. The flow heats the panel faster than its frame. The panel's temperature therefore rises above the frame's. The frame
constrains the panel's thermal expansion, which puts the panel into compression. It buckles once the rise exceeds about \SI{5.6}{\kelvin}, which is the flat panel's
buckling temperature rise reported in~\cite{brouwer2021}. One thermocouple on the panel and one on the frame measure their temperatures. The package's temperature rise
$\Delta T$ is the difference between the two readings.

\paragraph{The cases} The package defines three cases. Two have no shock on the panel and are labeled periodic and
chaotic. The periodic case runs at a cavity pressure of \SI{50.4}{\kilo\pascal} and a temperature rise of
\SI{12.76}{\kelvin}. Brouwer et al.~\cite{brouwer2021} report its measured dominant frequency as \SI{236}{\hertz}. The chaotic case runs at
\SI{51.5}{\kilo\pascal} and a temperature rise of \SI{14.74}{\kelvin}. Its cavity pressure is \SI{0.49}{\kilo\pascal}
above the RANS reference surface pressure $p_s = \SI{51.0}{\kilo\pascal}$ of~\cite{brouwer2021}. Brouwer et al.~\cite{brouwer2021} report its measured dominant frequency as \SI{257}{\hertz}. The third case, attached-shock,
mounts a wedge of length \SI{110.7}{\milli\meter} and height \SI{7.8}{\milli\meter} in the test section. Its leading
edge is \SI{86}{\milli\meter} upstream of the panel's. The wedge places an oblique shock, measured at \SI{35.7}{\degree} in the
schlieren, on the panel. That case runs at a cavity pressure of \SI{68.9}{\kilo\pascal} and is labeled periodic.

\paragraph{The instrumentation and what the package holds} For each case the package supplies the geometry of the tunnel, the wedge and the panel, and the operating conditions
with their error bars. It supplies seven mode shapes measured before the test on an eleven by eleven grid, with their
frequencies and damping (Figure~\ref{fig:rc19-modes}). The panel's edges were free during these modal measurements. It supplies the time-averaged surface pressure measured on a rigid
panel by $15$ Kulite transducers and $28$ scanner taps. For the attached-shock case it adds the schlieren and a pressure-sensitive paint field~\cite{brouwer2023pressure}. The response was measured by digital image
correlation at \SI{5}{\kilo\hertz}~\cite{aepw4hswg}. The package defines nine extraction points (Figure~\ref{fig:rc19-setup}b). It does not hold the mean and the standard
deviation of the displacement at those points. We take them from the published papers~\cite{brouwer2021,piccoloserafim2025}.
The package's spectral settings are Hamming windows of \SI{0.2048}{\second} with half overlap. The
resolution is therefore $1/\SI{0.2048}{\second} = \SI{4.88}{\hertz}$. Spectra and amplitudes in this report are read at three-quarter chord. The package supplies no time history of the response.

\begin{figure}[htbp]
  \centering
  \includegraphics[width=\textwidth]{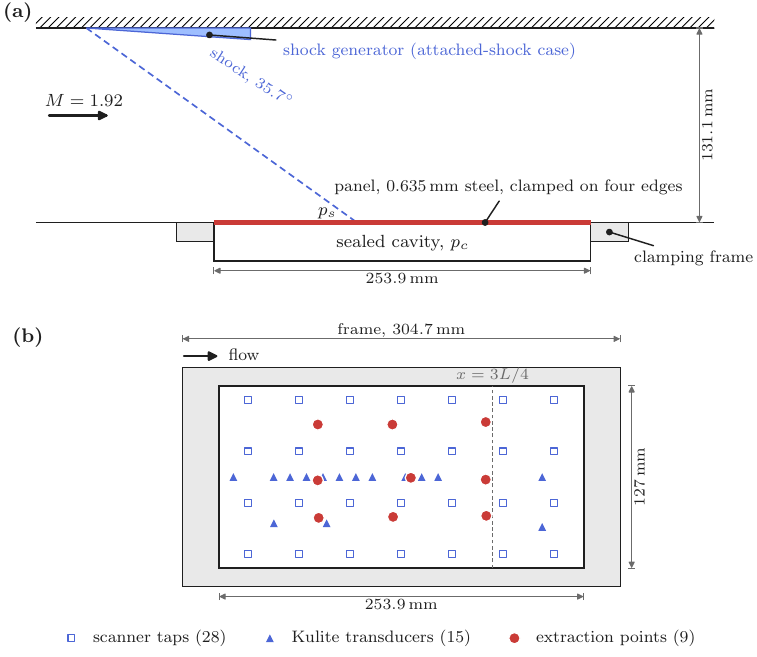}
  \caption{The RC-19 configuration, with dimensions from the package. (a) The test section in side view. The panel is
  flush in the floor over a sealed cavity at pressure $p_c$. Its flow side sees the surface pressure, whose RANS reference value is $p_s$. The shock generator and its shock, at the measured \SI{35.7}{\degree}, are present on the attached-shock case only. The
generator's wall and the cavity depth are schematic. (b) The planform, with the clamping frame shaded, the rigid-panel pressure taps and the nine extraction points. The dashed
line is the three-quarter-chord station at which spectra and amplitudes are read.}
  \label{fig:rc19-setup}
\end{figure}

\begin{figure}[htbp]
  \centering
  \includegraphics[width=\textwidth]{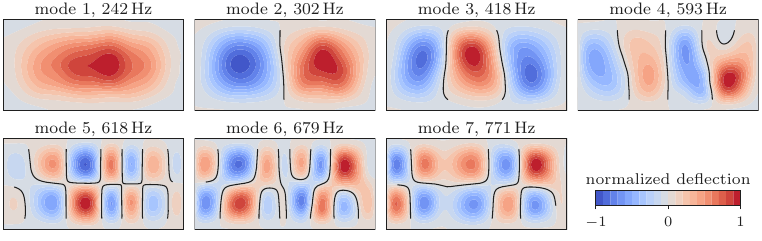}
  \caption{The seven RC-19 mode shapes measured before the test, with their measured frequencies. They are interpolated
by bicubic splines on the package's grid, as in Section~\ref{sec:methods-panel}. Each is normalized to unit peak with its
  larger lobe positive. Black lines are nodal lines, where the mode's deflection is zero. They separate regions that deflect in
  opposite directions.}
  \label{fig:rc19-modes}
\end{figure}

\subsection{UNSW HyMAX cantilever plate}
\label{sec:problems-hymax}

\paragraph{The facility and the plate} The HyMAX experiment of UNSW Canberra~\cite{currao2019} runs in the TUSQ Ludwieg tube of the University of Southern
Queensland. The test flow is a free jet of compression-heated air at Mach $5.85$. A forebody \SI{100}{\milli\meter} long carries a flat plate of 6061-T6, $130 \times 80 \times 2$~mm
(Figure~\ref{fig:hymax-setup}a). The plate is attached by a bolted joint at its hinge line and is free at its trailing
edge, \SI{230}{\milli\meter} from the forebody's leading edge. The shock is produced by a static shock generator~\cite{currao2021oscillating}. The generator is a wedge \SI{110}{\milli\meter} long with an
included angle of \SI{5}{\degree}. Its leading edge sits \SI{10}{\milli\meter} upstream of the forebody's and \SI{40}{\milli\meter} above its
surface. Its angle of attack is $\alpha_g = \SI{-3}{\degree}$. The package defines the flow deflection as $\alpha_g$ plus the
wedge angle, $\SI{-3}{\degree} + \SI{5}{\degree} = \SI{2}{\degree}$. That is
the case the working group makes mandatory, and \SI{10}{\degree} is optional. The inviscid oblique shock from its leading edge stands at \SI{11.16}{\degree}. That is the weak-branch solution of the
oblique-shock relation~\cite{ames1953naca1135} for a \SI{2}{\degree} deflection at Mach $5.85$. The shock meets the plate's
plane $\SI{40}{\milli\meter} / \tan \SI{11.16}{\degree} = \SI{203}{\milli\meter}$ behind the generator's leading edge.
That is \SI{93}{\milli\meter} behind the hinge, beside the first of the two shock positions the training trajectories use, $x_s = 91$ and \SI{117}{\milli\meter}. The flow lasts about
\SI{200}{\milli\second}, and its freestream pressure is not stationary over the first \SI{90}{\milli\second} of that
window (Figure~\ref{fig:hymax-setup}b).

\paragraph{What the package holds} The package supplies one run, with its freestream pressure and temperature and the geometry of both configurations. It
holds the surface pressure at four transducers sampled at \SI{50}{\kilo\hertz}, with their calibrations. Three sit on the
forebody at $40$, $50$ and \SI{60}{\milli\meter}, and one on the plate near its trailing edge at \SI{220}{\milli\meter}.
Schlieren images show the flexible plate and a rigid one at shock-generator angles of $2$ and \SI{10}{\degree}. The transducer on the plate sits under the impinging shock and reads about twice the
freestream pressure. The one structural measurement is the first frequency, \SI{83}{\hertz}~\cite{currao2019}. Xu and
Dowell~\cite{xu2025} compute \SI{84.7}{} and \SI{530.4}{\hertz} for the first two modes, which we carry as a computational
reference rather than as data. The package holds no displacement. It states that a measured trailing-edge displacement will be provided for
validation~\cite{aepw4hswg}. It is not used here.

\begin{figure}[htbp]
  \centering
  \includegraphics[width=\textwidth]{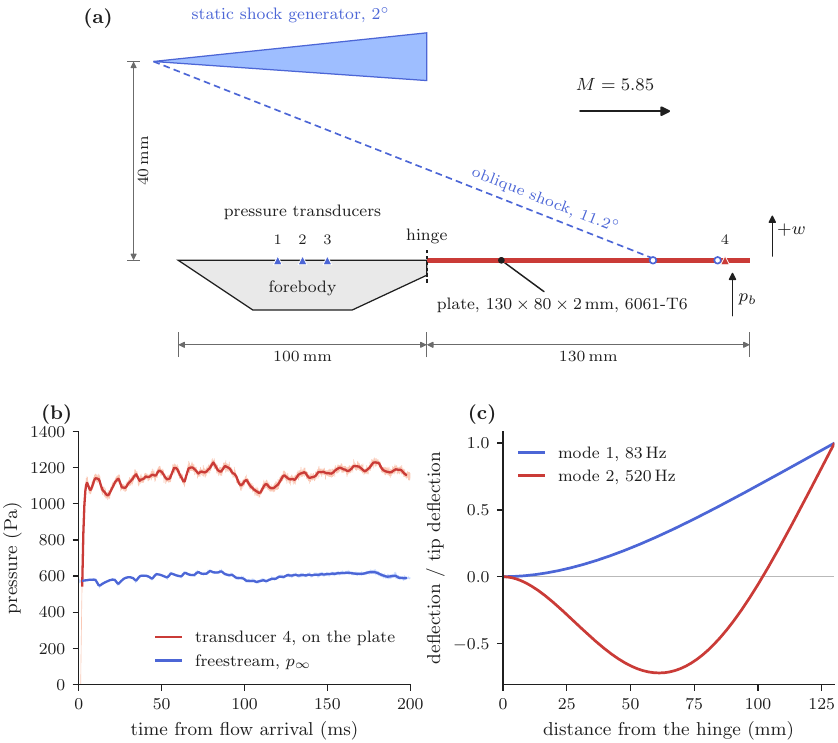}
  \caption{The HyMAX configuration, \SI{2}{\degree} case. (a) Side view, dimensioned from the UNSW data package. The vertical scale is exaggerated two to one; angles are labeled with their true values. It
shows the static shock generator, the forebody, the hinge, the plate and the four pressure transducers. The dashed line is
the inviscid oblique shock from the generator's leading edge at Mach $5.85$. Open circles mark the two shock positions of the training trajectories, $x_s = 91$ and
  \SI{117}{\milli\meter} behind the hinge. The back pressure $p_b$ is not
  measured, and $w$ is positive toward the shock side. (b) The measured freestream pressure and the trailing-edge transducer over the \SI{200}{\milli\second} flow window. These are the inputs of the aerodynamic force, raw (light) and averaged over \SI{2}{\milli\second} (dark). (c) The two clamped-free modes of the
  reference model, calibrated on the measured first frequency.}
  \label{fig:hymax-setup}
\end{figure}

\section{Background and related work}
\label{sec:background}

\paragraph{Panel flutter and its reduced-order models} Panel flutter has been studied for decades, with first-order
piston theory~\cite{lighthill1953oscillating,ashley1956piston} as its standard aerodynamic model~\cite{dowell1970review,
dowell1975plates}. Nonlinear panel flutter at supersonic and hypersonic speeds is reviewed in~\cite{mei1999review}, and
its aerothermoelastic setting in~\cite{mcnamara2011hypersonic}. A heated buckled plate in flutter can be chaotic, as
Dowell showed for a deterministic autonomous model~\cite{dowell1982chaotic}. Reduced-order models of geometrically
nonlinear structures are reviewed in~\cite{mignolet2013review}. The RC-19 experiments that the first challenge problem
draws on are reported in~\cite{spottswood2019exploring,brouwer2021experiments,brouwer2023pressure}.

\paragraph{The reduced-order model of Brouwer et al.} Brouwer et al.~\cite{brouwer2021} built a reduced-order model and the
simulations that accompany the RC-19 experiment. Their model projects a von K\'arm\'an plate on finite element modes and loads it with piston theory corrected for the boundary layer~\cite{brouwer2019enriched}. It references the static load to the computed surface pressure and damps the panel in Rayleigh form. It reproduces the published regimes, with two qualifications the authors state.
Their periodic simulation required a cavity pressure \SI{4.6}{\percent} from the measured one and a temperature rise
\SI{44.5}{\percent} from it. They attribute both offsets to the distribution of the temperature rise. Their reduced model
also shows a spectral peak near \SI{113}{\hertz} that the experiment does not show. They establish that the response is chaotic with a positive largest Lyapunov exponent, and report how it grows with
temperature. The exponent is computed from their reduced model's trajectories by the linearization of their Equations 19
to 23. Our reference model follows their
Sections 2.2 and 5 closely, and the differences are stated in Section~\ref{sec:methods}.

\paragraph{The other working-group analyses} Piccolo Serafim et al.~\cite{piccoloserafim2025} correlate
computations with the RC-19 measurements in the presence of an impinging shock. Xu and Dowell~\cite{xu2025} model the
HyMAX plate computationally, and Currao et al.~\cite{currao2019} report the HyMAX experiment, its response and its
first structural frequency. Ganti et al.~\cite{ganti2025panel} couple a wall-modeled large-eddy simulation to a
finite element plate for a cantilevered aluminum panel at Mach $2.11$. That panel overhangs a backward-facing step and
is not the RC-19 panel. Their simulated response frequency matches the measurement they compare against to within
\SI{1}{\hertz}. Brouwer~\cite{brouwer2025aepw4} summarizes the working group's cases and its first
round of results.

\paragraph{Structure in learned models} A learned model can be made to respect a physical structure in two ways.
The structure can be added to the objective as a penalty, as in physics-informed networks~\cite{raissi2019pinn}. It then
holds to the degree the penalty weight buys, and only where the penalty was sampled. Alternatively, it can be built into the parameterization, and then it holds for every value of the parameters. Examples are
Hamiltonian networks~\cite{greydanus2019hnn}, port-Hamiltonian models~\cite{vanderSchaftJeltsema2014} and
port-Hamiltonian networks~\cite{desai2021porthnn}. The models of this
report are of the second kind. The panel is open: the flow does work on it, and the work can have either sign. The
port-Hamiltonian form accounts for that work exactly instead of excluding it. The average-vector-field discrete gradient~\cite{quispel2008avf,gonzalez1996discrete,mclachlan1999discrete}
carries the accounting through each time step exactly.

\section{Notation}
\label{sec:notation}

Table~\ref{tab:notation} collects the symbols used in the rest of this report. Bold lowercase letters are vectors over the
retained modes, bold uppercase letters are matrices, and a calligraphic letter is a tensor of order four. An overdot is a
derivative in physical time $t$, and a prime on a nondimensional state is a derivative in the nondimensional time $\tau = \omega_1 t$. On the HyMAX beam, a
prime on a shape or a deflection is a derivative in $x$. Here $\omega_1$ is the
first modal circular frequency.

{\small
\begin{longtable}{p{0.20\textwidth}p{0.72\textwidth}}
  \caption{Symbols.}\label{tab:notation}\\[2pt]
  \toprule
  Symbol & Meaning \\
  \midrule
  \endfirsthead
  \toprule
  Symbol & Meaning \\
  \midrule
  \endhead
  \bottomrule
  \endfoot
  \multicolumn{2}{l}{\emph{Flow}} \\
  $M$, $\beta$ & Mach number; $\beta = \sqrt{M^2 - 1}$ \\
    $p_0$, $T_0$, $p_{0,\mathrm{r}}$ & stagnation pressure and temperature of a heated run; stagnation pressure of the rigid run \\
  $p_\infty$, $q_\infty$, $V_\infty$, $T_\infty$ & freestream static pressure, dynamic pressure, speed, static temperature \\
  $f_M$, $R_{\mathrm{air}}$, $\mathrm{Pr}$ & isentropic factor $1 + \tfrac{\gamma-1}{2}M^2$; gas constant of air; Prandtl number \\
  $\gamma$ & ratio of specific heats, $1.4$ \\
  $p_{\mathrm{pt}}$ & piston-theory pressure, Equation~\eqref{eq:piston-pressure} \\
  $p_K$ & static pressure at a Kulite transducer on the rigid RC-19 panel \\
  $p_c$, $p_s$, $\Delta p$ & RC-19 cavity pressure, reference surface pressure, $\Delta p = p_c - p_s$ \\
  $T_{aw}$, $r_c$ & adiabatic wall temperature and recovery factor, Equation~\eqref{eq:recovery} \\
  $M_\ell$ & local Mach number: on RC-19, from the measured surface pressure; on HyMAX, freestream or post-shock \\
  $q_\ell$, $\beta_\ell$, $V_\ell$ & local dynamic pressure, supersonic factor and flow speed on the HyMAX plate,
    Equation~\eqref{eq:local-piston} \\
  $p_b$, $x_s$, $\Pi$ & HyMAX back pressure, shock position from the hinge, measured shock pressure ratio \\
  $\alpha_g$ & HyMAX forebody angle of attack \\
  $M_n$, $\Theta_s$, $\Theta_w$ & HyMAX shock-normal Mach number, shock angle, flow deflection \\
  $p_\ell$, $T_\ell$ & HyMAX post-shock static pressure and temperature \\
  $p_4$, $p_{\mathrm{top}}$, $p_{\mathrm{st}}$ & HyMAX trailing-edge transducer pressure; flow-side pressure on the plate;
    static pressure distribution, Equation~\eqref{eq:hymax-port} \\
  $a$, $c$ & local piston-theory stiffness and damping coefficients on the HyMAX plate,
    Equation~\eqref{eq:local-piston} \\
  \midrule
  \multicolumn{2}{l}{\emph{Structure}} \\
  $x$, $y$, $t$ & streamwise and spanwise coordinates; physical time \\
  $n$ & number of retained modes \\
    $\Omega$, $h$, $\rho$ & RC-19 panel planform; plate thickness, of the RC-19 panel or the HyMAX plate by section; density \\
  $E$, $\nu$, $\alpha$, $\mathbb{C}$ & Young's modulus, Poisson's ratio, thermal expansion, plane-stress stiffness times $h$ \\
  $w$, $\mathbf{u}$ & transverse deflection; in-plane displacement \\
  $\mathcal{K}$, $U_b$, $U_m$, $\mathcal{W}_s$ & kinetic, bending and membrane energy of the plate; work of the static pressure \\
  $a_b$ & symmetric bilinear form of the bending energy, $U_b = \tfrac12 a_b(w, w)$ \\
  $\delta\mathcal{W}$, $\mathcal{W}$ & virtual work of the pressures on the HyMAX plate; total work of the aerodynamic force and
    the damping, Figure~\ref{fig:energy} \\
  $D_b$, $\mathbf{K}_b$ & flexural rigidity $E h^3 / [12(1-\nu^2)]$; modal bending stiffness matrix \\
  $\mathcal{L}_g$, $\tilde{\mathbf{M}}$ & Lagrangian $\mathcal{K} - U$; mass matrix of the orthogonalized shapes \\
  $\mathbf{u}_{ij}$, $\boldsymbol{\eta}_{ij}$ & in-plane displacement and membrane strain of the mode pair $(i,j)$ \\
  $\boldsymbol{\sigma}_T$, $\mathbf{u}_T$ & thermal stress per kelvin; in-plane displacement under the thermal strain alone \\
  $\mathbf{e}_q$ & membrane strain carried by the modal amplitudes, Equation~\eqref{eq:Um-split} \\
  $\boldsymbol{\varepsilon}$, $\boldsymbol{\eta}$, $\boldsymbol{\varepsilon}_T$ & in-plane strain, von K\'arm\'an
    membrane strain, thermal strain \\
  $\Delta T$, $s_T$, $\eta_T$ & temperature rise, panel thermocouple minus frame thermocouple; unit-mean shape of the panel's temperature excess; thermal effectiveness \\
  $\vartheta$, $\bar\vartheta$, $L_f$ & panel temperature excess over the frame, its panel mean, fin length,
    Equation~\eqref{eq:fin} \\
  $T_{\mathrm{frame}}$, $\bar T_{aw}$ & frame temperature; panel mean of the adiabatic wall temperature \\
  $k$, $h_c$ & thermal conductivity of the panel; film coefficient of the boundary layer \\
  $\phi_i$, $\boldsymbol{\phi}$, $\mathbf{M}$ & measured RC-19 mode shapes, singly and as a vector; their mass matrix, $(\mathbf{M})_{ij} = \rho h \int \phi_i \phi_j
    \,\mathrm{d}A$ \\
  $\tilde\phi_i$, $\tilde{\boldsymbol\phi}$ & unit-mass RC-19 mode shapes (measured, orthogonalized), singly and as a vector \\
  $\psi_i$, $\boldsymbol\psi$ & unit-mass HyMAX mode shapes (clamped-free), singly and as a vector \\
  $b$, $L$, $A_c$, $I$, $EI$ & HyMAX plate width, length, section area, second moment of area, bending stiffness \\
  $\xi$, $\kappa_i$, $\chi_i$ & position along the beam $x/L$; clamped-free eigenvalue roots; shape coefficient,
    Equation~\eqref{eq:cf-modes} \\
  $\mathbf{I}$, $\delta_{ij}$ & identity matrix; Kronecker delta \\
  $f_i$, $\omega_i$, $\boldsymbol{\omega}$ & modal frequency, $\omega_i = 2\pi f_i$; the vector of the $\omega_i$ \\
  $\mathbf{q}$, $\mathbf{p}$ & modal displacement and momentum, $\mathbf{p} = \dot{\mathbf{q}}$ \\
  $H$, $U$ & Hamiltonian; its potential part \\
  $\mathbf{G}$, $n_T$ & thermal preload operator, Equation~\eqref{eq:thermal-operator}; stress resultant per kelvin \\
  $\mathbf{K}(\Delta T)$ & linear stiffness of the heated flat panel, Equation~\eqref{eq:reference model-H} \\
  $\mathcal{D}$ & quartic membrane tensor, Equation~\eqref{eq:quartic} \\
  $\boldsymbol\ell$, $\boldsymbol\ell_f$ & generalized load of a unit pressure; of a measured pressure field \\
  $\mathbf{C}$, $\alpha_d$, $\beta_d$, $\zeta$ & reference model damping; Rayleigh coefficients; HyMAX modal damping ratio \\
  $c_i$ & Rayleigh damping of mode $i$ in the units of $\tau$, $(\alpha_d + \beta_d \omega_i^2)/\omega_1$ \\
  $\mathbf{A}$, $\mathbf{B}$ & aerodynamic stiffness and damping operators \\
  $\mathbf{f}$, $\mathbf{g}$ & non-conservative generalized force; prescribed drive, dimensional in the reference model and
    divided by $w_s \omega_1^2$ in the learned model \\
  $\mathbf{J}$, $\mathbf{R}$ & skew-symmetric structure matrix; dissipation matrix, Equation~\eqref{eq:phs} \\
  $\Delta t$, $\Delta\tau$ & integrator step size, in physical time and in $\tau$ \\
  $\rho_\infty$ & spectral radius at infinite frequency of the generalized-$\alpha$ step \\
  $\mathbf{q}_m$, $\mathbf{p}_m$, $\mathbf{f}_m$, $\overline{\nabla U}$ & step midpoint values; averaged potential gradient,
    Equation~\eqref{eq:avf} \\
  \midrule
  \multicolumn{2}{l}{\emph{Learned model}} \\
  $\Wvec$, $\Wdot$, $\tau$ & nondimensional modal displacement (thicknesses), its rate, time $\tau = \omega_1 t$ \\
  $w_s$ & displacement scale: $W_1 = 1$ is one thickness of peak first-mode deflection \\
  $\mathbf{x}$, $\mathbf{z}$ & state: $(\mathbf{q}, \mathbf{p})$ for the reference model, $(\Wvec, \Wdot)$ for the learned
    model; conditioning parameters \\
  $U_{\mathrm{ref}}$ & potential of the reference model, before nondimensionalization \\
  $N_m$, $N_t^{(m)}$, $N_r$, $\Delta t_r$ & number of record members; time samples of member $m$; time samples of the whole
    record; time-sample spacing \\
  $\bar{\mathbf{x}}$, $\boldsymbol{\Sigma}_x$ & record mean and covariance of the state \\
  $\mathbf{V}$, $\mathbf{S}$ & eigenvectors of $\boldsymbol{\Sigma}_x$; diagonal matrix of the square roots of its eigenvalues \\
  $\tilde{\mathbf{x}}$, $\bar{\tilde x}$, $\sigma_{\tilde x}$ & rotated and scaled state; the scalar mean and standard
    deviation of its entries over the record \\
  $\Dmat$, $\mathbf{d}$, $\Bmat$, $r$ & learned dissipation, Equation~\eqref{eq:dissipation}; its diagonal parameters,
    low-rank factor, and rank \\
  $\gamma_D$, $\boldsymbol{\Lambda}$ & dissipation floor; diagonal part of $\Dmat$ \\
  $\Qvec$, $\mathbf{A}_h$, $\mathbf{b}_h$ & fitted force term, Equation~\eqref{eq:head}; its coefficient matrix and
    constant offset \\
  $\mathbf{r}$, $\mathbf{v}$ & fit target, the residual acceleration of Equation~\eqref{eq:head-target}; features of the
    fitted force term \\
  $\hat{\mathbf{z}}$, $\bar z_j$, $\sigma_{z,j}$ & standardized conditioning parameters, Equation~\eqref{eq:z-standard};
    record mean and standard deviation of entry $j$ \\
  $\epsilon_r$ & ridge weight of the fit, Equation~\eqref{eq:ridge} \\
  $\mathbf{A}_W$, $\mathbf{A}_{W'}$, $\mathbf{A}_z$ & blocks of $\mathbf{A}_h$ acting on displacements, rates and
    conditioning parameters \\
  $\mathbf{F}$ & residual of the AVF step, solved by Newton's method \\
    $\theta$, $J$ & trained parameters; objective, Equation~\eqref{eq:loss} \\
  $\mathbf{s}$, $\hat{\mathbf{s}}$, $N$ & recorded and predicted states in whitened coordinates; number of terms in the
    data sum \\
  $\mu$, $c_u$, $s_{\max}$ & weight of the bound term; bound radius factor; largest state norm in the record \\
  $R_{\mathrm{sat}}$, $R_{\mathrm{knee}}$, $\varrho$ & radius and knee of the radial saturation of the rollout; its radial map \\
  \midrule
  \multicolumn{2}{l}{\emph{Chaos and assimilation}} \\
  $\lambda$ & largest Lyapunov exponent \\
  $\boldsymbol{\delta}$ & difference between two neighboring solutions, Equation~\eqref{eq:error} \\
  $\varepsilon_0$, $\sigma$, $T_h$ & initial separation; state scale, about one thickness; tracking horizon,
    Equation~\eqref{eq:horizon} \\
  $n_s$ & number of instrumented points \\
  $\Phimat$, $\mathbf{M}_\Phi$ & mode shapes at the instrumented points; observability Gramian $\Phimat\Phimat^{\mathsf{T}}$ \\
  $\mathbf{y}$ & displacement readings at the instrumented points \\
  $\mathbf{y}_{\mathrm{obs}}$, $\mathbf{f}_{\mathrm{obs}}$ & observed readings; assimilation force \\
  $\Delta t_s$, $t_j$, $\upsilon$ & observation interval; arrival time of sample $j$; elapsed fraction $(t - t_j)/\Delta t_s$ \\
  $k_p$, $k_d$ & observer gains on the reading and on its rate \\
  \midrule
  \multicolumn{2}{l}{\emph{Statistics}} \\
  $N_e$ & number of ensemble members \\
  $e_{\mathrm{rms}}$, $N_w$, $w_{\mathrm{te}}$ & HyMAX relative root-mean-square error, Equation~\eqref{eq:hymax-rms};
    time samples in its window; trailing-edge displacement \\
  $\delta_{\mathrm{eq}}$, $s_{\mathrm{pool}}$, $z^\ast_q$ & equivalence tolerance; pooled standard deviation; standard normal
    quantile at probability $q$ \\
\end{longtable}}

\section{Methods}
\label{sec:methods}

Each case is modeled at two levels, and the distinction runs through the rest of this report. The \emph{reference model} is a reduced physical model of the panel, assembled from the challenge package and from
published constants with no training. It is a Hamiltonian in the measured modes, a generalized aerodynamic force, and an integrator
whose discrete energy balance is exact. The \emph{learned model} has the same
Hamiltonian form. Its non-conservative forces are fitted and its dissipation is trained on a \emph{record}. A record is a set of reference model trajectories (Section~\ref{sec:methods-record}). After training, the learned model is called the trained model. Neither package holds a time-resolved response. The record is therefore generated
rather than measured. The experimental measurements the package does hold, the response statistics,
the surface pressures and the schlieren, grade the reference model against the experiment. The reference model grades
the learned model. Table~\ref{tab:levels} states which level produced each result reported here.
Where a result was first computed on the reference model, it was repeated on the trained model with the protocol
unchanged. The prediction for each repeat was registered before its run (Section~\ref{sec:predictions}). On RC-19, the trained model's fitted force term is a linear least-squares fit of the generalized aerodynamic and
damping forces. It reproduces the reference model's non-conservative forces over the record. The relative root-mean-square error is $3.6\times10^{-7}$, which is about one part in three million
(Section~\ref{sec:methods-head}). Agreement between the two levels therefore measures the fitted structure-preserving model
and its integrator. It does not show that the
dynamics were learned by training.

\begin{table}[htbp]
  \centering
  \caption{The level of the construction behind each result.}
  \label{tab:levels}
  \small
  \begin{tabular}{p{0.52\textwidth}p{0.38\textwidth}}
    \toprule
    Result & Produced by \\
    \midrule
    Energy identity at random states & reference model, and the trained model \\
    Discrete budget along a trajectory & reference model \\
    Shock pressure ratio and adiabatic wall temperature & derived from experimental measurements \\
    Regime map at the simulation operating conditions of Brouwer et al. & reference model and trained model, sixteen repeats each, eight released from each buckled shape \\
    Ensemble response frequency and amplitude & reference model and trained model, thirty realizations each \\
    Aerodynamic sensitivity & reference model, four ensembles of $192$ \\
        Lyapunov exponent & reference model and trained model, at $5.7$, $7.2$, $8.7$, $10.2$, $11.7$ and \SI{14.74}{\kelvin} \\
    Assimilation twin experiment & reference model, observing itself \\
    Force fit, spectral band, amplitude and mean deflection against the record & learned model \\
    HyMAX trailing-edge rollout & learned model \\
    \bottomrule
  \end{tabular}
\end{table}

\subsection{The RC-19 reference model}
\label{sec:methods-panel}

The panel is $253.9 \times 127 \times 0.635$~mm and clamped on its four edges (Figure~\ref{fig:rc19-setup}). The
coordinate $x$ runs streamwise from its leading edge, and $y$ spanwise from its centerline. The model follows Sections 2.2 and 5 of Brouwer et al.~\cite{brouwer2021}, with two differences. It is built on the
package's measured mode shapes rather than on a finite element basis. Its temperature field is derived from the package's
own pressures.

\paragraph{The continuum model} Consider a von K\'arm\'an plate~\cite{dowell1975plates} occupying the planform $\Omega$, of thickness $h$ and density
$\rho$. Its boundary $\partial\Omega$ is the panel's four edges, where the frame clamps it. The plate has transverse deflection $w(x,y,t)$, positive away from the cavity and into the flow, and in-plane displacement
$\mathbf{u}(x,y,t)$.
Four energies describe it. The kinetic and bending energies are
\begin{equation}
  \mathcal{K}[w] \;=\; \frac12 \int_\Omega \rho h\, \dot w^2 \,\mathrm{d}A, \qquad
  U_b[w] \;=\; \frac12 \int_\Omega D_b \Big[ (\nabla^2 w)^2 - 2(1 - \nu)\big( w_{,xx} w_{,yy} - w_{,xy}^2 \big) \Big]
  \,\mathrm{d}A ,
  \label{eq:K-Ub}
\end{equation}
with the flexural rigidity $D_b = E h^3 / [12 (1 - \nu^2)]$. A subscript after a comma denotes partial differentiation,
as in $w_{,xy} = \partial^2 w / \partial x\, \partial y$. The membrane energy couples $w$ to $\mathbf{u}$,
\begin{equation}
\begin{gathered}
  U_m[w] \;=\; \min_{\mathbf{u} = 0 \text{ on } \partial\Omega} \; \frac12 \int_\Omega
  \big( \boldsymbol{\varepsilon}(\mathbf{u}) + \boldsymbol{\eta}(w) - \boldsymbol{\varepsilon}_T \big) : \mathbb{C} :
  \big( \boldsymbol{\varepsilon}(\mathbf{u}) + \boldsymbol{\eta}(w) - \boldsymbol{\varepsilon}_T \big) \,\mathrm{d}A, \\
  \boldsymbol{\eta}(w) = \tfrac12 \nabla w \otimes \nabla w .
\end{gathered}
  \label{eq:membrane-energy}
\end{equation}
Here $\boldsymbol{\varepsilon}(\mathbf{u}) = \tfrac12 (\nabla\mathbf{u} + \nabla\mathbf{u}^{\mathsf{T}})$ is the linear
in-plane strain. $\mathbb{C}$ is the plane-stress stiffness
multiplied by $h$, with Young's modulus $E$ and Poisson's ratio $\nu$. $\boldsymbol{\varepsilon}_T = \alpha\, \eta_T\,
\Delta T\, s_T(x,y)\, \mathbf{I}$ is the thermal strain of a temperature rise $\Delta T$ above the frame. Here $\alpha$ is the coefficient of thermal
expansion and $\mathbf{I}$ the in-plane identity tensor. It is distributed with
the unit-mean shape $s_T$ and scaled by the thermal effectiveness $\eta_T$. Holding $\mathbf{u} = 0$ on $\partial\Omega$ is the clamped condition in the plane. The static pressure differential $\Delta p$ across the panel does the work
\begin{equation}
  \mathcal{W}_s[w] \;=\; \Delta p \int_\Omega w \,\mathrm{d}A ,
  \label{eq:Ws}
\end{equation}
plus the work of the measured surface pressure's deviation from its reference value, when that deviation is kept. The
total potential energy is $U = U_b + U_m - \mathcal{W}_s$. The flow's unsteady pressure $p_{\mathrm{pt}}$ is not conservative. It enters as the
generalized aerodynamic force of Section~\ref{sec:methods-aero}, through its virtual work $-\int_\Omega p_{\mathrm{pt}}\, \delta w
\,\mathrm{d}A$. The reference model never evaluates $D_b$. The measured frequencies stand in for the bending energy, as the
term list below explains.

\paragraph{The basis} The challenge package supplies seven mode shapes measured on the panel before the test
(Figure~\ref{fig:rc19-modes}). The basis retains all seven, giving $n = 7$ retained modes. The package
supplies no further measured shapes. Each measured shape is
interpolated by a bicubic spline on the package's eleven by eleven grid.
The integrals below are evaluated by four-point Gauss--Legendre quadrature on every spline cell. A bicubic spline is a
cubic polynomial in each variable on every cell. A product of two spline factors therefore has degree at most six in
each variable. The four-point rule is exact through degree seven. Hence the mass, aerodynamic and load integrals carry no quadrature error. The membrane and thermal integrals also
contain finite-element fields, and they carry the discretization error of those meshes. The measured shapes are not exactly orthogonal in the mass inner
product. With the mass matrix $(\mathbf{M})_{ij} = \rho h \int \phi_i \phi_j \,\mathrm{d}A$, the largest normalized cross term
$|(\mathbf{M})_{ij}|/\sqrt{(\mathbf{M})_{ii} (\mathbf{M})_{jj}}$, over $i \ne j$, is $0.11$. They are therefore replaced by their symmetric L\"owdin~\cite{lowdin1950nonorthogonality}
orthogonalization, $\tilde{\boldsymbol{\phi}} = \mathbf{M}^{-1/2} \boldsymbol{\phi}$. Of all orthonormal sets, this one is nearest to the measured shapes in the mass norm. That is what justifies assigning
the measured frequencies to it. The deflection is $w = \sum_i q_i(t)\, \tilde\phi_i(x,y)$. Unit modal mass, $\rho h \int \tilde\phi_i^2 \,\mathrm{d}A =
1$, gives each $\tilde\phi_i$ units of $\mathrm{kg^{-1/2}}$. Since $w$ is in meters, $\mathbf{q}$ is in
$\mathrm{m\,kg^{1/2}}$.

\paragraph{From the continuum model to the Hamiltonian} The continuum model has the Lagrangian $\mathcal{K}[w] -
U[w]$. Hamilton's principle, with the virtual work of the unsteady pressure added, gives the von K\'arm\'an equations of
the plate. The reduced model applies the same principle to deflections in the span of the modes only, $w = \sum_i q_i
\tilde\phi_i$. This is the Ritz--Galerkin method. Every energy then becomes a function of the modal amplitudes, and the
virtual work becomes $\mathbf{f} \cdot \delta\mathbf{q}$ with $\mathbf{f} = -\int_\Omega p_{\mathrm{pt}}\, \tilde{\boldsymbol{\phi}}
\,\mathrm{d}A$. The Hamiltonian below is therefore the continuum total energy $\mathcal{K} + U$, evaluated at the reduced
deflection. Its rate of change is the continuum energy balance restricted to the modes. 

Three steps carry out the substitution. First, substituting the expansion into the kinetic energy of Equation~\eqref{eq:K-Ub} gives $\mathcal{K} = \tfrac12 \dot{\mathbf{q}}^{\mathsf{T}}
\tilde{\mathbf{M}} \dot{\mathbf{q}}$, with $(\tilde{\mathbf{M}})_{ij} = \rho h \int \tilde\phi_i \tilde\phi_j \,\mathrm{d}A$. The orthogonalized
shapes have unit modal mass: $\tilde{\mathbf{M}} = \mathbf{M}^{-1/2}\mathbf{M}\,\mathbf{M}^{-1/2} = \mathbf{I}$, because
$\mathbf{M}$ is symmetric positive definite, as is its inverse square root. Hence $\mathcal{K} = \tfrac12
|\dot{\mathbf{q}}|^2$.

Second, substituting the expansion into the potential energy, $U = U_b + U_m - \mathcal{W}_s$, gives a
polynomial in $\mathbf{q}$,
\begin{equation}
  U(\mathbf{q}) \;=\;
  \underbrace{\tfrac12\, \mathbf{q}^{\mathsf{T}} \operatorname{diag}(\boldsymbol{\omega}^2)\, \mathbf{q}}_{\text{bending}}
  \;+\; \underbrace{\tfrac12\, \eta_T n_T\, \Delta T\; \mathbf{q}^{\mathsf{T}} \mathbf{G}\, \mathbf{q}}_{\text{thermal preload}}
  \;+\; \underbrace{\tfrac14\, \mathcal{D}_{ijkl}\, q_i q_j q_k q_l}_{\text{membrane hardening}}
  \;\underbrace{-\; \Delta p\; \boldsymbol{\ell} \cdot \mathbf{q} \;-\; \boldsymbol{\ell}_f \cdot \mathbf{q}}_{\text{static load}} .
  \label{eq:U}
\end{equation}
In Equation~\eqref{eq:U}, $\omega_i$ is the circular frequency of mode $i$. The scalar $n_T$ is the stress resultant
per kelvin, and $\mathbf{G}$ the preload operator of Equation~\eqref{eq:thermal-operator}. The tensor $\mathcal{D}$ is the
quartic membrane tensor of Equation~\eqref{eq:quartic}. The vector $\boldsymbol\ell$ is the generalized load of a unit
pressure, and $\boldsymbol\ell_f$ is that of the measured surface pressure's deviation from $p_s$. The bending and load
terms follow directly from the expansion, as the list below shows. The thermal and membrane terms
both come from $U_m$. Write the strain inside Equation~\eqref{eq:membrane-energy} as $\boldsymbol{\varepsilon}(\mathbf{u}) +
\boldsymbol{\eta} - \boldsymbol{\varepsilon}_T$. With the expansion, $\boldsymbol{\eta} = \sum_{ij} q_i q_j
\boldsymbol{\eta}_{ij}$, where $\boldsymbol{\eta}_{ij} = \tfrac12 \operatorname{sym}(\nabla\tilde\phi_i \otimes
\nabla\tilde\phi_j)$ and $\operatorname{sym}\mathbf{X} = \tfrac12(\mathbf{X} + \mathbf{X}^{\mathsf{T}})$. The thermal strain is $\boldsymbol{\varepsilon}_T = \eta_T \Delta T\, \alpha s_T \mathbf{I}$. For fixed
$\mathbf{q}$, the in-plane displacement minimizes $U_m$, whose integrand is quadratic in $\mathbf{u}$. The minimizer
therefore solves a linear plane-stress problem, loaded linearly by $\boldsymbol{\eta}$ and by $\boldsymbol{\varepsilon}_T$.
By superposition, $\mathbf{u} = \sum_{ij} q_i q_j \mathbf{u}_{ij} + \eta_T \Delta T\, \mathbf{u}_T$. Each
$\mathbf{u}_{ij}$ satisfies $\int \boldsymbol{\varepsilon}(\delta\mathbf{u}) : \mathbb{C} : \big(
\boldsymbol{\varepsilon}(\mathbf{u}_{ij}) + \boldsymbol{\eta}_{ij} \big) \,\mathrm{d}A = 0$. Likewise, $\mathbf{u}_T$ satisfies
$\int \boldsymbol{\varepsilon}(\delta\mathbf{u}) : \mathbb{C} : \big( \boldsymbol{\varepsilon}(\mathbf{u}_T) - \alpha s_T
\mathbf{I} \big) \,\mathrm{d}A = 0$. Both hold for every $\delta\mathbf{u}$ that vanishes on $\partial\Omega$, and both
displacements vanish there. With the thermal stress per kelvin
\begin{equation}
  \boldsymbol{\sigma}_T \;=\; \mathbb{C} : \big( \boldsymbol{\varepsilon}(\mathbf{u}_T) - \alpha\, s_T\, \mathbf{I} \big) ,
  \label{eq:sigmaT}
\end{equation}
the total strain becomes $\mathbf{e}_q + \eta_T \Delta T\, \mathbb{C}^{-1} : \boldsymbol{\sigma}_T$. Here $\mathbf{e}_q =
\sum_{ij} q_i q_j \big( \boldsymbol{\varepsilon}(\mathbf{u}_{ij}) + \boldsymbol{\eta}_{ij} \big)$. Expanding the quadratic
form, and using the symmetry of $\mathbb{C}$ for the cross term, $U_m$ splits into three parts,
\begin{equation}
  U_m \;=\; \underbrace{\tfrac12 \int \mathbf{e}_q : \mathbb{C} : \mathbf{e}_q \,\mathrm{d}A}_{\text{quartic in } \mathbf{q}}
  \;+\; \underbrace{\eta_T \Delta T \int \mathbf{e}_q : \boldsymbol{\sigma}_T \,\mathrm{d}A}_{\text{quadratic in } \mathbf{q}}
  \;+\; \underbrace{\tfrac12 (\eta_T \Delta T)^2 \int \boldsymbol{\sigma}_T : \mathbb{C}^{-1} : \boldsymbol{\sigma}_T
  \,\mathrm{d}A}_{\text{constant}} .
  \label{eq:Um-split}
\end{equation}
The constant does not affect the dynamics and is dropped. The middle term of Equation~\eqref{eq:Um-split} simplifies because $\boldsymbol{\sigma}_T$ is a self-equilibrated
in-plane stress. By the definition of $\mathbf{u}_T$, it satisfies $\int \boldsymbol{\varepsilon}(\delta\mathbf{u}) :
\boldsymbol{\sigma}_T \,\mathrm{d}A = 0$ for every in-plane displacement $\delta\mathbf{u}$ that vanishes on the edges $\partial\Omega$. Each $\mathbf{u}_{ij}$ vanishes
there. Hence the work of $\boldsymbol{\sigma}_T$ on $\boldsymbol{\varepsilon}(\mathbf{u}_{ij})$
is zero, and only $\boldsymbol{\eta}_{ij}$ remains. Because the stress $\boldsymbol{\sigma}_T$ is symmetric, its contraction with $\operatorname{sym}\mathbf{X}$ equals its
contraction with $\mathbf{X}$. Hence $\boldsymbol{\eta}_{ij} : \boldsymbol{\sigma}_T = \tfrac12 \boldsymbol{\sigma}_T :
(\nabla\tilde\phi_i \otimes \nabla\tilde\phi_j)$. The middle term is then $\tfrac12 \eta_T \Delta T \sum_{ij} q_i q_j \int
\boldsymbol{\sigma}_T : (\nabla\tilde\phi_i \otimes \nabla\tilde\phi_j) \,\mathrm{d}A$. This is the thermal preload of
Equation~\eqref{eq:U}, with
\begin{equation}
  G_{ij} \;=\; \frac{1}{n_T} \int \boldsymbol{\sigma}_T : \big( \nabla\tilde\phi_i \otimes \nabla\tilde\phi_j \big)
  \,\mathrm{d}A .
  \label{eq:thermal-operator}
\end{equation}
Here $n_T = -E h \alpha / (1 - \nu)$ is the in-plane stress resultant per kelvin of uniform heating. For uniform
heating, $s_T = 1$, and $\mathbf{u}_T = 0$ satisfies its defining equation. The trial stress $\mathbb{C} : (-\alpha
\mathbf{I})$ is then constant, and a constant stress does no work on a displacement that vanishes on $\partial\Omega$.
The plane-stress stiffness gives $\mathbb{C} : \mathbf{I} = E h\, (1 + \nu)/(1 - \nu^2)\, \mathbf{I} = E h/(1 - \nu)\,
\mathbf{I}$, which gives $\boldsymbol{\sigma}_T = n_T \mathbf{I}$.
$\mathbf{G}$ is thus the preload operator per kelvin of $\eta_T \Delta T$, in units of $n_T$. The first term of
Equation~\eqref{eq:Um-split} is the membrane hardening of Equation~\eqref{eq:U}, with
\begin{equation}
  \mathcal{D}_{ijkl} \;=\; 2\, \operatorname{Sym} \int \big( \boldsymbol{\varepsilon}(\mathbf{u}_{ij}) +
  \boldsymbol{\eta}_{ij} \big) : \mathbb{C} : \big( \boldsymbol{\varepsilon}(\mathbf{u}_{kl}) + \boldsymbol{\eta}_{kl} \big)
  \,\mathrm{d}A ,
  \label{eq:quartic}
\end{equation}
where $\operatorname{Sym}$ averages over all twenty-four permutations of $(i,j,k,l)$. The factor and the average
follow from the first term of Equation~\eqref{eq:Um-split}. Call its integral $X_{ijkl}$. The term is then $\tfrac12
X_{ijkl}\, q_i q_j q_k q_l = \tfrac14 (2 X_{ijkl})\, q_i q_j q_k q_l$. The product $q_i q_j q_k q_l$ is unchanged by any
permutation of the indices. Replacing $2X$ by its average over permutations therefore leaves the sum unchanged.

Third, the Lagrangian $\mathcal{L}_g = \mathcal{K} - U$ defines the generalized momenta $\mathbf{p} = \partial
\mathcal{L}_g / \partial \dot{\mathbf{q}}$. The modal mass is unity, which gives $\mathbf{p} = \dot{\mathbf{q}}$. The Legendre transform
$H = \mathbf{p} \cdot \dot{\mathbf{q}} - \mathcal{L}_g$ then gives the total energy, $H = \tfrac12 |\mathbf{p}|^2 +
U(\mathbf{q})$. With Equation~\eqref{eq:U}, and the two quadratic terms combined,
\begin{equation}
\begin{split}
  H(\mathbf{q}, \mathbf{p}; \Delta T, \Delta p) \;=\; {} & \tfrac12 |\mathbf{p}|^2
  \;+\; \tfrac12\, \mathbf{q}^{\mathsf{T}} \underbrace{\big[ \operatorname{diag}(\boldsymbol{\omega}^2)
        + \eta_T n_T\, \Delta T\, \mathbf{G} \big]}_{\mathbf{K}(\Delta T)} \mathbf{q} \\
  & \;+\; \tfrac14\, \mathcal{D}_{ijkl}\, q_i q_j q_k q_l
  \;-\; \Delta p\; \boldsymbol{\ell} \cdot \mathbf{q}
  \;-\; \boldsymbol{\ell}_f \cdot \mathbf{q} .
\end{split}
  \label{eq:reference model-H}
\end{equation}
The matrix $\mathbf{K}(\Delta T)$ is the linear stiffness of the heated flat panel. For the measured temperature fields, the computed $\mathbf{G}$ is positive definite. Since $\eta_T > 0$ and
$n_T < 0$, heating then lowers every quadratic form $\mathbf{q}^{\mathsf{T}} \mathbf{K}(\Delta T)\, \mathbf{q}$. The
flat panel buckles at the smallest $\Delta T$ at which $\mathbf{K}(\Delta T)$ loses positive definiteness.
The quartic term of Equation~\eqref{eq:reference model-H} contributes a cubic force to Hamilton's
equations. That force is a true gradient only because $\mathcal{D}$ is symmetric.

\begin{lemma}[The cubic force is a gradient]
\label{lem:gradient}
If $\mathcal{D}$ is invariant under every permutation of its indices, then $\partial_{q_i} \big( \tfrac14
\mathcal{D}_{jklm} q_j q_k q_l q_m \big) = \mathcal{D}_{ijkl}\, q_j q_k q_l$.
\end{lemma}

\begin{proof}
By the product rule, $\partial_{q_i} (q_j q_k q_l q_m) = \delta_{ij} q_k q_l q_m + \delta_{ik} q_j q_l q_m + \delta_{il}
q_j q_k q_m + \delta_{im} q_j q_k q_l$, with $\delta$ the Kronecker delta. Contracting with $\tfrac14 \mathcal{D}_{jklm}$
gives four sums. The second, for example, is $\tfrac14 \mathcal{D}_{jilm}\, q_j q_l q_m$. Full symmetry lets $i$ move to
the first index, $\mathcal{D}_{jilm} = \mathcal{D}_{ijlm}$. Renaming the summed indices to $j, k, l$ then gives
$\tfrac14 \mathcal{D}_{ijkl}\, q_j q_k q_l$. The other three sums give the same, and the four add to $\mathcal{D}_{ijkl}\,
q_j q_k q_l$.
\end{proof}

Without full symmetry, the cubic force of a tensor assembled by quadrature is not the gradient of any energy. The
energy identities below would then fail by the size of the asymmetry. This is why Equation~\eqref{eq:quartic}
symmetrizes, and why the reference and learned models both use the symmetrized tensor.

Four remarks complete the terms of Equation~\eqref{eq:U}.

\begin{itemize}
  \item \emph{Bending.} Substituting the expansion into $U_b$ of Equation~\eqref{eq:K-Ub} gives $\tfrac12
  \mathbf{q}^{\mathsf{T}} \mathbf{K}_b \mathbf{q}$, with a bending stiffness matrix $\mathbf{K}_b$. Write $U_b[w] = \tfrac12 a_b(w, w)$, with
  $a_b$ the symmetric bilinear form of Equation~\eqref{eq:K-Ub}. Then $(\mathbf{K}_b)_{ij} = a_b(\tilde\phi_i,
  \tilde\phi_j)$. A free-vibration mode $\tilde\phi_j$ of frequency $\omega_j$ satisfies $a_b(\tilde\phi_i, \tilde\phi_j)
  = \omega_j^2\, \rho h \int \tilde\phi_i \tilde\phi_j \,\mathrm{d}A$ for every $i$. Unit modal mass then gives
  $(\mathbf{K}_b)_{ij} = \omega_j^2 \delta_{ij}$. Taking the measured shapes as those modes gives the bending term of
  Equation~\eqref{eq:U}. Here $\omega_i = 2\pi f_i$, with the measured frequencies $f_i$
  from \SI{242}{\hertz} to \SI{771}{\hertz}. The package measured these frequencies with free edges, and all but the
  first are carried as measured. The bending energy $U_b$ is therefore represented by the measurements, and
  the flexural rigidity never has to be assumed. The one exception is $\omega_1$. It is replaced by the installed
  fundamental of \SI{263.67}{\hertz} given in the package~\cite{aepw4hswg}. The package estimates it from the strain
  spectrum of the panel during tunnel startup.
  \item \emph{Thermal preload.} For a uniform field $\mathbf{u}_T = 0$, $\boldsymbol{\sigma}_T = n_T \mathbf{I}$, and $G_{ij} = \int \nabla\tilde\phi_i \cdot
  \nabla\tilde\phi_j \,\mathrm{d}A$. Against a uniform field of the same mean, a field cooler near the frame carries less compression where the shapes' slopes are largest. The panel then buckles at a higher mean temperature rise. For the measured field, $\mathbf{u}_T$ is computed with the
same quadratic triangular elements as below, on a finer $48 \times 24$ structured mesh.
  \item \emph{Membrane hardening.} Each $\mathbf{u}_{ij}$ solves the plane-stress problem with $\mathbf{u}_{ij} = 0$ on $\partial\Omega$, $\int
  \boldsymbol{\varepsilon}(\delta\mathbf{u}) : \mathbb{C} : \big( \boldsymbol{\varepsilon}(\mathbf{u}_{ij}) + \boldsymbol{\eta}_{ij}
  \big) \,\mathrm{d}A = 0$ for every admissible $\delta\mathbf{u}$. The plane-stress problems, one per
  pair of modes, are solved with quadratic triangular elements on a $24 \times 12$ structured mesh of the panel. 
  \item \emph{Static load.} The pressure work of Equation~\eqref{eq:Ws} is $\mathcal{W}_s = \Delta p \sum_i q_i \int \tilde\phi_i
  \,\mathrm{d}A$. Its negative, $-\mathcal{W}_s$, is the load term of Equation~\eqref{eq:U}. Here $\Delta p = p_c - p_s$ is the cavity pressure less the RANS surface pressure $p_s =
  \SI{51.0}{\kilo\pascal}$ of~\cite{brouwer2021}, and $\boldsymbol{\ell} = \int \tilde{\boldsymbol\phi}\,\mathrm{d}A$.
  Referencing the static load to the surface pressure, not to the freestream, produces the asymmetry of the response
about that pressure (Equation 25 of~\cite{brouwer2021}). $\boldsymbol\ell_f$ is the generalized load of
  the measured surface pressure's deviation from $p_s$. It is zero in every simulated record of this report (Section~\ref{sec:methods-record}). The whole static load therefore enters through $\Delta p$.
\end{itemize}

The package gives no material constants. We take $E = \SI{205}{\giga\pascal}$, $\nu = 0.29$ and $\rho = \SI{7850}{\kilo\gram\per
\cubic\meter}$. These are the values for annealed AISI 4140 in~\cite{matweb4140}. The thermal expansion coefficient is
$\alpha = \SI{1.23e-5}{\per\kelvin}$, against the \SI{1.22e-5}{\per\kelvin} that~\cite{matweb4140} gives from
\SI{0}{} to \SI{100}{\celsius}. The measured frequencies fix the flexural stiffness whatever $E$ is taken to be. $E$
enters only the membrane and thermal terms, and $\alpha$ only the thermal one. There $\alpha$ appears only in the product
$\alpha \eta_T$. The thermal effectiveness $\eta_T$ is the fraction of the thermocouple temperature difference that
strains the panel in its plane. It is calibrated below on the measured buckling temperature. The choice between the
two values of $\alpha$ therefore changes $\eta_T$ by \SI{0.8}{\percent} and leaves the model unchanged.

\paragraph{The temperature field} The shape $s_T$ is derived from the package's measurements in two steps. First, the
adiabatic wall temperature that the boundary layer drives the surface toward follows from the Kulite pressures measured on
the rigid panel. The local Mach number $M_\ell$ follows from the static pressure $p_K$ and the rigid run's stagnation pressure
$p_{0,\mathrm{r}}$ through the isentropic relation $p_K/p_{0,\mathrm{r}} = \big(1 + \tfrac{\gamma-1}{2}
M_\ell^2\big)^{-\gamma/(\gamma-1)}$~\cite{ames1953naca1135}. The recovery relation~\cite{white2006viscous} then gives
\begin{equation}
  T_{aw} \;=\; T_0\, \frac{1 + r_c\, \tfrac{\gamma - 1}{2} M_\ell^2}{1 + \tfrac{\gamma - 1}{2} M_\ell^2},
  \label{eq:recovery}
\end{equation}
with the heated run's stagnation temperature $T_0$ and $\gamma = 1.4$. The turbulent recovery factor is
$r_c = \mathrm{Pr}^{1/3} = 0.89$, with Prandtl number $\mathrm{Pr} = 0.71$ for air~\cite{white2006viscous}. A
quadratic surface in $x$ and $y$, with six coefficients, is fitted by least squares to the fifteen Kulite values. It
carries $T_{aw}$ over the panel. Second, the panel is a thin fin~\cite{incropera2007heat} heated by the boundary layer and cooled by conduction to the
frame. Its excess
temperature over the frame, $\vartheta$, satisfies
\begin{equation}
  L_f^2\, \nabla^2 \vartheta \;=\; \vartheta - \big( T_{aw} - T_{\mathrm{frame}} \big) \quad \text{in } \Omega, \qquad
  \vartheta = 0 \quad \text{on } \partial\Omega,
  \label{eq:fin}
\end{equation}
with one parameter, the fin length $L_f = \sqrt{k h / h_c}$ of a plate of conductivity $k$ and thickness $h$ under a
film coefficient $h_c$. For the panel, $h = \SI{0.635}{\milli\meter}$ and $k = \SI{42.6}{\watt\per\meter\per\kelvin}$. This is
the value for annealed AISI 4140 at \SI{100}{\celsius}~\cite{matweb4140}.
$L_f$ is set by the package's panel and frame thermocouples. Their measured difference $\Delta T$, divided by the mean driving difference $\bar T_{aw} - T_{\mathrm{frame}}$, gives a ratio. $L_f$ is the fin length whose panel-mean excess $\bar\vartheta$ reproduces that ratio. On the chaotic case that ratio is $0.464$ and $L_f =
\SI{25.7}{\milli\meter}$. The panel thermocouple sits in the steep temperature gradient near the frame~\cite{brouwer2021}. The measured ratio therefore understates
the panel's mean excess, and this $L_f$ cools the panel too much. The calibrated $L_f$ is therefore an upper bound on the fin length. It
lies above the range of $3.7$ to \SI{11.6}{\milli\meter} that film coefficients of $2000$ down to
\SI{200}{\watt\per\square\meter\per\kelvin} give. That range brackets an estimate at the test condition. The turbulent
flat-plate correlation with the Colburn analogy~\cite{incropera2007heat} gives $h_c \approx
\SI{570}{\watt\per\square\meter\per\kelvin}$ with freestream properties. Taking the properties at Eckert's reference
temperature~\cite{eckert1956} gives \SI{460}{\watt\per\square\meter\per\kelvin}. Both use the running length that grows
the package's \SI{8.6}{\milli\meter} boundary layer, and they give $L_f$ of \SI{6.9}{} and \SI{7.7}{\milli\meter}. The shape is $s_T = \vartheta / \bar\vartheta$. The
thermal effectiveness $\eta_T = 0.468$ is then the one value that places the flat panel's buckling temperature rise at the
\SI{5.6}{\kelvin} reported in~\cite{brouwer2021}.

\paragraph{Adding the non-conservative forces} Hamilton's equations, $\dot{\mathbf{q}} = \partial H / \partial \mathbf{p}
= \mathbf{p}$ and $\dot{\mathbf{p}} = -\partial H / \partial \mathbf{q}$, give the conservative motion. Damping and the
aerodynamic pressure are not derived from a potential. By the Lagrange--d'Alembert principle, they enter through their
virtual work. In modal coordinates that work is $(-\mathbf{C}\mathbf{p} + \mathbf{f}) \cdot \delta\mathbf{q}$. They therefore appear
as generalized forces on the momentum equation alone. With $\mathbf{x} = (\mathbf{q}, \mathbf{p})$, the reference model is
\begin{equation}
  \dot{\mathbf{x}} \;=\; \big( \mathbf{J} - \mathbf{R} \big) \nabla_{\mathbf{x}} H + \begin{bmatrix} \mathbf{0} \\ \mathbf{f}
  \end{bmatrix}, \qquad
  \mathbf{J} = \begin{bmatrix} \mathbf{0} & \mathbf{I} \\ -\mathbf{I} & \mathbf{0} \end{bmatrix}, \qquad
  \mathbf{R} = \begin{bmatrix} \mathbf{0} & \mathbf{0} \\ \mathbf{0} & \mathbf{C} \end{bmatrix} .
  \label{eq:phs}
\end{equation}
This is the port-Hamiltonian form~\cite{vanderSchaftJeltsema2014}. By the chain rule, $\mathrm{d}H/\mathrm{d}t =
\nabla H^{\mathsf{T}} \dot{\mathbf{x}} = \nabla H^{\mathsf{T}} (\mathbf{J} - \mathbf{R}) \nabla H + \nabla_{\mathbf{p}} H
\cdot \mathbf{f}$. Since $\mathbf{J}$ is skew-symmetric, $\nabla H^{\mathsf{T}} \mathbf{J} \nabla H = 0$. With
$\nabla_{\mathbf{p}} H = \mathbf{p}$, the remaining terms give
\begin{equation}
  \frac{\mathrm{d}H}{\mathrm{d}t} \;=\; -\,\mathbf{p}^{\mathsf{T}} \mathbf{C}\, \mathbf{p} \;+\; \mathbf{p} \cdot \mathbf{f} .
  \label{eq:reference-balance}
\end{equation}
Damping keeps this structure when $\mathbf{C}$ is symmetric positive semidefinite. Then $-\mathbf{p}^{\mathsf{T}}
\mathbf{C}\mathbf{p} \le 0$, and damping can remove energy but never add it. Here $\mathbf{C} = \alpha_d \mathbf{I} + \beta_d
\operatorname{diag}(\boldsymbol\omega^2)$ is the Rayleigh damping of~\cite{brouwer2021}, with $\alpha_d =
\SI{15}{\per\second}$ and $\beta_d = \SI{8e-7}{\second}$. It is diagonal with positive entries and therefore meets the condition. The
learned dissipation of Section~\ref{sec:methods-dissipation} meets the same condition by construction. The term
$\mathbf{p} \cdot \mathbf{f}$ is the power of the generalized aerodynamic force of Section~\ref{sec:methods-aero}. It can
take either sign. The records of this report carry no turbulent boundary-layer pressure forcing.

\paragraph{Nondimensional state} The record and the learned model use $\Wvec = \mathbf{q}/w_s$ and $\tau =
\omega_1 t$, with $\Wdot = \mathrm{d}\Wvec/\mathrm{d}\tau$. The scale is $w_s = h / \max_\Omega |\tilde\phi_1|$. With this scale, $W_1 = 1$ places the first mode's peak
deflection at one panel thickness. In these units the post-flutter state is of order one in
every component. With $\mathbf{q} = w_s \Wvec$ and $\mathbf{p} = w_s \omega_1 \Wdot$, the kinetic energy is $w_s^2 \omega_1^2\,
\tfrac12 |\Wdot|^2$. The nondimensional energy is therefore $H/(w_s^2 \omega_1^2)$. From Section~\ref{sec:methods-identity} on, $H$ denotes this
nondimensional energy and $\mathbf{x}$ the nondimensional state $(\Wvec, \Wdot)$. Before it, $\mathbf{x}$ is $(\mathbf{q},
\mathbf{p})$ and $H$ is dimensional.

\subsection{The generalized aerodynamic force}
\label{sec:methods-aero}

First-order piston theory~\cite{ashley1956piston}, Equations 1 to 5 of~\cite{brouwer2021}, gives the quasi-steady
pressure on the flow side of the panel,
\begin{equation}
  p_{\mathrm{pt}}(x,y,t) \;=\; \frac{2 q_\infty}{\beta} \left( \frac{\partial w}{\partial x}
  + \frac{1}{V_\infty} \frac{\partial w}{\partial t} \right), \qquad \beta = \sqrt{M^2 - 1}, \quad
  q_\infty = \tfrac12 \gamma p_\infty M^2 ,
  \label{eq:piston-pressure}
\end{equation}
with $V_\infty$ the freestream speed. Projected on the unit-mass modes it gives a generalized force $\mathbf{f} =
-(\mathbf{A}\mathbf{q} + \mathbf{B}\mathbf{p})$, with
\begin{equation}
  A_{ij} = \frac{2 q_\infty}{\beta} \int \tilde\phi_i\, \tilde\phi_{j,x} \,\mathrm{d}A, \qquad
  B_{ij} = \frac{2 q_\infty}{\beta V_\infty} \int \tilde\phi_i\, \tilde\phi_j \,\mathrm{d}A .
  \label{eq:piston}
\end{equation}
These entries follow from substituting $w = \sum_j q_j \tilde\phi_j$ into $\mathbf{f} = -\int p_{\mathrm{pt}}\,
\tilde{\boldsymbol\phi} \,\mathrm{d}A$, with $\partial w/\partial t = \sum_j p_j \tilde\phi_j$ since $\mathbf{p} =
\dot{\mathbf{q}}$. Unit modal mass with uniform $\rho h$ gives $\int \tilde\phi_i \tilde\phi_j \,\mathrm{d}A =
\delta_{ij}/(\rho h)$. Hence $\mathbf{B} = 2 q_\infty / (\beta V_\infty \rho h)\, \mathbf{I}$, which is symmetric positive
definite, and its power $-\mathbf{p}^{\mathsf{T}} \mathbf{B}\mathbf{p}$ is never positive. Integrating $A_{ij}$ by parts in
$x$ gives $-A_{ji}$ plus a boundary integral of $\tilde\phi_i \tilde\phi_j$. The package supplies the shapes with zero edge rows and
columns, hence they vanish on the clamped edges. That integral is therefore zero, and $\mathbf{A}$ is skew-symmetric. Computed, $\|\mathbf{A} + \mathbf{A}^{\mathsf{T}}\| /
\|\mathbf{A}\|$ is about $10^{-15}$. The power of $\mathbf{A}$, $-\mathbf{p} \cdot \mathbf{A}\mathbf{q}$, can take
either sign. It is the only term in the model that can feed energy from the flow into the panel. Flutter enters there
and nowhere else. The
flow state comes from each case's heated-run stagnation pressure $p_0$ and temperature $T_0$ in the package. The Mach
number is the test-section value $M = 1.94$ of~\cite{brouwer2021}, which they support with boundary-layer measurements.
The package's RANS simulations give $1.92 \pm 0.02$ at the panel (Section~\ref{sec:problems-rc19}). With $f_M = 1 + \tfrac{\gamma-1}{2}M^2$ and $\gamma =
1.4$, the isentropic relations~\cite{ames1953naca1135} give $p_\infty = p_0 f_M^{-\gamma/(\gamma-1)}$ and $T_\infty =
T_0/f_M$. The dynamic pressure is $q_\infty = \tfrac12\gamma p_\infty M^2$, and the speed is $V_\infty = M\sqrt{\gamma
R_{\mathrm{air}} T_\infty}$, with $R_{\mathrm{air}} = \SI{287}{\joule\per\kilo\gram\per\kelvin}$. For the periodic case,
$p_0 = \SI{347.5}{\kilo\pascal}$ and $T_0 = \SI{388.3}{\kelvin}$ give $p_\infty = \SI{48.8}{\kilo\pascal}$, $q_\infty =
\SI{128.4}{\kilo\pascal}$ and $V_\infty = \SI{579}{\meter\per\second}$. For the chaotic case, $p_0 =
\SI{346.5}{\kilo\pascal}$ and $T_0 = \SI{387.2}{\kelvin}$ give \SI{48.6}{\kilo\pascal}, \SI{128.1}{\kilo\pascal} and
\SI{578}{\meter\per\second}.
The power of this force, $-\mathbf{p}\cdot(\mathbf{A}\mathbf{q} + \mathbf{B}\mathbf{p})$, can take either sign, and
nothing in the construction bounds it. That is intended: this force is where the flow does work on the
panel, and the construction reports that work rather than constraining it.

\subsection{The HyMAX reference model}
\label{sec:methods-hymax}

The HyMAX plate is $130 \times 80 \times 2$~mm of 6061-T6, bolted to the forebody at its hinge line and free at its
trailing edge (Figure~\ref{fig:hymax-setup}). The model idealizes the bolted joint as a clamp. The plate's measured response is linear and dominated by the first mode~\cite{currao2019}. The ratio of the two computed
frequencies of~\cite{xu2025}, $530.4/84.7 = 6.26$, is close to the clamped-free beam's own. That ratio is
$(\kappa_2/\kappa_1)^2 = (4.694/1.875)^2 = 6.267$, with $\kappa_i$ the eigenvalue roots of Equation~\eqref{eq:cf-modes}
below. The reference model is therefore an Euler--Bernoulli beam of the plate's width $b$ and length $L$.

\paragraph{The continuum model} Let $x$ run from the hinge, $w(x,t)$ be the deflection, positive toward the shock side,
$A_c = b h$ and $I = b h^3/12$. The thickness is $h = \SI{2}{\milli\meter}$, and the 6061-T6 density is $\rho =
\SI{2700}{\kilo\gram\per\cubic\meter}$, from the data sheet the package names~\cite{matweb6061t6}. On the beam, a prime denotes differentiation in $x$. The beam's kinetic and
bending energies are
\begin{equation}
  \mathcal{K}[w] \;=\; \frac12 \int_0^L \rho A_c\, \dot w^2 \,\mathrm{d}x, \qquad
  U_b[w] \;=\; \frac12 \int_0^L EI\, (w'')^2 \,\mathrm{d}x .
  \label{eq:beam-energies}
\end{equation}
The flow-side pressure $p_{\mathrm{top}}$ pushes the plate toward $-w$. The back pressure $p_b$, which is not measured, pushes
it back. Their virtual work is $\delta\mathcal{W} = -\,b \int_0^L \big[ p_{\mathrm{top}}(x,t) - p_b \big]\, \delta w \,\mathrm{d}x$.
Hamilton's principle with this work, and with the clamped-free conditions, gives the beam equation
\begin{equation}
  \rho A_c\, \ddot w + E I\, w'''' \;=\; -\,b \big[ p_{\mathrm{top}}(x,t) - p_b \big], \qquad
  w(0) = w'(0) = 0, \quad w''(L) = w'''(L) = 0 .
  \label{eq:beam}
\end{equation}

\paragraph{The basis} The plate is expanded in its first two clamped-free modes~\cite{blevins1979formulas}, and
the number of retained modes is $n = 2$. Section~\ref{sec:hymax} shows that a third mode does not reduce the error.
The mode shapes are
\begin{equation}
  \psi_i(\xi) \;=\; \cosh \kappa_i \xi - \cos \kappa_i \xi - \chi_i \big( \sinh \kappa_i \xi - \sin \kappa_i \xi \big),
  \qquad \chi_i = \frac{\sinh \kappa_i - \sin \kappa_i}{\cosh \kappa_i + \cos \kappa_i},
  \label{eq:cf-modes}
\end{equation}
with $\xi = x/L$ and $\cos\kappa_i \cosh\kappa_i + 1 = 0$, which gives $\kappa_1 = 1.8751$ and $\kappa_2 = 4.6941$. Each mode solves the
beam's own eigenproblem, $EI\, \psi_i'''' = \omega_i^2\, \rho A_c\, \psi_i$, with the clamped-free conditions and $\omega_i^2 =
(\kappa_i / L)^4\, E I / (\rho A_c)$. The functions of Equation~\eqref{eq:cf-modes} satisfy $\int_0^1 \psi_i^2 \,\mathrm{d}\xi = 1$. Multiplying each by
$(\rho A_c L)^{-1/2}$ scales it to unit modal mass, $\int_0^L \rho A_c\, \psi_i \psi_j \,\mathrm{d}x = \delta_{ij}$. From
here on, $\psi_i$ denotes the scaled shape, and $\boldsymbol\psi = (\psi_1, \psi_2)$. $EI$ is scaled so that $f_1$ is the
measured \SI{83}{\hertz}~\cite{currao2019}. Since $\omega_i \propto \kappa_i^2$, $f_2 = f_1 (\kappa_2/\kappa_1)^2 =
\SI{83}{\hertz} \times 6.267 = \SI{520}{\hertz}$.

\paragraph{From the continuum model to the Hamiltonian} The reduction follows the three steps of
Section~\ref{sec:methods-panel}, with $w = \sum_i q_i(t)\, \psi_i(x)$. First, unit modal mass turns the kinetic energy of
Equation~\eqref{eq:beam-energies} into $\mathcal{K} = \tfrac12 |\dot{\mathbf{q}}|^2$.

Second, the bending energy becomes $U_b = \tfrac12 \mathbf{q}^{\mathsf{T}} \mathbf{K}_b \mathbf{q}$, with $(\mathbf{K}_b)_{ij} =
\int_0^L EI\, \psi_i'' \psi_j'' \,\mathrm{d}x$. Integrating by parts twice moves two derivatives from $\psi_i$ onto $\psi_j$.
The boundary terms vanish, because $\psi_i$ and $\psi_i'$ are zero at the clamp and $\psi_j''$ and $\psi_j'''$ at the free end.
The eigenproblem then gives
\begin{equation}
  (\mathbf{K}_b)_{ij} \;=\; \int_0^L \psi_i \, EI\, \psi_j'''' \,\mathrm{d}x
  \;=\; \omega_j^2 \int_0^L \rho A_c\, \psi_i \psi_j \,\mathrm{d}x \;=\; \omega_j^2\, \delta_{ij} ,
  \label{eq:beam-Kb}
\end{equation}
which gives $U_b = \tfrac12 \mathbf{q}^{\mathsf{T}} \operatorname{diag}(\boldsymbol{\omega}^2)\, \mathbf{q}$. The plate's measured
response is linear. The beam therefore carries no membrane or thermal term, and its potential is purely quadratic.

Third, the Lagrangian $\mathcal{K} - U_b$ defines the generalized momenta $\mathbf{p} = \dot{\mathbf{q}}$, and the Legendre
transform gives
\begin{equation}
  H(\mathbf{q}, \mathbf{p}) \;=\; \tfrac12 |\mathbf{p}|^2 + \tfrac12\, \mathbf{q}^{\mathsf{T}}
  \operatorname{diag}(\boldsymbol{\omega}^2)\, \mathbf{q} .
  \label{eq:beam-H}
\end{equation}
The virtual work of the pressures becomes $\mathbf{f} \cdot \delta\mathbf{q}$, with
\begin{equation}
  \mathbf{f}(t) \;=\; -\,b \int_0^L \big[ p_{\mathrm{top}}(x,t) - p_b \big] \boldsymbol{\psi} \,\mathrm{d}x .
  \label{eq:beam-modal}
\end{equation}
Unlike the RC-19 static load, the measured static pressure here varies in time. It therefore cannot enter $U$ as a
potential, and it enters $\mathbf{f}$ as a prescribed drive.

\paragraph{The port-Hamiltonian form} With modal damping $\mathbf{C} = \operatorname{diag}(2\zeta\omega_i)$, the HyMAX
reference model takes the form of Equation~\eqref{eq:phs}, with the $H$ of Equation~\eqref{eq:beam-H}. Its rate law is
$\dot{\mathbf{q}} = \mathbf{p}$ and $\dot{\mathbf{p}} = -\operatorname{diag}(\boldsymbol{\omega}^2)\, \mathbf{q} -
\mathbf{C}\mathbf{p} + \mathbf{f}$. The energy balance is that of Equation~\eqref{eq:reference-balance}. $\mathbf{C}$ is
diagonal with positive entries, and damping can therefore only remove energy. The damping ratio is $\zeta = 0.01$. The package holds no
displacement. This value is therefore a setting rather than a measurement. The paragraph below splits $\mathbf{f}$ into the
prescribed drive and a state-dependent aerodynamic force.

\paragraph{The aerodynamic force} The flow-side pressure in Equation~\eqref{eq:beam-modal} has two parts,
$p_{\mathrm{top}} = p_{\mathrm{st}} + p_{\mathrm{pt}}$. The static part $p_{\mathrm{st}}$ is set by the measured pressures. The part $p_{\mathrm{pt}}$
follows the plate's motion through piston theory, Equation~\eqref{eq:piston-pressure}, with local coefficients. The
freestream pressure $p_\infty(t)$ at Mach $5.85$ is not stationary over the first
\SI{90}{\milli\second} of the Ludwieg-tube flow (Figure~\ref{fig:hymax-setup}b). It enters first-order piston theory as a
time-varying inflow. Downstream of the shock position $x_s$, the flow state is the one behind an oblique shock~\cite{ames1953naca1135} of the measured pressure
ratio. That ratio is $\Pi(t) = p_4(t)/p_\infty(t)$, with $p_4$ the trailing-edge transducer. The ratio fixes the shock's
normal Mach number $M_n$, its angle $\Theta_s$ on the weak branch, and the flow deflection $\Theta_w$,
\begin{equation}
  M_n^2 = 1 + \frac{\gamma+1}{2\gamma} (\Pi - 1), \qquad \sin\Theta_s = \frac{M_n}{M}, \qquad
  \tan\Theta_w = \frac{2 \cot\Theta_s\, (M^2 \sin^2\Theta_s - 1)}{M^2 (\gamma + \cos 2\Theta_s) + 2} ,
  \label{eq:oblique}
\end{equation}
with $M = 5.85$ and $\gamma = 1.4$. The post-shock Mach number and temperature are
\begin{equation}
  M_\ell = \frac{1}{\sin(\Theta_s - \Theta_w)} \sqrt{\frac{1 + \tfrac{\gamma-1}{2} M_n^2}{\gamma M_n^2 -
  \tfrac{\gamma-1}{2}}}, \qquad
  \frac{T_\ell}{T_\infty} = \Pi\, \frac{2 + (\gamma-1) M_n^2}{(\gamma+1) M_n^2} .
  \label{eq:post-shock}
\end{equation}
Here $T_\infty$ is the package's measured freestream temperature, averaged over the flow window. The post-shock pressure is $p_\ell = \Pi p_\infty$. The local dynamic pressure is $q_\ell = \tfrac12 \gamma p_\ell M_\ell^2$,
and the local speed is $V_\ell = M_\ell \sqrt{\gamma R_{\mathrm{air}} T_\ell}$. With the
local coefficients
\begin{equation}
  a(x,t) = \frac{2 q_\ell(x,t)}{\beta_\ell(x,t)}, \qquad c(x,t) = \frac{2 q_\ell(x,t)}{\beta_\ell(x,t)\, V_\ell(x,t)},
  \label{eq:local-piston}
\end{equation}
where $q_\ell$, $\beta_\ell = \sqrt{M_\ell^2 - 1}$ and $V_\ell$ are the local dynamic pressure, supersonic factor and
flow speed. They take the freestream state upstream of $x_s$ and the post-shock state downstream of it. Then the generalized force is
\begin{equation}
\begin{gathered}
  \mathbf{f}(t) \;=\; \mathbf{g}(t) - \mathbf{A}(t)\mathbf{q} - \mathbf{B}(t)\mathbf{p}, \qquad
  A_{ij}(t) = b \int_0^L a\, \psi_i \psi_j' \,\mathrm{d}x, \qquad
  B_{ij}(t) = b \int_0^L c\, \psi_i \psi_j \,\mathrm{d}x, \\
  g_i(t) = -\,b \int_0^L \big[ p_{\mathrm{st}}(x,t) - p_b \big] \psi_i \,\mathrm{d}x, \qquad
  p_{\mathrm{st}}(x,t) = \begin{cases} p_\infty(t), & x \le x_s, \\ \Pi(t)\, p_\infty(t), & x > x_s. \end{cases}
\end{gathered}
  \label{eq:hymax-port}
\end{equation}
Here the $\psi_i$ are the unit-mass shapes. The static pressure $p_{\mathrm{st}}$ jumps at $x_s$. It is the footprint of the impinging shock on the
plate, anchored to the level the trailing-edge transducer measures. The measured signals are resampled at the record's
time samples, which are the instants at which the record stores the state (Section~\ref{sec:methods-record}). Hence $\mathbf{g}$, $\mathbf{A}$ and $\mathbf{B}$ are known at every time sample and are interpolated linearly between
them. The two unmeasured quantities, $p_b$ and $x_s$, are not fitted. The record's trajectories sample three values of $p_b$ and two of $x_s$. The back pressures are $0.9$, $1.0$ and $1.1$
times the time mean of the three forebody transducers. The shock positions are $0.7L$ and $0.9L$ from the hinge. The inviscid estimate of
Section~\ref{sec:problems-hymax}, \SI{93}{\milli\meter}, lies between them.

\subsection{Time integration of the reference model}
\label{sec:methods-integrator}

The reference model is integrated with the average-vector-field (AVF) discrete gradient~\cite{quispel2008avf}. For a step of size $\Delta t$ from
$(\mathbf{q}_0, \mathbf{p}_0)$ to $(\mathbf{q}_1, \mathbf{p}_1)$, with midpoints $\mathbf{q}_m = \tfrac12(\mathbf{q}_0 + \mathbf{q}_1)$ and $\mathbf{p}_m = \tfrac12(\mathbf{p}_0 +
\mathbf{p}_1)$, the step is
\begin{equation}
  \mathbf{q}_1 - \mathbf{q}_0 = \Delta t\, \mathbf{p}_m, \qquad
  \mathbf{p}_1 - \mathbf{p}_0 = \Delta t \big[ -\overline{\nabla U} + \mathbf{f}_m \big], \qquad
  \overline{\nabla U} = \tfrac16 \big[ \nabla U(\mathbf{q}_0) + 4 \nabla U(\mathbf{q}_m) + \nabla U(\mathbf{q}_1) \big],
  \label{eq:avf}
\end{equation}
where $U$ is the potential part of $H$. The term $\mathbf{f}_m$ collects every non-conservative generalized force at the
midpoint: the damping, the aerodynamic force, and any prescribed load $\mathbf{g}$. Unlike $\mathbf{f}$ in
Equation~\eqref{eq:phs}, $\mathbf{f}_m$ includes the damping force.

\begin{proposition}[Discrete energy balance]
\label{prop:avf}
Let $H(\mathbf{q}, \mathbf{p}) = \tfrac12 |\mathbf{p}|^2 + U(\mathbf{q})$, with $U$ a polynomial of degree at most four
in $\mathbf{q}$, and write $H_k = H(\mathbf{q}_k, \mathbf{p}_k)$ for $k = 0, 1$. For any step of
Equation~\eqref{eq:avf}, and for any $\mathbf{f}_m$ whatever, $H_1 - H_0 = \Delta t\, \mathbf{p}_m \cdot \mathbf{f}_m$
exactly.
\end{proposition}

\begin{proof}
Write $\Delta\mathbf{q} = \mathbf{q}_1 - \mathbf{q}_0$. The difference $H_1 - H_0$ is the change in kinetic energy plus the change in
potential energy. We treat the two parts in turn.

\emph{Kinetic part.} A difference of squares factors as a sum times a difference. With $\mathbf{p}_m = \tfrac12(\mathbf{p}_0
+ \mathbf{p}_1)$,
\begin{equation*}
  \tfrac12 \big( |\mathbf{p}_1|^2 - |\mathbf{p}_0|^2 \big)
  \;=\; \tfrac12 (\mathbf{p}_1 + \mathbf{p}_0) \cdot (\mathbf{p}_1 - \mathbf{p}_0)
  \;=\; \mathbf{p}_m \cdot (\mathbf{p}_1 - \mathbf{p}_0)
  \;=\; \Delta t\, \mathbf{p}_m \cdot \big( -\overline{\nabla U} + \mathbf{f}_m \big),
\end{equation*}
where the last equality uses the momentum update of Equation~\eqref{eq:avf}.

\emph{Potential part.} Follow the straight segment $\mathbf{q}(\varsigma) = \mathbf{q}_0 + \varsigma\, \Delta\mathbf{q}$ for $0 \le \varsigma \le 1$. By
the fundamental theorem of calculus and the chain rule, with $\mathrm{d}\mathbf{q}/\mathrm{d}\varsigma = \Delta\mathbf{q}$,
\begin{equation*}
  U(\mathbf{q}_1) - U(\mathbf{q}_0) \;=\; \int_0^1 \frac{\mathrm{d}}{\mathrm{d}\varsigma} U\big(\mathbf{q}(\varsigma)\big) \,\mathrm{d}\varsigma
  \;=\; \int_0^1 \nabla U\big(\mathbf{q}(\varsigma)\big) \cdot \Delta\mathbf{q} \,\mathrm{d}\varsigma .
\end{equation*}
Since $U$ has degree at most four in $\mathbf{q}$, $\nabla U$ has degree at most three. Since $\mathbf{q}(\varsigma)$ is linear in
$\varsigma$, the integrand is a polynomial of degree at most three in $\varsigma$. Simpson's rule, $\int_0^1 \varphi \,\mathrm{d}\varsigma =
\tfrac16 [\varphi(0) + 4\varphi(\tfrac12) + \varphi(1)]$, is exact for polynomials of degree at most three. Apply it to the integrand $\varphi(\varsigma) = \nabla
U\big(\mathbf{q}(\varsigma)\big) \cdot \Delta\mathbf{q}$. Its nodes are $\mathbf{q}(0) =
\mathbf{q}_0$, $\mathbf{q}(\tfrac12) = \mathbf{q}_m$ and $\mathbf{q}(1) = \mathbf{q}_1$. The vector $\Delta\mathbf{q}$ does not depend on $\varsigma$ and factors out of the three-point sum. The integral
is therefore exactly
\begin{equation*}
  U(\mathbf{q}_1) - U(\mathbf{q}_0) \;=\; \overline{\nabla U} \cdot \Delta\mathbf{q} \;=\; \Delta t\, \overline{\nabla U} \cdot \mathbf{p}_m ,
\end{equation*}
with $\overline{\nabla U}$ as in Equation~\eqref{eq:avf}, and the last equality from its position update $\Delta\mathbf{q} =
\Delta t\,\mathbf{p}_m$.

\emph{Sum.} Adding the two parts, the $\overline{\nabla U}$ terms cancel:
\begin{equation*}
  H_1 - H_0 \;=\; \Delta t\, \mathbf{p}_m \cdot \big( -\overline{\nabla U} + \mathbf{f}_m \big) + \Delta t\, \overline{\nabla U} \cdot \mathbf{p}_m
  \;=\; \Delta t\, \mathbf{p}_m \cdot \mathbf{f}_m .
\end{equation*}
No assumption on $\mathbf{f}_m$ was used. The identity therefore holds for any non-conservative force.
\end{proof}

Equation~\eqref{eq:reference model-H} is quartic, and the HyMAX Hamiltonian is quadratic.
Proposition~\ref{prop:avf} therefore applies to both. The implicit step is solved by Newton's method with the analytic
Hessian. The energy balance of Proposition~\ref{prop:avf} is then exact up to the solver tolerance and roundoff, and the computed residual is what is
reported. With $\mathbf{f}_m = -\mathbf{C}\mathbf{p}_m$ alone, Proposition~\ref{prop:avf} gives $H_1 - H_0 = -\Delta t\,
\mathbf{p}_m^{\mathsf{T}} \mathbf{C}\, \mathbf{p}_m \le 0$, because $\mathbf{C}$ is positive semidefinite. With
$\mathbf{f}_m = 0$, $H_1 = H_0$.

\subsection{The record}
\label{sec:methods-record}

A record is a family of reference model trajectories that differ in the parameters of the Hamiltonian or of the
aerodynamic force. Each trajectory is a \emph{member} of the record. Every RC-19 member starts from the buckled static shape under the temperature rise $\Delta T$ alone, with $\Delta p =
0$, as in Section 5.3 of~\cite{brouwer2021}. The shape is found by a trust-region minimization of the potential and polished by Newton's
method. A first-mode displacement of $10^{-3}$ thicknesses, the initial \emph{kick}, is then added. The kick sets a reproducible departure from the static shape. Every HyMAX member starts with the plate at rest before the flow arrives. The state is stored at fixed instants called \emph{time samples}, spaced uniformly by $\Delta t_r$. At each time sample a member stores the state, the prescribed load $\mathbf{g}$ and the energy $H$. It also stores the exact rate
$\dot{\mathbf{p}}$ from the reference model's equations of motion, not differenced from stored states.

\begin{itemize}
  \item \textbf{RC-19.} Three members at $\Delta T = \SI{12.76}{\kelvin}$, with cavity pressures of
  $49.90$, $50.43$ and \SI{50.96}{\kilo\pascal}. These are the package's measured \SI{50.43}{\kilo\pascal} and that value
  plus and minus its stated error of \SI{0.53}{\kilo\pascal}. The cavity differential $\Delta p = p_c - p_s$, with $p_s = \SI{51.0}{\kilo\pascal}$, is $-1103$, $-571$ and
  \SI{-39}{\pascal}. Each
  lasts \SI{0.5}{\second}, which is a fixed setting. It is stored at \SI{25}{\kilo\hertz} and integrated at
  \SI{50}{\kilo\hertz}, which gives two steps per time sample. The storage rate is a setting of the training, not of the experiment. At the experiment's
  \SI{5}{\kilo\hertz}, a component near \SI{1}{\kilo\hertz} has five time samples per period. At \SI{25}{\kilo\hertz} it has
  twenty-five. All three are released
from the same buckled static shape. That shape is one of the two wells of the potential. All three are chaotic, and each switches between the two wells.
  \item \textbf{HyMAX.} Six members: back pressures of $494$, $549$ and \SI{603}{\pascal}, each with the shock at
  $x_s = 91$ and \SI{117}{\milli\meter}. Each runs from \SI{-2}{\milli\second} to \SI{200}{\milli\second} of the
  measured flow, stored at \SI{25}{\kilo\hertz} and integrated at \SI{50}{\kilo\hertz}. The \SI{200}{\milli\second} is
  the package's flow window. At \SI{5}{\kilo\hertz} storage the Nyquist frequency would be \SI{2.5}{\kilo\hertz}, below the
  higher modes of the basis comparisons of Section~\ref{sec:hymax}. Every member is driven by
  the same measured freestream and trailing-edge pressures.
\end{itemize}

In neither case is the family an ensemble of experiments. It varies how one set of measured inputs loads the
structure.

\subsection{The learned model}
\label{sec:methods-identity}

The learned model advances the nondimensional state $\mathbf{x} = (\Wvec, \Wdot) \in \mathbb{R}^{2n}$. It is
conditioned on the parameter vector $\mathbf{z}$: $(\Delta p, \Delta T)$ for RC-19 and $(p_b, x_s)$ for HyMAX.
Its rate law is
\begin{equation}
  \Wvec'' \;=\; -\nabla_{\Wvec} U(\Wvec; \mathbf{z}) \;-\; \Dmat\, \Wdot \;+\; \gvec(\tau)
          \;+\; \Qvec(\Wvec, \Wdot, \mathbf{z}),
  \qquad H = \tfrac12 |\Wdot|^2 + U(\Wvec; \mathbf{z}),
  \label{eq:learned}
\end{equation}
in which $U(\Wvec; \mathbf{z}) = U_{\mathrm{ref}}(w_s \Wvec)/(w_s^2 \omega_1^2)$ is the reference model's potential
$U_{\mathrm{ref}}$ in nondimensional form, held fixed and not trained. The drive $\gvec$ is the reference model's generalized
drive divided by $w_s \omega_1^2$,
$\Dmat$ is the learned dissipation, and $\Qvec$ is the fitted force term. For RC-19 the drive is zero and $U$ depends on
$\mathbf{z}$. For HyMAX $U$ is independent of $\mathbf{z}$. The conditioning enters the prescribed drive $\gvec$ and the fitted term.

\begin{proposition}[The energy identity]
\label{prop:identity}
Let $U$ be any continuously differentiable potential. For every value of the trained and fitted parameters,
every state and every $\mathbf{z}$, solutions of Equation~\eqref{eq:learned} satisfy
\begin{equation}
  \frac{\mathrm{d}H}{\mathrm{d}\tau}
    \;=\; \Wdot \cdot \left( \gvec + \Qvec \right) \;-\; \Wdot^{\mathsf{T}} \Dmat\, \Wdot .
  \label{eq:identity}
\end{equation}
\end{proposition}

\begin{proof}
Along a solution, $\Wvec$ and $\Wdot$ vary in $\tau$, and $\mathbf{z}$ is held fixed. The chain rule applied to $H =
\tfrac12 |\Wdot|^2 + U(\Wvec; \mathbf{z})$ then gives two terms, one from each part,
\begin{equation*}
  \frac{\mathrm{d}H}{\mathrm{d}\tau} \;=\; \Wdot \cdot \Wvec'' \;+\; \nabla_{\Wvec} U \cdot \Wdot .
\end{equation*}
The first comes from $\tfrac12 |\Wdot|^2$, whose gradient with respect to $\Wdot$ is $\Wdot$. The second needs $U$ to be
continuously differentiable, and nothing more. Substitute $\Wvec''$ from Equation~\eqref{eq:learned}:
\begin{equation*}
  \Wdot \cdot \Wvec'' \;=\; -\,\Wdot \cdot \nabla_{\Wvec} U \;-\; \Wdot^{\mathsf{T}} \Dmat\, \Wdot \;+\; \Wdot \cdot \big(
  \gvec + \Qvec \big) .
\end{equation*}
Adding $\nabla_{\Wvec} U \cdot \Wdot$ cancels the first term, which leaves Equation~\eqref{eq:identity}. The cancellation
uses only the form of Equation~\eqref{eq:learned}, in which the conservative force is exactly $-\nabla_{\Wvec} U$. No
property of $\Qvec$, $\Dmat$ or $\gvec$ was used. The identity therefore holds for every value of the parameters. If $\mathbf{z}$
varied along the trajectory, a term $\partial_{\mathbf{z}} U \cdot \mathbf{z}'$ would appear. Every record and rollout in this
report holds $\mathbf{z}$ fixed.
\end{proof}

Proposition~\ref{prop:identity} places no condition on $\Qvec$ or $\Dmat$, and that is the point of it. The
identity is a property of the form of Equation~\eqref{eq:learned}, not of its parameters. It holds at the
random initialization and after every epoch. An epoch is one pass of the optimizer over the training samples. It holds even
after a poor optimizer step. For the quartic $U$ of RC-19, the model evaluates the cubic force as $\mathcal{D}_{ijkl} W_j W_k W_l$. Here
$\mathcal{D}$ is the quartic membrane tensor of Equation~\eqref{eq:quartic}, in nondimensional form. This force is the
gradient of $U$ only when $\mathcal{D}$ is fully symmetric (Lemma~\ref{lem:gradient}). The identity therefore rests on that symmetrization.

\paragraph{The coordinates} The optimizer works in standardized, whitened coordinates. Let $\bar{\mathbf{x}}$ and
$\boldsymbol{\Sigma}_x$ be the mean and covariance of $\mathbf{x}$ over the record. Let $\mathbf{V}$ hold the
eigenvectors of $\boldsymbol{\Sigma}_x$ and $\mathbf{S}$ the diagonal matrix of the square roots of its eigenvalues. Let
$\tilde{\mathbf{x}} = \mathbf{S}^{-1}\mathbf{V}^{\mathsf{T}} (\mathbf{x} - \bar{\mathbf{x}})$, with mean
$\bar{\tilde x}$ and standard deviation $\sigma_{\tilde x}$ taken over all entries of $\tilde{\mathbf{x}}$ in the record.
Both are scalars. The whitened state is $\mathbf{s} = (\tilde{\mathbf{x}} - \bar{\tilde x}\,\mathbf{1})/\sigma_{\tilde x}$,
with $\mathbf{1}$ the vector of ones. An eigenvector is dropped when its
eigenvalue is below $10^{-10}$ times the largest. This affine change of variables is computed once from the record. The model
evaluates Equation~\eqref{eq:learned} in the physical coordinates $(\Wvec, \Wdot)$ and maps the resulting rate back
to the whitened coordinates. The map is invertible only if the whitening keeps all $2n$ directions. A model whose
whitening drops a direction is therefore rejected. Its coordinates no longer represent the state in which the
identity is stated.

\subsection{The learned dissipation}
\label{sec:methods-dissipation}

The dissipation is
\begin{equation}
  \Dmat \;=\; \operatorname{diag}\!\big( \gamma_D + \operatorname{softplus}(\mathbf{d}) \big)
    \;+\; \Bmat\Bmat^{\mathsf{T}},
  \qquad \mathbf{d} \in \mathbb{R}^{n}, \quad \Bmat \in \mathbb{R}^{n \times r}, \quad \gamma_D \ge 0 ,
  \label{eq:dissipation}
\end{equation}
where $\operatorname{softplus}(d) = \ln(1 + e^{d})$ is applied to each entry $d_i$ of $\mathbf{d}$. The floor $\gamma_D$ is a
constant fixed before training. The optimizer adjusts $\theta = (\mathbf{d}, \Bmat)$ and nothing else.

Equation~\eqref{eq:dissipation} is built to meet the requirement that Equation~\eqref{eq:identity} sets. There
$\Wdot^{\mathsf{T}}\Dmat\Wdot$ is the power the dissipation removes, and it must be positive for every nonzero velocity
$\Wdot$. Otherwise the dissipation could feed energy into the panel. Only the symmetric part of $\Dmat$ enters this
power. Let $\mathbf{Z}$ be antisymmetric, with $\mathbf{Z}^{\mathsf{T}} = -\mathbf{Z}$. A scalar equals its own
transpose. Hence $\Wdot^{\mathsf{T}}\mathbf{Z}\Wdot = \Wdot^{\mathsf{T}}\mathbf{Z}^{\mathsf{T}}\Wdot =
-\Wdot^{\mathsf{T}}\mathbf{Z}\Wdot$. A scalar equal to its own negative is zero. The requirement is
therefore that $\Dmat$ be symmetric positive definite. Each of the two parts of Equation~\eqref{eq:dissipation} has one
role. The diagonal part, with strictly positive entries $\gamma_D + \operatorname{softplus}(d_i)$, damps every modal
velocity on its own. The part $\Bmat\Bmat^{\mathsf{T}}$ adds $|\Bmat^{\mathsf{T}}\Wdot|^2 \ge 0$ to the power, and its
off-diagonal entries couple the modes. For $r < n$ it is singular, because $\Bmat^{\mathsf{T}}\Wdot = 0$ has nonzero
solutions $\Wdot$. It therefore cannot serve alone. With $r = n$ and $\gamma_D = 0$, the two parts can represent any
symmetric positive definite matrix $\mathbf{D}_\star$. Take $\gamma_\star$ between zero and the smallest eigenvalue of
$\mathbf{D}_\star$, and set every $d_i = \ln(e^{\gamma_\star} - 1)$, so that every diagonal entry is $\gamma_\star$. The
remainder $\mathbf{D}_\star - \gamma_\star\mathbf{I}$ is then symmetric positive definite, and its Cholesky factor gives
$\Bmat$ with $\Bmat\Bmat^{\mathsf{T}} = \mathbf{D}_\star - \gamma_\star\mathbf{I}$.

\begin{lemma}[The dissipation cannot supply energy]
\label{lem:spd}
For every finite $\theta$, $\Dmat$ is symmetric positive definite, and hence $\Wdot^{\mathsf{T}}\Dmat\Wdot > 0$ whenever
$\Wdot \neq 0$.
\end{lemma}

\begin{proof}
Write $\Dmat = \boldsymbol{\Lambda} + \Bmat\Bmat^{\mathsf{T}}$, with $\boldsymbol{\Lambda}$ the diagonal matrix of
entries $\Lambda_{ii} = \gamma_D + \operatorname{softplus}(d_i)$.

\emph{Symmetry.} A diagonal matrix equals its transpose: $\boldsymbol{\Lambda}^{\mathsf{T}} = \boldsymbol{\Lambda}$.
The transpose of a product reverses its factors, which gives
$(\Bmat\Bmat^{\mathsf{T}})^{\mathsf{T}} = (\Bmat^{\mathsf{T}})^{\mathsf{T}}\Bmat^{\mathsf{T}} = \Bmat\Bmat^{\mathsf{T}}$.
The sum $\Dmat$ is therefore symmetric.

\emph{The quadratic form.} Take any $\mathbf{a} \in \mathbb{R}^n$. The form is linear in the matrix, which gives
\begin{equation}
  \mathbf{a}^{\mathsf{T}}\Dmat\mathbf{a}
  \;=\; \mathbf{a}^{\mathsf{T}}\boldsymbol{\Lambda}\mathbf{a} + \mathbf{a}^{\mathsf{T}}\Bmat\Bmat^{\mathsf{T}}\mathbf{a} .
  \label{eq:spd-split}
\end{equation}
In index form the first term is $\sum_{i,j} a_i \Lambda_{ij} a_j$. Every off-diagonal entry of $\boldsymbol{\Lambda}$ is
zero, which leaves only the terms with $i = j$, and $\mathbf{a}^{\mathsf{T}}\boldsymbol{\Lambda}\mathbf{a} = \sum_i
\Lambda_{ii} a_i^2$. In the second term, set $\mathbf{c} = \Bmat^{\mathsf{T}}\mathbf{a} \in \mathbb{R}^r$. Then
$\mathbf{a}^{\mathsf{T}}\Bmat = \mathbf{c}^{\mathsf{T}}$, and the term is $\mathbf{c}^{\mathsf{T}}\mathbf{c} =
|\mathbf{c}|^2$. Equation~\eqref{eq:spd-split} becomes
\begin{equation}
  \mathbf{a}^{\mathsf{T}}\Dmat\mathbf{a}
  \;=\; \sum_{i=1}^{n} \big(\gamma_D + \operatorname{softplus}(d_i)\big)\, a_i^2 \;+\; |\Bmat^{\mathsf{T}}\mathbf{a}|^2 .
  \label{eq:spd-form}
\end{equation}

\emph{The sign of each term.} For finite $d_i$, $e^{d_i} > 0$. Hence $1 + e^{d_i} > 1$ and
$\operatorname{softplus}(d_i) > 0$. With $\gamma_D \ge 0$, every coefficient $\Lambda_{ii}$ is strictly positive. A
squared norm is never negative. Hence the last term of Equation~\eqref{eq:spd-form} is non-negative.

\emph{Conclusion.} Let $\mathbf{a} \neq 0$. At least one entry $a_k$ is nonzero. Hence the sum in
Equation~\eqref{eq:spd-form} is at least $\Lambda_{kk} a_k^2 > 0$. Adding the non-negative last term keeps it
positive. Hence $\mathbf{a}^{\mathsf{T}}\Dmat\mathbf{a} > 0$ for every nonzero $\mathbf{a}$. A symmetric matrix with
this property is positive definite by definition. Setting $\mathbf{a} = \Wdot$ gives the statement.
\end{proof}

The lemma holds for every value of $\mathbf{d}$ and $\Bmat$ that training can produce. Its proof relies only on the
algebraic form of Equation~\eqref{eq:dissipation}, a positive diagonal matrix plus $\Bmat\Bmat^{\mathsf{T}}$. It does
not depend on the training record, the number of epochs, or whether the optimizer converged. Positive definiteness is
thus guaranteed before training, after training, and after an unsuccessful training run. The one requirement is that
$\theta$ be finite. As $d_i \to -\infty$, $\operatorname{softplus}(d_i) \to 0$, and with $\gamma_D = 0$ the matrix would
be only positive semidefinite.

No constraint and no penalty is needed to keep the dissipation admissible because every finite $\theta$ gives a
positive definite $\Dmat$. The vector $\mathbf{d}$ holds one entry for each of
the $n$ retained modes and the $n$ by $r$ factor $\Bmat$ holds $nr$ entries. The parameter vector $\theta$ therefore holds $n + nr$ numbers. RC-19
retains $n = 7$ measured modes and uses rank $r = 4$, a fixed setting, giving $7 + 28 = 35$ numbers. HyMAX retains
$n = 2$ modes and uses the full rank $r = 2$, giving $2 + 4 = 6$ numbers. Every run reported here sets the floor to $\gamma_D = 0$. 
Brouwer et al.~\cite{brouwer2021} damp the same panel with the two-parameter Rayleigh form. 
Equation~\eqref{eq:dissipation} contains that form as a special case. In the units of $\tau$, the Rayleigh force
$-\mathbf{C}\mathbf{p}$ becomes $-\omega_1^{-1} \mathbf{C}\, \Wdot$, a diagonal matrix with positive entries $c_i =
(\alpha_d + \beta_d \omega_i^2)/\omega_1$. Setting $\Bmat = 0$, $\gamma_D = 0$ and $d_i = \ln(e^{c_i} - 1)$ gives
$\operatorname{softplus}(d_i) = c_i$, which reproduces it exactly.
\subsection{The fitted force term}
\label{sec:methods-head}

The fitted force term is affine in the features $\mathbf{v} = (\Wvec, \Wdot, \hat{\mathbf{z}})$,
\begin{equation}
  \Qvec(\Wvec, \Wdot, \mathbf{z}) \;=\; \mathbf{A}_h\, \mathbf{v} + \mathbf{b}_h ,
  \label{eq:head}
\end{equation}
where $\mathbf{A}_h \in \mathbb{R}^{n \times (2n+2)}$ is the coefficient matrix and $\mathbf{b}_h \in \mathbb{R}^{n}$ the
constant offset. The feature vector $\mathbf{v}$ stacks the $n$ modal displacements, the $n$ modal rates and the two
standardized conditioning parameters $\hat{\mathbf{z}}$. Its length is therefore $2n + 2$.

\paragraph{The conditioning features} Let $\mathbf{z}^{(m)}$ be the conditioning vector of record member $m$, for
$m = 1, \dots, N_m$. Each entry $j = 1, 2$ is standardized over the members,
\begin{equation}
  \begin{gathered}
  \hat z_j \;=\; \frac{z_j - \bar z_j}{\sigma_{z,j}}, \qquad
  \bar z_j \;=\; \frac{1}{N_m} \sum_{m=1}^{N_m} z_j^{(m)}, \\
  \sigma_{z,j} \;=\; \max\!\left( \left[ \frac{1}{N_m-1} \sum_{m=1}^{N_m} \big( z_j^{(m)} - \bar z_j \big)^2 \right]^{1/2},\; 10^{-8} \right).
  \end{gathered}
  \label{eq:z-standard}
\end{equation}
The mean $\bar z_j$ and the deviation $\sigma_{z,j}$ are computed once from the record. An entry that
does not vary over the record equals $\bar z_j$ at every member, and hence $\hat z_j = 0$ there. RC-19 retains $n = 7$
measured modes (Section~\ref{sec:methods-panel}). Hence $\mathbf{A}_h$ is seven by sixteen and $\mathbf{b}_h$ has seven
entries.

\paragraph{The fit} The fitted term is not trained. Before the first epoch it is fitted by ridge-regularized least
squares~\cite{hoerl1970ridge} to the residual acceleration
\begin{equation}
    \mathbf{r} \;=\; \Wvec''_{\mathrm{exact}} + \nabla_{\Wvec} U - \gvec
  \label{eq:head-target}
\end{equation}
over every time sample of every record member. Here $\Wvec''_{\mathrm{exact}} =
\dot{\mathbf{p}}/(w_s \omega_1^2)$ is the exact acceleration stored with the record, in the units of $\tau$. Let
$\mathbf{v}_k$ and $\mathbf{r}_k$ be the features and the residual acceleration at time sample $k$. Member $m$ stores $N_t^{(m)}$ time samples, and $k$ runs over all
$N_r = \sum_{m=1}^{N_m} N_t^{(m)}$ of them. The coefficients solve
\begin{equation}
  \min_{\mathbf{A}_h \in \mathbb{R}^{n \times (2n+2)},\; \mathbf{b}_h \in \mathbb{R}^{n}} \;
  \sum_{k=1}^{N_r} \big| \mathbf{A}_h \mathbf{v}_k + \mathbf{b}_h - \mathbf{r}_k \big|^2
  \;+\; \epsilon_r N_r \big( \|\mathbf{A}_h\|_F^2 + |\mathbf{b}_h|^2 \big), \qquad \epsilon_r = 10^{-8} ,
  \label{eq:ridge}
\end{equation}
with $\|\cdot\|_F$ the Frobenius norm. The fitted term is then held fixed during training.
Each RC-19 member covers \SI{0.5}{\second} at \SI{25}{\kilo\hertz}. Counting the initial state,
that gives $N_t^{(m)} = 0.5 \times 25\,000 + 1 = 12\,501$, and the three members give $N_r = 37\,503$. Each HyMAX
member covers \SI{202}{\milli\second}, which gives $N_t^{(m)} = 5\,051$, and the six members give $N_r = 30\,306$.

\paragraph{The rate} The fit must use the exact rate $\dot{\mathbf{p}}$ stored with the record
(Section~\ref{sec:methods-record}). The aerodynamic force is a correction of about two percent on the elastic rate. In a test, the reference model was sampled at \SI{5}{\kilo\hertz}, which is the rate of the experiment's digital image
correlation (Section~\ref{sec:problems-rc19}). A finite-difference rate then erred by twelve times that force. The fit then returned a term that was not the aerodynamic force.
A measured record carries no exact rate. The only rate available would then be a fourth-order central difference,
$\dot{\mathbf{p}}_k \approx (-\mathbf{p}_{k+2} + 8\mathbf{p}_{k+1} - 8\mathbf{p}_{k-1} + \mathbf{p}_{k-2})/(12\,
\Delta t_r)$, with $\Delta t_r$ the time-sample spacing. The test above shows that this rate needs a finer sampling
than the experiment's.

\paragraph{The drive} The drive $\gvec$ in Equation~\eqref{eq:head-target} is subtracted at the time samples, not at the
points between them where the integrator samples it. A drive that varies between time samples would otherwise leave a
residual that the fitted term absorbs.

\paragraph{What the fit can represent} On RC-19 the reference model's non-conservative forces,
$-(\mathbf{A}\mathbf{q} + \mathbf{B}\mathbf{p}) - \mathbf{C}\mathbf{p}$, are linear in the state. With $\mathbf{q} = w_s
\Wvec$, $\mathbf{p} = w_s \omega_1 \Wdot$, and accelerations divided by $w_s \omega_1^2$, they become $-\omega_1^{-2}
\mathbf{A}\, \Wvec - \omega_1^{-1} (\mathbf{B} + \mathbf{C})\, \Wdot$. Equation~\eqref{eq:head} represents them exactly,
with $\mathbf{A}_h = [-\omega_1^{-2} \mathbf{A},\; -\omega_1^{-1} (\mathbf{B} + \mathbf{C}),\; \mathbf{0}]$ and
$\mathbf{b}_h = \mathbf{0}$. The fit residual then measures only the error of the least-squares solve. On HyMAX the
operators follow the measured inflow in time. A fitted term with constant coefficients can therefore represent them only
on average. The fit residual there measures that limitation, and Section~\ref{sec:hymax} reports it.

\paragraph{Parameter count} A reader checking the RC-19 model should count thirty-five trained numbers and $7 \times 16 +
7 = 119$ fitted ones. The potential adds fixed coefficients that are neither trained nor fitted. Those are seven
frequencies, forty-nine entries in the preload operator $\mathbf{G}$ and seven in the load vector $\boldsymbol{\ell}$.
The quartic tensor $\mathcal{D}$ adds $7^4 = 2401$. The total is $7 + 49 + 7 + 2401 = 2464$.

\subsection{Time integration of the learned model}
\label{sec:methods-time}

The learned model, written as a first-order system in $\mathbf{x}$, is integrated by one of two time-stepping
schemes. Both are second-order accurate.

\paragraph{Generalized-$\alpha$} The first scheme is the generalized-$\alpha$ method of Chung and
Hulbert~\cite{chungHulbert1993}, which the earlier reports of this series also use~\cite{aguilo2026spmlopen}. Its
first-order form is that of Jansen et al.~\cite{jansenWhitingHulbert2000}. Its
parameters follow from one number, the spectral radius at infinite frequency $\rho_\infty$, which sets how strongly the
step damps high frequencies. This $\rho_\infty$ is a parameter of the step, not a density. We take $\rho_\infty = 0.5$,
the value of the open-systems report~\cite{aguilo2026spmlopen}. At $\rho_\infty = 1$ the step is the implicit midpoint rule, which leaves high-frequency content undamped. It satisfies Equation~\eqref{eq:identity} only up to its truncation error. Its numerical damping and truncation
error differ from those of the AVF step that generated the record. The trained RC-19 and HyMAX models reported in this
work were trained and evaluated with this step.

\paragraph{The AVF step} This is Equation~\eqref{eq:avf} applied to Equation~\eqref{eq:learned}, with $\mathbf{f}_m
= -\Dmat\Wdot_m + \Qvec(\mathbf{x}_m, \mathbf{z}) + \gvec$, at the step size the reference model used to generate the record.

\begin{proposition}[Discrete identity for the learned model]
\label{prop:avf-learned}
Let $\mathbf{x}_0$ and $\mathbf{x}_1$ be the states before and after one AVF step of size $\Delta\tau$. Let $\mathbf{x}_m = \tfrac12 (\mathbf{x}_0 + \mathbf{x}_1) = (\Wvec_m, \Wdot_m)$, and $H_k = \tfrac12
|\Wdot_k|^2 + U(\Wvec_k; \mathbf{z})$ for $k = 0, 1$. For every value of the trained and fitted parameters, $H_1 - H_0 =
\Delta\tau\, \Wdot_m \cdot \big( \Qvec(\mathbf{x}_m, \mathbf{z}) + \gvec - \Dmat\, \Wdot_m \big)$.
\end{proposition}

\begin{proof}
Replace $\mathbf{q}$ by $\Wvec$, $\mathbf{p}$ by $\Wdot$ and $\Delta t$ by $\Delta\tau$ in Equation~\eqref{eq:avf}. The
learned step then has the same two updates, with $U(\cdot\,; \mathbf{z})$ in place of $U$. With $\mathbf{z}$ fixed, $U$ is
a polynomial in $\Wvec$ of degree at most four: quartic for RC-19 and quadratic for HyMAX. The kinetic part of $H$ is
$\tfrac12 |\Wdot|^2$, as Proposition~\ref{prop:avf} requires. That proof uses only the two updates and places no
condition on $\mathbf{f}_m$. In particular, $\mathbf{f}_m$ may depend on $\mathbf{x}_1$ through $\mathbf{x}_m$.
Substituting $\mathbf{f}_m = -\Dmat\Wdot_m + \Qvec(\mathbf{x}_m, \mathbf{z}) + \gvec$ gives the statement.
\end{proof}

With the AVF step, the learned model and the reference model share one integrator. No difference therefore arises from
the time stepping itself. On a chaotic record this matters, because any per-step difference grows as $e^{\lambda t}$, with $\lambda$ the
record's largest Lyapunov exponent.
Section~\ref{sec:rc19-chaos} measures what the choice changes. The AVF step is solved by Newton's method with the Jacobian in closed form. Split $\mathbf{A}_h = [\mathbf{A}_W,
\mathbf{A}_{W'}, \mathbf{A}_z]$ by the three blocks of features. The step solves $\mathbf{F}(\mathbf{x}_1) = 0$, with
$\mathbf{F}_W = \Wvec_1 - \Wvec_0 - \Delta\tau\, \Wdot_m$ and $\mathbf{F}_{W'} = \Wdot_1 - \Wdot_0 + \Delta\tau
(\overline{\nabla U} - \mathbf{f}_m)$. Differentiating each term,
\begin{equation}
\begin{aligned}
  \frac{\partial \mathbf{F}_W}{\partial \Wvec_1} &= \mathbf{I}, &
  \frac{\partial \mathbf{F}_W}{\partial \Wdot_1} &= -\tfrac{\Delta\tau}{2} \mathbf{I}, \\
  \frac{\partial \mathbf{F}_{W'}}{\partial \Wvec_1} &= \Delta\tau \big[ \tfrac13 \nabla^2 U(\Wvec_m) + \tfrac16 \nabla^2
  U(\Wvec_1) \big] - \tfrac{\Delta\tau}{2} \mathbf{A}_W, &
  \frac{\partial \mathbf{F}_{W'}}{\partial \Wdot_1} &= \mathbf{I} - \tfrac{\Delta\tau}{2} \big( \mathbf{A}_{W'} - \Dmat
  \big) .
\end{aligned}
\label{eq:newton-jacobian}
\end{equation}
The factors $\tfrac12$ and $\tfrac16$ come from $\mathbf{x}_m = \tfrac12(\mathbf{x}_0 + \mathbf{x}_1)$ and from the weights
of $\overline{\nabla U}$. The Hessian of $U$ is explicit, $\Dmat$ is
itself a matrix, and the fitted force term is evaluated at the midpoint $\mathbf{x}_m$. Its derivative with respect to the end state
$\mathbf{x}_1$ is therefore half the columns of $\mathbf{A}_h$ that multiply $\Wvec$ and $\Wdot$. The columns that
multiply $\mathbf{z}$ do not enter, because $\mathbf{z}$ is fixed. Gradients of the objective with respect to $\theta$ are taken through the converged Newton solution by the
implicit function theorem~\cite{bai2019deq}. At $\mathbf{F}(\mathbf{x}_1) = 0$, $\partial\mathbf{x}_1/\partial\theta =
-(\partial\mathbf{F}/\partial\mathbf{x}_1)^{-1}\, \partial\mathbf{F}/\partial\theta$. The parameters enter $\mathbf{F}$
only through $\Dmat\Wdot_m$ in $\mathbf{f}_m$.

\subsection{The objective}
\label{sec:methods-objective}

The objective is a data term and a one-sided bound,
\begin{equation}
        J(\theta) \;=\; \frac{1}{2nN} \sum_{i=1}^{N} \big\| \hat{\mathbf{s}}_i(\theta) - \mathbf{s}_i \big\|^2
    \;+\; \mu\, \Big\langle \max\!\big(0,\, \|\hat{\mathbf{s}}\| - c_u\, s_{\max}\big)^2 \Big\rangle ,
  \label{eq:loss}
\end{equation}
in which $\hat{\mathbf{s}}_i$ is the predicted state in the whitened coordinates and $\mathbf{s}_i$ is the recorded
one. The norm $\|\cdot\|$ is Euclidean, and in the second term $\hat{\mathbf{s}}$ stands for each $\hat{\mathbf{s}}_i$.

\paragraph{The data term} Each training sample is a rollout. The model starts from one recorded state and is integrated
over sixteen time samples, with no recorded state fed back. The sum and the mean $\langle\cdot\rangle$ run over these
sixteen time samples and over the samples of a batch. Together they comprise $N$ predicted states. Dividing by $2n$
averages the squared error over the $2n$ components of each state.

\paragraph{The bound term} Let $s_{\max}$ be the largest $\|\mathbf{s}\|$ over all time samples of the record. The second
term is zero whenever the norm of the predicted state stays within $c_u = 1.4$ times $s_{\max}$. Beyond that bound it
keeps one diverging rollout from spoiling an epoch. Its weight is $\mu = 0.05$.

\paragraph{Guards on the rollout} Two guards act on the rollout itself. The first is a radial saturation. With radius
$R_{\mathrm{sat}} = 3 s_{\max}$ and knee $R_{\mathrm{knee}} = 0.9 R_{\mathrm{sat}}$, it maps $\hat{\mathbf{s}}$ to
$\hat{\mathbf{s}}\, \varrho(\|\hat{\mathbf{s}}\|)/\|\hat{\mathbf{s}}\|$. Write $\hat s = \|\hat{\mathbf{s}}\|$. Then
$\varrho(\hat s) = \hat s$ for $\hat s \le R_{\mathrm{knee}}$, and
$\varrho(\hat s) = R_{\mathrm{knee}} + (R_{\mathrm{sat}} - R_{\mathrm{knee}}) \tanh\!\big((\hat s - R_{\mathrm{knee}})/(R_{\mathrm{sat}} - R_{\mathrm{knee}})\big)$ beyond.
The map is exactly the identity below $2.7\, s_{\max}$, is continuously differentiable, and never exceeds $3\, s_{\max}$.
The second guard skips the optimizer update for any batch whose rollout exceeds fifty times $s_{\max}$. No update was
skipped in the runs reported here.

\paragraph{The optimizer} Optimization is by Adam~\cite{kingma2015adam}, at a learning rate of $2\times10^{-3}$ annealed
on a cosine schedule~\cite{loshchilov2017sgdr} to $10^{-5}$ over a budget of $400$ epochs. Batches hold eight samples.
Gradients are clipped to unit norm. When the Euclidean norm of the gradient over all trained parameters exceeds one, the
gradient is rescaled to norm one. The learning rates, epoch budget, window, batch, bound and guard factors follow the
protocol of the open-systems report of this series~\cite{aguilo2026spmlopen}. They were carried over unchanged and not
tuned for these problems. The trained model carries the parameters of its final epoch.

The objective contains no term that expresses conservation, dissipativity, symmetry, or any residual of the governing
equations. Proposition~\ref{prop:identity} is not encouraged by the objective. It is a property of the
parameterization, and Lemma~\ref{lem:spd} is a property of Equation~\eqref{eq:dissipation}. In a penalty formulation such as a physics-informed network~\cite{raissi2019pinn}, the governing equations compete with
the data for weight in a sum. They hold only to the degree the penalty buys. Here there is no rate of exchange between the two, and the optimizer
cannot trade structure for fit.

\subsection{Assimilation as a generalized force}
\label{sec:methods-observer}

Because the response is chaotic (Section~\ref{sec:rc19-chaos}), no model tracks a single realization indefinitely. A small difference between the model state and the true state grows as $e^{\lambda t}$, with $\lambda$ the largest
Lyapunov exponent. Let the difference start at size $\varepsilon_0$. The state scale $\sigma$ is the size of the response,
at which the difference saturates. The tracking horizon is the time to grow from $\varepsilon_0$ to $\sigma$, $T_h =
\lambda^{-1}\ln(\sigma/\varepsilon_0)$.
A more accurate model reduces $\varepsilon_0$, which lengthens the horizon only logarithmically. Correcting the state with measurement is therefore built into the same
energy accounting, as one more generalized force.
This section adds that correction to the reference model and derives its gains. Section~\ref{sec:rc19-assimilation}
tests it in a twin experiment with synthetic observations. The test asks whether the experiment's published points and
rate would suffice to track the panel. No other result in this report depends on it.

Let $\Phimat \in \mathbb{R}^{n \times n_s}$ hold the unit-mass mode shapes at the $n_s$ instrumented points. The
instrument then reads $\mathbf{y} = \Phimat^{\mathsf{T}}\mathbf{q}$. The correction is a form of nudging~\cite{hoke1976nudging}, a Luenberger observer~\cite{luenberger1964observer} acting
through the same generalized forces as the flow. It feeds back the error in both the reading
and its rate,
\begin{equation}
  \mathbf{f}_{\mathrm{obs}} \;=\; k_p\, \Phimat \big( \mathbf{y}_{\mathrm{obs}} - \Phimat^{\mathsf{T}}\mathbf{q} \big)
  \;+\; k_d\, \Phimat \big( \dot{\mathbf{y}}_{\mathrm{obs}} - \Phimat^{\mathsf{T}}\mathbf{p} \big),
  \qquad \mathbf{M}_\Phi = \Phimat\Phimat^{\mathsf{T}} .
  \label{eq:nudging}
\end{equation}
Here $\mathbf{f}_{\mathrm{obs}}$ is the correction force, and $\mathbf{y}_{\mathrm{obs}}$ and $\dot{\mathbf{y}}_{\mathrm{obs}}$
are the observed reading and its rate. The scalars $k_p > 0$ and $k_d > 0$ are the gains on the reading and on its rate.
The matrix $\mathbf{M}_\Phi$ is the observability Gramian of the instrumented points.
It is applied in two halves. The state-dependent half, $-k_p \mathbf{M}_\Phi \mathbf{q} - k_d \mathbf{M}_\Phi \mathbf{p}$,
is added to the aerodynamic operators, $\mathbf{A} \to \mathbf{A} + k_p\mathbf{M}_\Phi$ and $\mathbf{B} \to \mathbf{B} +
k_d\mathbf{M}_\Phi$, and the implicit step solves it with the Newton Jacobian. The measurement half, $k_p \Phimat \mathbf{y}_{\mathrm{obs}} + k_d \Phimat
\dot{\mathbf{y}}_{\mathrm{obs}}$, is a function of time alone and enters as a drive. The rate term must use the rate of
the observation. Feeding back $-k_d\mathbf{M}_\Phi\mathbf{p}$ alone would leave a drive proportional to the true velocity that never
vanishes. The error equation below would then not be the one being solved. The observed rate is the
difference of consecutive samples, which carries the instrument's noise with it. Each reading at each point carries independent zero-mean Gaussian noise of a stated standard deviation. Noise is drawn
once per sample, not at every evaluation. The observer therefore cannot average it away.

Between samples, the reading is reconstructed in one of two ways. Let samples $\mathbf{y}_j$ arrive at $t_j = j\,
\Delta t_s$, and let $t_j \le t < t_{j+1}$. A zero-order \emph{hold} uses $\mathbf{y}_j$ and the rate $(\mathbf{y}_j -
\mathbf{y}_{j-1})/\Delta t_s$. It is the only reconstruction available in real time. A \emph{smoother} uses $\mathbf{y}_j +
\upsilon (\mathbf{y}_{j+1} - \mathbf{y}_j)$, with $\upsilon = (t - t_j)/\Delta t_s$, and the rate $(\mathbf{y}_{j+1} -
\mathbf{y}_j)/\Delta t_s$. It uses the next sample as well as the last. The smoother is non-causal. It is legitimate here only
because assimilating a recorded series is an offline problem. Every tracking number reported in
Section~\ref{sec:rc19-assimilation} states which reconstruction produced it.

\begin{corollary}[The budget closes under assimilation]
\label{cor:observer}
Consider the reference model in physical time, with the operators and the drive above. Every AVF step satisfies $H_1 -
H_0 = \Delta t\, \mathbf{p}_m \cdot \mathbf{f}_m$ exactly, where $\mathbf{f}_m$ now includes the correction. The work done by the observer is a reported term of the
budget.
\end{corollary}

\begin{proof}
The correction changes $\mathbf{f}_m$ and leaves $U$ unchanged. With it, $\mathbf{f}_m = -\mathbf{C}\mathbf{p}_m -
(\mathbf{A} + k_p \mathbf{M}_\Phi) \mathbf{q}_m - (\mathbf{B} + k_d \mathbf{M}_\Phi) \mathbf{p}_m + \mathbf{g} + k_p
\Phimat \mathbf{y}_{\mathrm{obs}} + k_d \Phimat \dot{\mathbf{y}}_{\mathrm{obs}}$. Proposition~\ref{prop:avf} holds for any
$\mathbf{f}_m$, which gives the statement. The observer's share of the change is $\Delta t\, \mathbf{p}_m \cdot
\mathbf{f}_{\mathrm{obs},m}$, with $\mathbf{f}_{\mathrm{obs},m}$ the correction evaluated at the step's midpoint.
\end{proof}

The corollary does not say that energy is conserved. Assimilation injects energy by design. It says that the discrete energy balance holds exactly, and that the injected energy appears as a separate,
reported term. A filter that reset the state between steps would change $H$ by a jump
that is the work of no force. The balance $H_1 - H_0 = \Delta t\, \mathbf{p}_m \cdot \mathbf{f}_m$ would then fail
at every reset. The correction is therefore a continuous force rather than a jump.

The gains are derived, not tuned. The derivation starts from the corrected reference model, which consists of the modal equations of
motion of the panel with the correction added. With $\mathbf{p} = \dot{\mathbf{q}}$ and the correction's
$\mathbf{f}_m$ from the proof above, it reads
\begin{equation}
  \ddot{\mathbf{q}} + \big( \mathbf{C} + \mathbf{B} + k_d \mathbf{M}_\Phi \big) \dot{\mathbf{q}} + \nabla U(\mathbf{q})
  + \big( \mathbf{A} + k_p \mathbf{M}_\Phi \big) \mathbf{q}
  \;=\; \mathbf{g} + k_p \Phimat \mathbf{y}_{\mathrm{obs}} + k_d \Phimat \dot{\mathbf{y}}_{\mathrm{obs}} .
  \label{eq:corrected}
\end{equation}
Let $\mathbf{q}$ and $\mathbf{q} + \boldsymbol{\delta}$ both solve Equation~\eqref{eq:corrected} with the same
observations. Write Equation~\eqref{eq:corrected} for each and subtract the first from the second. The right-hand side
is a function of time alone and cancels. Each linear term leaves the same operator acting on $\boldsymbol{\delta}$. The
potential leaves $\nabla U(\mathbf{q} + \boldsymbol{\delta}) - \nabla U(\mathbf{q})$. A first-order Taylor
expansion gives $\nabla U(\mathbf{q} + \boldsymbol{\delta}) - \nabla U(\mathbf{q}) = \nabla^2 U(\mathbf{q})\,
\boldsymbol{\delta} + O(|\boldsymbol{\delta}|^2)$, with $\nabla^2 U$ the Hessian of the potential. Dropping the
second-order term gives the error equation
\begin{equation}
    \ddot{\boldsymbol{\delta}} + \big( \mathbf{C} + \mathbf{B} + k_d \mathbf{M}_\Phi \big) \dot{\boldsymbol{\delta}}
  + \big( \nabla^2 U(\mathbf{q}) + \mathbf{A} + k_p \mathbf{M}_\Phi \big) \boldsymbol{\delta} = 0 ,
  \label{eq:error}
\end{equation}
which is Equation 19 of~\cite{brouwer2021} with the gains added to the operators. Its largest Lyapunov exponent, the closed-loop exponent, is computed along the trajectory being tracked. The
estimator~\cite{benettin1980lyapunov} is the one used for the uncorrected, open-loop panel. The critical gain is the smallest $k_p$, at a fixed $k_d$, for which the
closed-loop exponent is negative. It is found by bisection. At the chaotic condition, with an open-loop exponent of \SI{285}{\per\second}, no position gain alone made the
closed-loop exponent negative. A restoring force turns growth into oscillation but does not damp it. Rate feedback
contracted the error.

$\mathbf{M}_\Phi$ is symmetric positive semidefinite, because $\mathbf{a}^{\mathsf{T}} \Phimat\Phimat^{\mathsf{T}}
\mathbf{a} = |\Phimat^{\mathsf{T}} \mathbf{a}|^2 \ge 0$ for every $\mathbf{a}$. A positive gain can therefore only add
stiffness or damping to the error dynamics. Its condition number decides whether the instrumented points can observe every mode.
If mode $i$ had a node at every point, row $i$ of $\Phimat$ would be zero. Then $\mathbf{1}_i^{\mathsf{T}}
\mathbf{M}_\Phi \mathbf{1}_i = 0$, with $\mathbf{1}_i$ the $i$th unit vector, and $\mathbf{M}_\Phi$ would be singular. No
gain could then recover that mode's amplitude.

The observations in this work are synthetic. They are readings of a reference model trajectory, sampled at the
rate and with the noise stated with each result in Section~\ref{sec:rc19-assimilation}. This is a twin
experiment~\cite{kalnay2003assimilation}, in which the model is exact and only its initial state is wrong. It measures
the tracking error that a given sampling rate and noise level allow. It is not agreement with the
experiment.

\section{Experimental setup and verification protocol}
\label{sec:setup}

\subsection{Training}
\label{sec:setup-training}

\paragraph{Training data} The reference models generate both training records from the challenge-problem
packages~\cite{aepw4hswg}. The HyMAX model is trained on the six-member record of Section~\ref{sec:methods-record}.
The RC-19 model is trained on the three-member record of the same section, at the periodic case's temperature rise of
\SI{12.76}{\kelvin}. Sections~\ref{sec:rc19-regime} and~\ref{sec:rc19-chaos} run the RC-19 model at other temperatures.

Temperature could enter the RC-19 model in two places. One is the exact potential $U$. The other is the fitted force
term, where the standardized temperature rise $\hat z_2$ of Equation~\eqref{eq:z-standard} multiplies the column
$\mathbf{a}_T$ of $\mathbf{A}_z$. The training record holds the single value $\Delta T = \SI{12.76}{\kelvin}$. Every
member therefore has $z_2^{(m)} = \bar z_2$, and Equation~\eqref{eq:z-standard} gives $\hat z_2 = 0$ at every time
sample. Thus, the product $\mathbf{a}_T \hat z_2$ vanishes, and $\mathbf{a}_T$ does not appear in the data term of
Equation~\eqref{eq:ridge}. Only the ridge term acts on $\mathbf{a}_T$, and it is smallest at $\mathbf{a}_T = \mathbf{0}$.
Hence, $\mathbf{a}_T = \mathbf{0}$ and the temperature reaches the RC-19 model only through $U$.

\paragraph{Stopping} Training stopped after $128$ epochs for RC-19 and $263$ for HyMAX, within a budget of $400$.
A second HyMAX model with three retained modes, reported in Section~\ref{sec:hymax}, stopped after $247$ epochs.
Convergence is judged against the lowest epoch mean of the objective $J$ of Equation~\eqref{eq:loss} recorded in each
run. On RC-19 that lowest value is $2.44\times10^{-5}$, at the 97th epoch. The epoch mean was $2.59\times10^{-5}$ after
the first epoch and within one percent of the lowest value after the second. On HyMAX the lowest value is
$3.17\times10^{-5}$, at the 135th epoch, and the epoch mean was within one percent of it after the fifth. Over the last
twenty epochs the epoch mean varied by $2.5$ percent on RC-19 and $3.5$ percent on HyMAX.

\subsection{RC-19 protocol}
\label{sec:setup-rc19}

\paragraph{Preregistration} Following~\cite{nosek2018preregistration}, the prediction of each study was fixed before its
run, together with its acceptance conditions and the statistics they test. Table~\ref{tab:predictions} lists them, and no
condition was moved afterward. One condition proved unusable: R2 set a tolerance below the noise of its statistic. Its
failure is reported, and its replacement, R3, was fixed before the replacement run (Section~\ref{sec:predictions}).

\paragraph{Estimators} Each RC-19 response is summarized by the quantities of Table~\ref{tab:estimators}, all read at
three-quarter chord, with displacements in thicknesses. The spectra use the package's spectral settings
(Section~\ref{sec:problems-rc19}). Every model frequency therefore lies on a \SI{4.88}{\hertz} grid. The measured
\SI{236}{\hertz} of Brouwer et al.~\cite{brouwer2021} lies between two bins of this grid, at $234.4$ and
\SI{239.3}{\hertz}. The two regime thresholds and the band-error power fraction are fixed settings
of this study. The Lyapunov estimator was checked on one mode of stiffness $-900^2\,\mathrm{s}^{-2}$, whose error obeys
$\ddot\delta = 900^2 \delta$ and grows at exactly \SI{900}{\per\second}. It returns \SI{900.1}{\per\second}.

\begin{table}[htbp]
  \centering
  \caption{Estimators of the RC-19 responses, all read at three-quarter chord.}
  \label{tab:estimators}
  \small
  \begin{tabular}{>{\raggedright\arraybackslash}p{0.17\textwidth}>{\raggedright\arraybackslash}p{0.47\textwidth}>{\raggedright\arraybackslash}p{0.28\textwidth}}
    \toprule
    Quantity & Definition & Setting \\
    \midrule
    Regime & static when the displacement standard deviation is below the static threshold; otherwise periodic when the
      dominant spectral peak and its two neighboring bins carry more than the periodic fraction of the variance;
      otherwise chaotic & static threshold $0.02$ thicknesses; periodic fraction one half \\
    Dominant frequency & strongest bin of the displacement spectrum & Hamming windows of \SI{0.2048}{\second}, half
      overlap; \SI{4.88}{\hertz} bins \\
    Amplitude & standard deviation of the displacement, in thicknesses & -- \\
    Peak amplitude & largest deflection magnitude over the record, in thicknesses & -- \\
    Band error & root-mean-square difference of base-10 log power between two spectra, divided by the root-mean-square
      of the reference spectrum's log power, over the strongest reference bins & bins ranked by reference power and kept
      until they hold \SI{99}{\percent} of it \\
    Largest Lyapunov exponent & growth rate of the linearized error equations of Brouwer et al., Equations 19 to 23
      of~\cite{brouwer2021}, marched along the model's own trajectory with the operators of the model under test, in
      \si{\per\second} & -- \\
    \bottomrule
  \end{tabular}
\end{table}

\paragraph{Repeat sets and ensembles} In a chaotic response, a frequency means nothing without the number of
trajectories behind it. The two buckled static shapes of opposite sign of the heated panel are called its \emph{wells}. A \emph{kick} is a small first-mode displacement that is added to the buckled shape before release. A \emph{repeat set} is eight kicked releases per well, with kicks spaced logarithmically from $10^{-3}$ to $10^{-1}$
thicknesses. Eight is a fixed setting of this study. A single release cannot decide a regime. A set therefore needs several. An
\emph{ensemble} of $N_e$ members samples kicks inside a stated interval and reports a mean with its standard deviation.
The ensembles of Section~\ref{sec:rc19-agreement} hold thirty members. The relative standard error of a sample standard
deviation is about $1/\sqrt{2(N_e - 1)}$, which is \SI{32}{\percent} at six members and \SI{13}{\percent} at thirty.

\paragraph{Vocabulary} Several sources produce a number in this report, and they are easily run together. The terms
in Table~\ref{tab:vocabulary} are used strictly. Two operating conditions recur, and they differ. The
\emph{experimental operating condition} is the one the tunnel ran at. The \emph{simulation operating conditions} are two 
conditions from the simulations of Brouwer et al.~\cite{brouwer2021}. The periodic one, at \SI{52.7}{\kilo\pascal} and
\SI{8.8}{\kelvin}, is the single-well response of their Fig.~11a. The chaotic one, at \SI{51.1}{\kilo\pascal} and
\SI{13.6}{\kelvin}, is that of their Fig.~12. 
\emph{Self-consistency} is agreement of the trained model with the reference model that generated its record. 
It also bounds the learning rather than the physics. \emph{Experimental agreement} appears only in
Section~\ref{sec:rc19-agreement}.

\begin{table}[htbp]
  \centering
  \caption{The source of a quantity, and the term this report uses for it.}
  \label{tab:vocabulary}
  \small
  \begin{tabular}{p{0.27\textwidth}p{0.30\textwidth}p{0.33\textwidth}}
    \toprule
    Term & What produced it & Example \\
    \midrule
    experimental measurement & the wind tunnel & the \SI{236}{\hertz} response frequency \\
    derived from experimental measurements & a calculation from quantities the tunnel measured & the adiabatic
      wall temperature from the Kulite pressures; the pressure ratio of $1.521$ from the pressure-sensitive paint \\
    result of Brouwer et al. & the reduced-order model of~\cite{brouwer2021} & its regimes; its \SI{113}{\hertz} peak \\
    simulation result & the computational fluid dynamics of~\cite{brouwer2021} & the \SI{377}{\kelvin} wall
      temperature \\
    reference model result & the reduced physics model of Section~\ref{sec:methods}, not trained & the aerodynamic
      sensitivity; the assimilation twin experiment \\
    trained-model prediction & the learned model of Section~\ref{sec:methods-identity} & the ensemble frequency of
      \SI{238.4}{\hertz} \\
    \bottomrule
  \end{tabular}
\end{table}

\subsection{HyMAX protocol}
\label{sec:setup-hymax}

The HyMAX study compares the trained model's trailing-edge displacement with one member of the six-member HyMAX record
(Section~\ref{sec:methods-record}). That member has a back pressure $p_b = \SI{549}{\pascal}$ and the shock at
$x_s = \SI{117}{\milli\meter}$ from the hinge. The trained model runs in free rollout from recorded time sample $100$.
No recorded state is fed back after it. Over a window of $N_w$ time samples after the start, the relative root-mean-square
error is
\begin{equation}
  e_{\mathrm{rms}} \;=\; \left[ \frac{\sum_{k=1}^{N_w} \big( w_{\mathrm{te}}(t_k) - w_{\mathrm{te}}^{\mathrm{ref}}(t_k) \big)^2}
  {\sum_{k=1}^{N_w} w_{\mathrm{te}}^{\mathrm{ref}}(t_k)^2} \right]^{1/2} .
  \label{eq:hymax-rms}
\end{equation}
Here $w_{\mathrm{te}}(t_k)$ and $w_{\mathrm{te}}^{\mathrm{ref}}(t_k)$ are the trailing-edge displacements of the trained
model and of the record member at time sample $k$. The error is quoted with its window since an error without its
window is not a result. The early window spans $N_w = 48$ time samples, \SI{1.9}{\milli\second}, which is three training
windows of $16$ each. The full window spans the whole \SI{64.9}{\milli\second} rollout. The start time sample and the
early window are fixed settings of the study.

\section{RC-19: the clamped panel at Mach 1.92}
\label{sec:rc19}

\subsection{Exact structure, and self-consistency}
\label{sec:rc19-self}

The energy identity of Proposition~\ref{prop:identity} was checked on the trained weights at random states with
amplitudes from $10^{-6}$ to $10^{6}$, $24$ states per amplitude. At each state, $\mathrm{d}H/\mathrm{d}\tau$ by automatic
differentiation is compared with the right side of Equation~\eqref{eq:identity}. The relative residual is their
difference divided by the largest of the terms that cancel. The worst relative residual is $2.2\times10^{-15}$. The individual terms span about forty-five decades over that amplitude range. These are measurements of roundoff, not of how well a condition was
learned. Proposition~\ref{prop:avf} states the discrete energy balance of one AVF step. We checked it along eighty
reference model trajectories at the simulation and the experimental operating conditions. The worst residual per
step, divided by the energy $H$ at that step, is $3.8\times10^{-11}$. The learned dissipation was positive definite at every training epoch of every model, 
as Lemma~\ref{lem:spd} requires. Figure~\ref{fig:energy} shows both checks.

\begin{figure}[htbp]
  \centering
  \includegraphics[width=\textwidth]{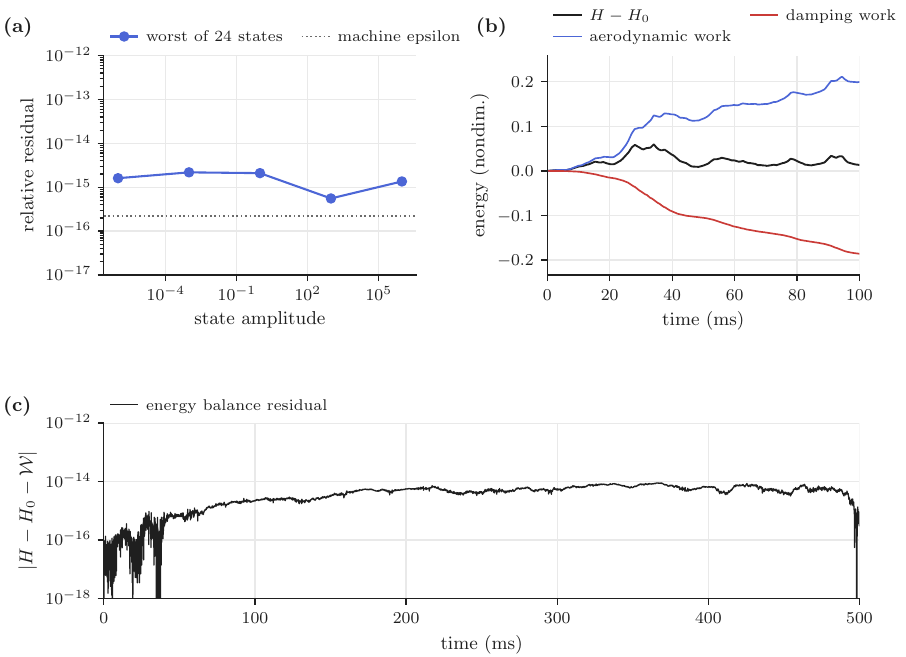}
  \caption{The energy accounting, measured. (a) The worst relative residual of the energy identity,
  Equation~\eqref{eq:identity}, on the trained weights, against state amplitude over twelve decades. Each point is the worst
  of $24$ random states. (b) Along one record member of the reference model, the change in energy $H - H_0$ and the work
  done by the aerodynamic force and the damping. (c) The closure $|H - H_0 - \mathcal{W}|$ at every step of the record, where $\mathcal{W}$ is the total work of the
  two forces.}
    \label{fig:energy}
\end{figure}

\begin{figure}[htbp]
  \centering
  \includegraphics[width=\textwidth]{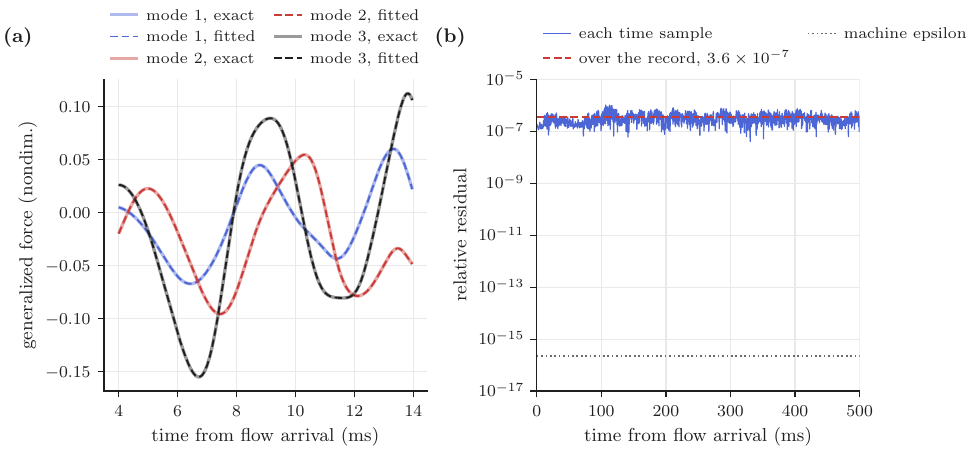}
  \caption{The fitted force term of the periodic-case trained model against the reference model's exact non-conservative
  force, $-(\mathbf{A}\mathbf{q} + (\mathbf{B} + \mathbf{C})\mathbf{p})$. Both are evaluated on the states of the record
  member at \SI{50.43}{\kilo\pascal}. (a) The first three of the seven modes over \SI{10}{\milli\second}. Solid lines are
  exact and dashed lines fitted. (b) The residual norm at each time sample, relative to the root-mean-square norm of the
  exact force over the record. The dashed line is the relative root-mean-square error, $3.6\times10^{-7}$ over the
  time samples where the exact force exceeds five percent of its largest norm.}
  \label{fig:force}
\end{figure}

Table~\ref{tab:self} reports the trained model against the reference model whose record it was fitted to. The fitted force term
reproduces the reference model's non-conservative operators to $3.6\times10^{-7}$ relative error. Figure~\ref{fig:force}
shows the fitted and exact forces of the first three modes and the residual at every time sample of the record. 
Its displacement block matches the aerodynamic stiffness operator $-\mathbf{A}$ to $9.6\times10^{-7}$. 
Its rate block matches the aerodynamic and structural damping, $-(\mathbf{B} + \mathbf{C})$, to $7.3\times10^{-7}$. 
The fitted term thus carries the structural damping as well as the aerodynamic force. The trained dissipation
therefore has nothing left to carry. Its eigenvalues settle at $2\times10^{-8}$ to $10^{-7}$, against $0.010$ to $0.020$
for the reference model's damping in the same units. Training confirms the fit on this record rather than improving it. On a
record whose non-conservative forces are not linear in the state, the trained dissipation would carry what the fitted term
cannot. These records do not test that.

\begin{table}[htbp]
  \centering
  \caption{Self-consistency: the trained RC-19 model against the reference model that generated its record. These are not
  comparisons with experimental measurements, which are in Section~\ref{sec:rc19-agreement}. The last three statistics
  are taken on one record member, at \SI{50.43}{\kilo\pascal}, at the nine extraction points. They cover one free rollout of
  $48$ time samples, \SI{1.9}{\milli\second}. The band error uses the log-power measure of Table~\ref{tab:estimators}, with the spectral window set to
  the whole $48$-sample rollout, and is the largest over the points. The amplitude ratio is the smallest and largest over the points. The mean error is the largest
  over the points.}
  \label{tab:self}
  \begin{tabular}{p{0.52\textwidth}c}
    \toprule
    Quantity & Measured \\
    \midrule
    Energy identity on the trained weights, worst relative residual & $2.2\times10^{-15}$ \\
    Fitted non-conservative forces, relative error & $3.6\times10^{-7}$ \\
    Spectral band error & \SI{0.16}{\percent} \\
    Amplitude, ratio of standard deviations & $0.993$ to $1.015$ \\
    Mean deflection error & $0.0009$ thicknesses \\
    Learned dissipation, positive definite & in every epoch, by construction \\
    \bottomrule
  \end{tabular}
\end{table}

Three of the reference model's own spectral quantities are inputs rather than results. 
Modes two to seven carry the measured frequencies because the basis is the measured shapes. 
The installed first frequency and the buckling temperature are the two values the reference model is calibrated on 
(Section~\ref{sec:methods-panel}). Agreement on those three is a check of the construction, not a prediction.

\subsection{Agreement with the experimental measurements}
\label{sec:rc19-agreement}

The package and the published papers support three comparisons with the experiment, and we make all three. None of the numbers below
is fitted to a response measurement.

\paragraph{The shock, from geometry alone} As in~\cite{brouwer2023pressure}, the wedge angle follows from the package's wedge height and length, and the
shock angle from its schlieren. The exact two-shock reflection~\cite{anderson2007fundamentals} then predicts a pressure ratio of $1.521$. The
measured ratio is also $1.521$. It is the ratio of two centerline medians of the package's pressure-sensitive paint,
measured on the rigid panel in the attached-shock case~\cite{brouwer2023pressure}. The medians are taken over
$x = 112$ to \SI{145}{\milli\meter} behind the pressure rise and $x = 80$ to \SI{92}{\milli\meter} ahead of it. Other
reasonable windows give $1.515$ to $1.522$. The prediction band of $1.45$ to $1.60$ comes from the package's
uncertainties. These are $\pm\SI{0.9}{\milli\meter}$ on the \SI{7.8}{\milli\meter} wedge height and $\pm\SI{0.62}{\degree}$ on
the \SI{35.7}{\degree} shock angle.

\paragraph{The wall temperature, from the package's own pressures} On the chaotic case the adiabatic wall
temperature derived from the Kulite pressures is \SI{369.4}{\kelvin}, and from the scanner pressures
\SI{369.5}{\kelvin}. The simulation result of~\cite{brouwer2021} is \SI{377}{\kelvin}. The two instruments agree with each other to \SI{0.1}{\kelvin}, and each differs from the simulation by about \SI{7.5}{\kelvin}.
That agreement, not the difference from \SI{377}{\kelvin}, is the finding. Both values are derived from experimental
measurements, and the simulation is a computation. Figure~\ref{fig:measurements} shows the shock and the wall temperature.

\begin{figure}[htbp]
  \centering
  \includegraphics[width=\textwidth]{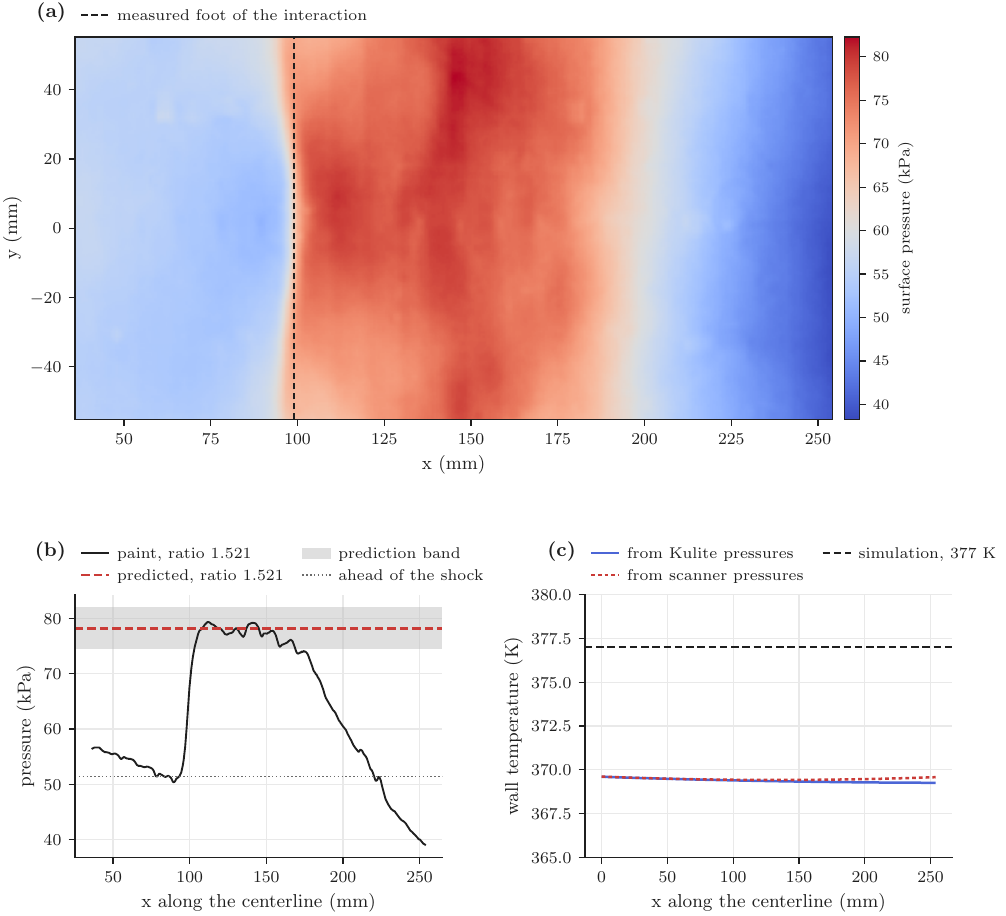}
  \caption{Two comparisons from the package's own measurements with nothing fitted. (a) The pressure-sensitive paint field
  of the attached-shock case on the rigid panel; the dashed line is the measured foot of the interaction. (b) Its
  centerline against the plateau the two-shock reflection predicts from the wedge geometry, with the band set by the package's
  uncertainties in wedge height and shock angle. The dotted line is the pressure ahead of the shock. (c) The adiabatic wall temperature along the centerline of the chaotic case.
  The temperature is derived from the Kulite and from the scanner pressures. The dashed line is the simulation result of Brouwer et al.~\cite{brouwer2021}.}
  \label{fig:measurements}
\end{figure}

\paragraph{The response frequency, as a statistic} We ran thirty-member ensembles of the trained and the reference
model at the periodic case's experimental operating condition. Table~\ref{tab:ensemble} compares their statistics with
the measured \SI{236}{\hertz}. The two mean frequencies lie within one \SI{4.88}{\hertz} spectral bin of each other.
Each lies within \SI{1.1}{\percent} of the measurement. The amplitude of Table~\ref{tab:estimators} is equally stable
in both ensembles. Its standard deviation across the members is \SI{1.6}{\percent} of its mean for each model. On the
chaotic case the reference model's thirty-member ensemble gives a mean frequency of \SI{239.6}{\hertz} with a standard
deviation of \SI{12.8}{\hertz}. That mean is \SI{6.8}{\percent} below the measured \SI{257}{\hertz} that Brouwer et al.\ report for the chaotic
case~\cite{brouwer2021}. Figure~\ref{fig:ensemble} shows
the two periodic-case ensembles.

\begin{table}[htbp]
  \centering
  \caption{Thirty-member ensembles at the periodic case's experimental operating condition.}
  \label{tab:ensemble}
  \small
  \begin{tabular}{lccc}
    \toprule
    & Trained model & Reference model & Measured \\
    \midrule
    Mean dominant frequency (\si{\hertz}) & $238.4$ & $237.2$ & $236$ \\
    Standard deviation across members (\si{\hertz}) & $12.3$ & $10.5$ & -- \\
    Difference from measured & \SI{1.0}{\percent} & \SI{0.5}{\percent} & -- \\
    Mean amplitude (thicknesses) & $0.760$ & $0.752$ & -- \\
    Amplitude standard deviation / mean & \SI{1.6}{\percent} & \SI{1.6}{\percent} & -- \\
    Amplitude full range / mean & \SI{5.7}{\percent} & \SI{6.5}{\percent} & -- \\
    \bottomrule
  \end{tabular}
\end{table}

\begin{figure}[htbp]
  \centering
  \includegraphics[width=\textwidth]{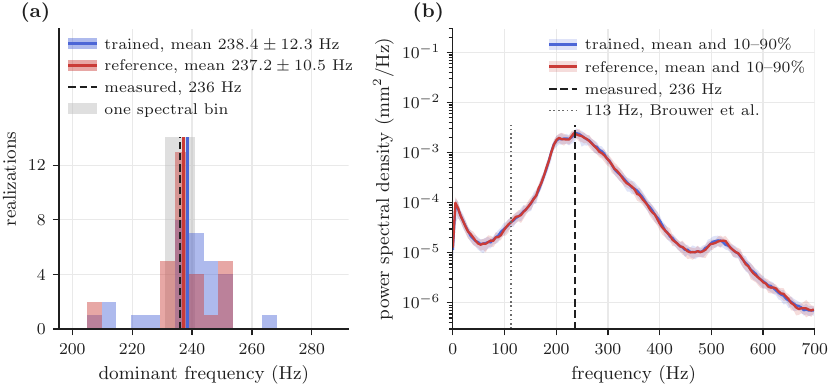}
  \caption{The thirty-member ensembles at the periodic case's experimental operating condition. (a) The dominant frequency
  of every realization, trained model against reference model, with the ensemble means as solid lines. The dashed line is
  the measured \SI{236}{\hertz}, and the shaded band one spectral bin. (b) The mean spectrum of each ensemble at
  three-quarter chord, with the 10th to 90th percentile band. The dotted line marks the \SI{113}{\hertz} peak of the
  reduced-order model of Brouwer et al., which neither ensemble shows.}
  \label{fig:ensemble}
\end{figure}

These ensemble means depend on the interval of kicks that the ensemble samples. A kick is the first-mode displacement
added to the buckled shape before release (Section~\ref{sec:methods-record}). The thirty kicks of each ensemble span
$\pm\SI{1.8}{\percent}$ about the set's median kick. The aerodynamic study of Section~\ref{sec:rc19-aero} measured this
dependence with a larger ensemble of $192$ members. That ensemble uses unscaled aerodynamic operators and releases every
member from the upper well. Its kicks span $\pm\SI{11.9}{\percent}$. Its first thirty members are the reference ensemble
of Table~\ref{tab:ensemble}, and they give a mean frequency of \SI{237.19}{\hertz}. All $192$ give \SI{233.34}{\hertz}. The difference could come from the sample size or from the
wider interval. Splitting the $192$ members by interval separates the two. The $31$ members with kicks inside
$\pm\SI{1.8}{\percent}$ give \SI{237.20}{\hertz}, and the $161$ outside it give \SI{232.60}{\hertz}. The mean therefore
follows the interval, not the sample size. An ensemble mean must therefore be quoted with its kick interval.

\subsection{The regime map at the simulation operating conditions of Brouwer et al.}
\label{sec:rc19-regime}

The reference model, at the experimental operating condition of the periodic case, returns a chaotic response. The reduced-order model of Brouwer et al. returns one as well, and it needed adjusted simulation operating conditions to place its periodic case.
We therefore asked whether our two models reproduce its regime map when given its own inputs. The prediction, registered before the run, was periodic at a cavity pressure of \SI{52.7}{\kilo\pascal} and a temperature rise of \SI{8.8}{\kelvin}. It was
chaotic at \SI{51.1}{\kilo\pascal} and \SI{13.6}{\kelvin}. The periodic condition is the single-well response of
Fig.~11a of~\cite{brouwer2021}. Their periodic companion to the experiment, in their Fig.~13, uses \SI{7.1}{\kelvin}
instead. The chaotic condition is that of their Fig.~12.

\begin{table}[htbp]
  \centering
  \caption{The regime at the simulation operating conditions of Brouwer et al.~\cite{brouwer2021}. Each entry is a repeat set of sixteen
  releases, eight from each buckled well. The trained model was trained at a temperature rise of \SI{12.76}{\kelvin}. Both conditions therefore
  lie outside its training temperature.}
  \label{tab:regime}
  \begin{tabular}{lcccc}
    \toprule
    Condition & Brouwer et al. & Reference & Trained & Agreement \\
              & \cite{brouwer2021} & model & model & \\
    \midrule
    \SI{52.7}{\kilo\pascal}, \SI{8.8}{\kelvin}  & periodic & periodic & periodic & 16 of 16 \\
    \SI{51.1}{\kilo\pascal}, \SI{13.6}{\kelvin} & chaotic  & chaotic  & chaotic  & 16 of 16 \\
    \bottomrule
  \end{tabular}
\end{table}

\begin{figure}[htbp]
  \centering
  \includegraphics[width=\textwidth]{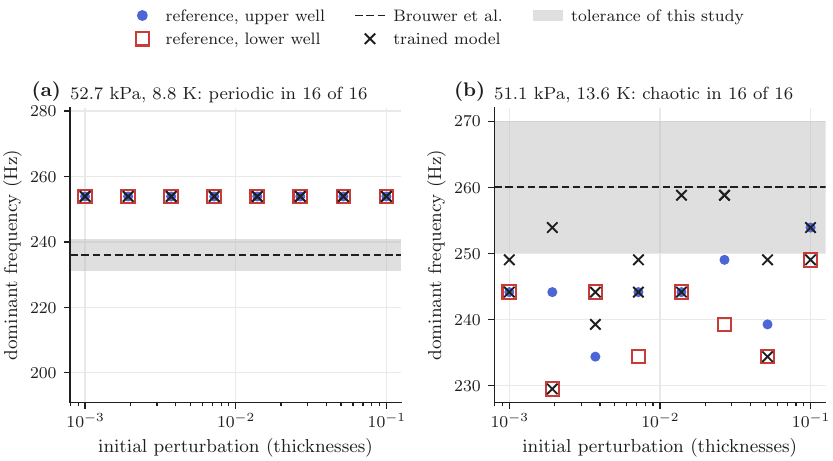}
  \caption{The regime at the simulation operating conditions of Brouwer et al., on the reference model. Each panel holds sixteen repeats, kicked
  over two decades and released from both buckled wells. (a) At the periodic simulation operating condition every repeat returns
  one frequency, and the two wells coincide. (b) At the chaotic simulation operating condition the frequencies scatter, which is
  the response rather than numerical noise. The dashed line is the measured dominant frequency reported by Brouwer et al.~\cite{brouwer2021}.
  It is \SI{236}{\hertz} for the periodic case. For the chaotic case it is approximately \SI{260}{\hertz}, from the
  discussion of their Fig.~12. Their Section 5.4 gives \SI{257}{\hertz} for the same case, which lies inside the band.
  The shaded band is a tolerance set in this study. It is $\pm\SI{4.88}{\hertz}$, one bin of the frequency grid, for
  the sharp periodic peak. It is $\pm\SI{10}{\hertz}$ for the broad chaotic one. Crosses mark the trained model at the same kicks. It returns the same regimes in every repeat
  (Table~\ref{tab:regime}).}
  \label{fig:regime}
\end{figure}

The prediction held on both the reference and the trained model, at both conditions, from both wells, in every repeat (Table~\ref{tab:regime},
Figure~\ref{fig:regime}). At the periodic simulation operating condition, every repeat of the trained model returns the same \SI{253.9}{\hertz} bin,
as the reference model does. The amplitude of Table~\ref{tab:estimators} is $0.7164$ thicknesses for the trained model and $0.7161$ for the
reference model, which agree to \SI{0.04}{\percent}. At the chaotic simulation operating condition, the trained model's median frequency over its sixteen repeats is
\SI{249.0}{\hertz}. The reference model that generated its record, with the periodic case's thermal field, gives \SI{239.3}{\hertz}.
Brouwer et al.~\cite{brouwer2021} report a peak at \SI{113}{\hertz} in their reduced-order model that the experiment does
not show. Neither of our models shows it in any repeat or in any of the thirty ensemble realizations of
Section~\ref{sec:rc19-agreement}.

What this settles is where the regime disagreement at the experimental operating condition lives. With the inputs of Brouwer et al., our models return the regimes of their model. The disagreement is therefore a property of the inputs, not of the aerodynamic model. The inputs in question are the
distribution of the temperature rise and the pressure the static load is referenced to. The reduced-order model of Brouwer et al.\ has the same disagreement. Brouwer et al.~\cite{brouwer2021} report the
offsets of their periodic companion simulation at \SI{7.1}{\kelvin}. It sits \SI{4.6}{\percent} off the experiment in
cavity pressure and \SI{44.5}{\percent} off in temperature rise.

One comparison with the experiment at the simulation operating conditions does not agree. Our frequency at the periodic
simulation operating condition is \SI{253.9}{\hertz}, which is \SI{7.6}{\percent} above the measured \SI{236}{\hertz}. Our peak
amplitude at the simulation operating conditions is $1.54$ to $1.95$ thicknesses. The measured oscillation amplitude at
three-quarter chord reaches about $1.8$ thicknesses in Fig.~7a of~\cite{brouwer2021}. Those measurements span the
experimental runs rather than the simulation operating conditions. Their amplitude measure may also differ from our peak
amplitude, and the comparison is therefore approximate. Brouwer et al.\ report that their periodic companion misses
the measured amplitude by \SI{31.7}{\percent}.

\subsection{What the aerodynamic model contributes}
\label{sec:rc19-aero}

The reference model's aerodynamic force is first-order piston theory with freestream coefficients. The enrichment that Brouwer et al.\ apply replaces freestream quantities by local ones. That scales the
aerodynamic operators by the local to freestream pressure ratio. Brouwer et al.\ give a mean RANS surface pressure of
\SI{51.0}{\kilo\pascal} against a freestream pressure of \SI{48.4}{\kilo\pascal}~\cite{brouwer2021}, a ratio of
$51.0/48.4 = 1.054$. Rather than build the enriched model, we measured what that scaling changes.

The first registered prediction was that the scaling leaves the regime unchanged and moves the dominant frequency by
less than \SI{2}{\percent}. At the periodic simulation operating condition, it moved the frequency by \SI{0.00}{\hertz}. 
At the experimental operating condition, where the response is chaotic, the median frequency moved by $+3.23$ and
$-3.00$~\SI{}{\percent} in the two wells. The condition failed. It set a tolerance of \SI{4.7}{\hertz}, two percent
of \SI{236}{\hertz}, on a frequency that the estimator quantizes at \SI{4.88}{\hertz}. Across repeats that frequency has
a standard deviation of \SI{16.3}{\hertz}. The two wells also moved in opposite directions, which a change in the
aerodynamic coefficients cannot cause. The failure therefore lies in the condition rather than the physics
(Section~\ref{sec:predictions}). The replacement, R3, was registered before its run. It has two clauses: the mean frequency
is equivalent within two percent, and the spectral change stays within sampling.

The mean-frequency clause uses an equivalence test, because a difference test cannot support a claim of no effect.
It applies two one-sided tests~\cite{schuirmann1987tost} to each well. The \SI{95}{\percent} interval for the difference
of the two ensemble means must lie inside $\pm\delta_{\mathrm{eq}}$. The tolerance $\delta_{\mathrm{eq}}$ is two percent
of the well's unscaled ensemble mean, which gives about \SI{4.7}{\hertz}. This two percent was registered before the
runs without reference to the frequency's sampling noise. Section~\ref{sec:predictions} reports what that cost.

The ensemble size gives the test \SI{80}{\percent} power. It uses the pooled standard deviation
$s_{\mathrm{pool}} = \SI{16.3}{\hertz}$ across repeats. It also uses $\delta_{\mathrm{eq}} = \SI{4.72}{\hertz}$, the
tolerance of the well whose unscaled mean is \SI{236.12}{\hertz}. Each configuration then needs
\begin{equation}
  N_e \;\ge\; \frac{2\, s_{\mathrm{pool}}^2 \big( z^\ast_{0.975} + z^\ast_{0.80} \big)^2}{\delta_{\mathrm{eq}}^2} \;=\; 187.2
\end{equation}
members, where $z^\ast_{0.975} = 1.96$ and $z^\ast_{0.80} = 0.84$ are standard normal quantiles. The minimum is
therefore $188$ members, and each ensemble used $192$.

The spectral clause tests the change in the spectrum against sampling, using the band error of
Table~\ref{tab:estimators}. The band error between the mean spectra of the two configurations must not exceed that
between the odd and even members of one configuration.

Four ensembles of $192$ reference model members were run at the experimental operating condition of the periodic case,
unscaled and scaled, released from each well. The reference model is the right model for this study, because its aerodynamic operators can be scaled directly. The
trained model's fitted force term matches them to relative errors of $9.6\times10^{-7}$ and $7.3\times10^{-7}$
(Section~\ref{sec:rc19-self}).

\begin{figure}[htbp]
  \centering
  \includegraphics[width=\textwidth]{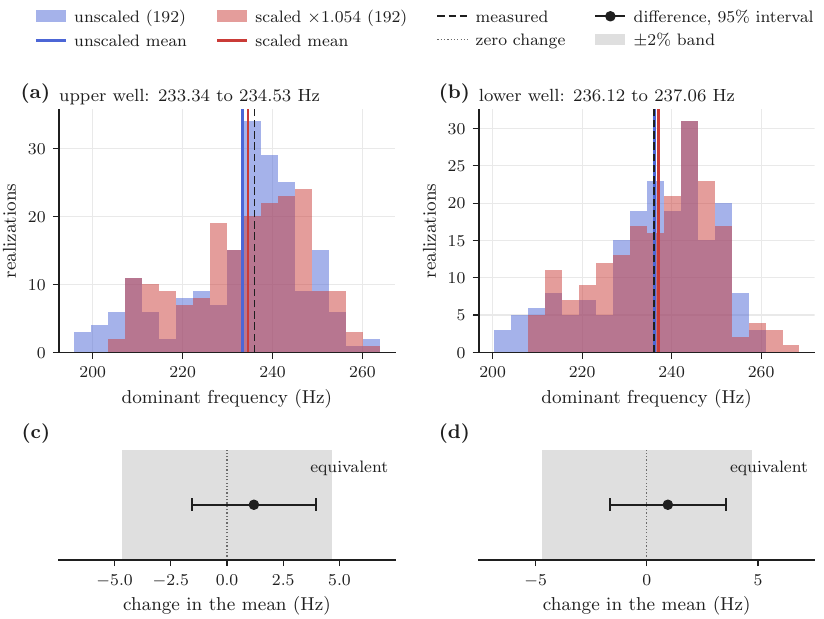}
  \caption{Equivalence of the mean response frequency under the enrichment, on the reference model. (a, b) The dominant
  frequency over $192$ realizations in each configuration, per buckled well. Solid lines are the ensemble means, and
  the dashed line is the measured value. (c, d) The difference of the means with its \SI{95}{\percent} interval,
  inside the registered equivalence band. The ensemble means from the two wells bracket the measured value.}
  \label{fig:equivalence}
\end{figure}
\begin{figure}[htbp]
  \centering
  \includegraphics[width=\textwidth]{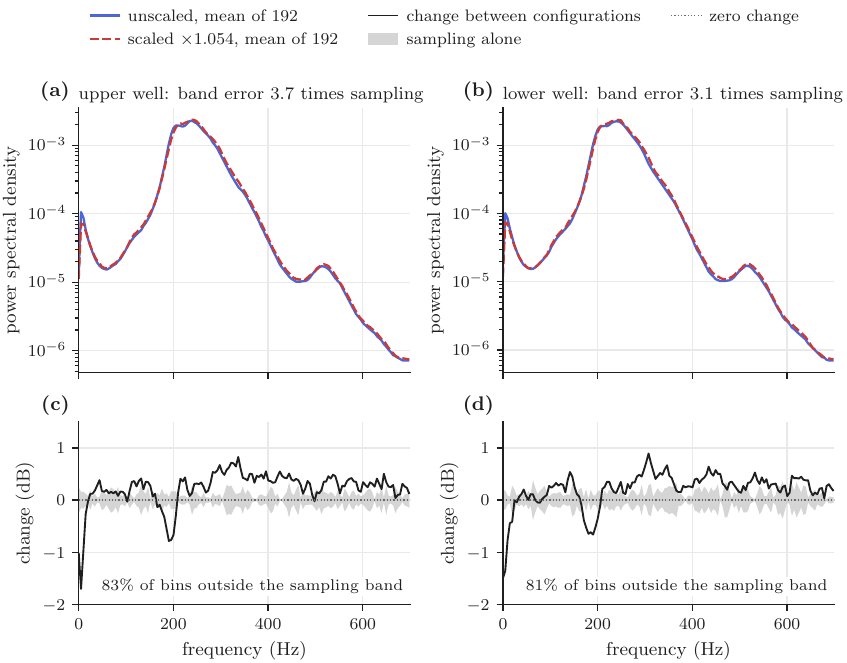}
  \caption{Aerodynamic sensitivity at the experimental operating condition, from four ensembles of $192$ reference model
  members. (a, b) The mean spectrum per buckled well, with the aerodynamic operators unscaled and scaled by the local to
  freestream pressure ratio. (c, d) The change between the two configurations, against the difference between the odd and even members of one configuration, which is sampling alone.}
  \label{fig:spectra}
\end{figure}

The mean frequency is equivalent within \SI{2}{\percent} in both wells. In the lower well the shift is \SI{+0.40}{\percent}, with an interval of $[-1.67, +3.56]$~\SI{}{\hertz}. In the
upper well it is \SI{+0.51}{\percent}, with $[-1.55, +3.94]$~\SI{}{\hertz} (Figure~\ref{fig:equivalence}). Each band is two
percent of that well's unscaled mean: $\pm\SI{4.72}{\hertz}$ of \SI{236.12}{\hertz} and $\pm\SI{4.67}{\hertz}$ of
\SI{233.34}{\hertz}. The movements of $+3.23$ and $-3.00$~\SI{}{\percent} under the first prediction were therefore sampling. The spectrum, however, does respond. The band error between configurations is $3.09$ and $3.74$ times the band error between odd and even members (Figure~\ref{fig:spectra}). The amplitude rises by $2.31$ and \SI{2.55}{\percent}, outside its own interval in both
wells. The criterion required both the frequency test and the spectrum test to pass, and it is not met. The two tests disagree,
and the disagreement is the result.

The local-coefficient enrichment of Brouwer et al.\ has a small but measurable effect. A \SI{5.4}{\percent} change in
the aerodynamic coefficients moves the amplitude by \SI{2.4}{\percent}, which is an elasticity of $2.4/5.4 = 0.44$. It does
not move the mean frequency measurably. The enrichment is therefore a percent-level refinement of accuracy. It does not
set the regime label, and Section~\ref{sec:rc19-open} shows that it does not set the fluctuation level.

\subsection{Chaos and the tracking horizon}
\label{sec:rc19-chaos}

\paragraph{The exponent} Brouwer et al.~\cite{brouwer2021} establish the chaos with a positive largest Lyapunov
exponent that rises with temperature. Figure~\ref{fig:lyapunov} shows the sweep described below.

\begin{figure}[htbp]
  \centering
  \includegraphics[width=\textwidth]{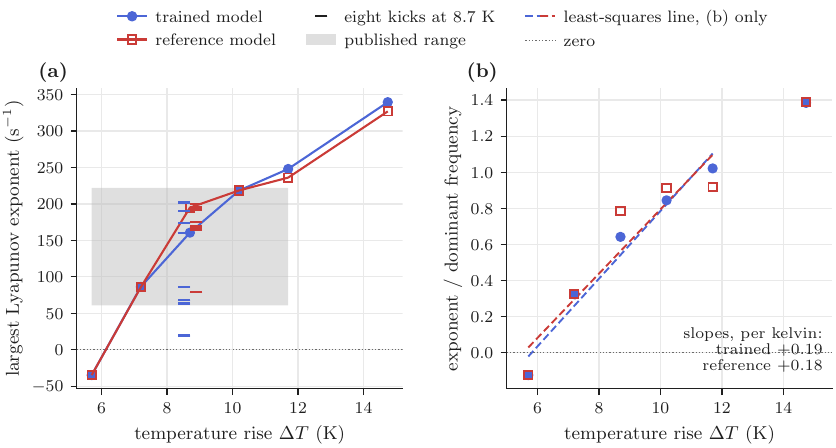}
  \caption{The largest Lyapunov exponent over temperature at \SI{51.1}{\kilo\pascal}, one trajectory per temperature, trained
  model against the reference model that generated its record. (a) The exponent, with the range Brouwer et al.\ report
  over their five temperatures, \SI{5.7}{} to \SI{11.7}{\kelvin}. The short bars at \SI{8.7}{\kelvin} are eight trajectories per
  model, kicked logarithmically from $10^{-3}$ to $10^{-1}$ thicknesses. (b) The exponent divided by each model's dominant frequency, with least-squares lines over
  the five published temperatures.}
  \label{fig:lyapunov}
\end{figure}
 We march the linearized error equations of Brouwer et al.\ (Equations 19 to 23 of~\cite{brouwer2021}) along the
trained model's own trajectories. The linearization is the trained model's own. The result reproduces that behavior. The sweep runs at \SI{51.1}{\kilo\pascal} and six temperatures from \SI{5.7}{} to \SI{14.74}{\kelvin}. There the trained model and the reference model that generated its record, with the periodic case's thermal field, agree in sign at all six. Both are negative at \SI{5.7}{\kelvin}, where each gives \SI{-34.9}{\per\second}. At $7.2$, $10.2$, $11.7$
and \SI{14.74}{\kelvin} their difference, relative to the reference model, is at most \SI{5.1}{\percent}. At \SI{8.7}{\kelvin} they differ by \SI{17.6}{\percent},
\SI{160.6}{} against \SI{194.9}{\per\second}, which broke the registered tolerance of \SI{15}{\percent}.

That temperature was then examined by a measurement registered before its run. Eight trajectories were run on each model at \SI{8.7}{\kelvin}, with kicks spaced logarithmically from $10^{-3}$ to $10^{-1}$ thicknesses. They give exponents from \SI{19.6}{}
to \SI{201.9}{\per\second} on the trained model and \SI{79.2}{} to \SI{199.9}{\per\second} on the reference model. The ranges overlap, each is wider than
\SI{17.6}{\percent} of its median, and a Mann--Whitney rank test~\cite{mann1947test} does not distinguish them (two-sided $p$-value $0.19$). At this temperature a single
trajectory does not determine the exponent on either level. The broken comparison says nothing against the trained model, and
the prediction stays broken as registered. The trained model's median, \SI{123.5}{\per\second}, sits below the reference model's,
\SI{172.5}{\per\second}. With eight trajectories per model that difference is not resolved. If it is real, the integrators are a candidate
cause. The trained model uses the generalized-$\alpha$ step, which damps high frequencies numerically. The reference
model uses the energy-conserving AVF step. Numerical damping could lower the growth rate of small separations, and with
it the exponent. We did not run the larger set of trajectories that would test this.

The same caution applies to the trend. Brouwer et al.\ report exponents at five temperatures, $5.7$ to
\SI{11.7}{\kelvin} at \SI{51.1}{\kilo\pascal}. Over those five we divide the exponent by each model's own dominant frequency. That ratio rises
$0.19$ per kelvin on the trained model and $0.18$ on the reference model. On both, the exponents at $7.2$,
$8.7$ and \SI{10.2}{\kelvin} fall inside the published range of $61$ to \SI{222}{\per\second}. Each value rests on one
trajectory per temperature and carries the spread just measured. The slope and the count are therefore indicative, not
established.

\paragraph{The horizon} Chaos limits every model of this panel, ours included. A model started at a separation $\varepsilon_0$ from the truth separates as $e^{\lambda t}$ until it reaches the state
scale $\sigma$. The state scale is the size of the response itself, at which the separation saturates. The tracking horizon is therefore
\begin{equation}
  T_h \;=\; \frac{1}{\lambda} \ln\!\left( \frac{\sigma}{\varepsilon_0} \right).
  \label{eq:horizon}
\end{equation}
Starting accuracy enters through a logarithm. At \SI{51.1}{\kilo\pascal} and the chaotic case's \SI{14.74}{\kelvin}, the reference model with the chaotic case's
thermal field has an exponent of \SI{349}{\per\second}. The sweep above, with the periodic case's thermal field, gives
\SI{327}{\per\second} at the same temperature. One decade of starting accuracy then buys $\ln 10 / \lambda = \SI{6.6}{\milli\second}$. A model exact to
double-precision roundoff starts near $\varepsilon_0 = 10^{-16}$. The state scale is about one thickness. Here and below, tracking ends at half of it, $\sigma/2 = 0.5$. With
$\sigma/2$ in place of $\sigma$, that model would track $\ln(0.5/10^{-16})/\lambda \approx \SI{104}{\milli\second}$. That is the
ceiling for any model of this panel computed in that arithmetic.

The bound needs no learned model. Its rate $\lambda$ is measured on the reference model alone, and its starting error
$\varepsilon_0$ is whatever a model achieves. A model that starts at $\varepsilon_0 = 10^{-2}$ tracks
$\ln(0.5/10^{-2})/\lambda = \SI{11.2}{\milli\second}$. At $\varepsilon_0 = 10^{-5}$ it tracks
$\ln(0.5/10^{-5})/\lambda = \SI{31.0}{\milli\second}$. Each further decade of starting accuracy adds only
\SI{6.6}{\milli\second}. Improving the model therefore extends the horizon only slowly.

The horizon is why the comparisons of this report are statistics, which are reproducible, rather than the peak of a single
realization, which is not. Brouwer et al.~\cite{brouwer2021} grade their own model the same way. Their published exponent comes from their reduced
model's trajectories. Their model therefore loses a trajectory on the same clock. Tracking beyond the horizon cannot come from model quality or integrator quality. It needs measurement.

\subsection{Tracking by assimilation}
\label{sec:rc19-assimilation}

The assimilation of Section~\ref{sec:methods-observer} was tested in a twin experiment. In a twin experiment the model
is exact and only its state is wrong. The test therefore measures what the observations can support, not what the
model's physics lacks.

The truth is one reference model member at the chaotic case's experimental operating condition of
\SI{51.49}{\kilo\pascal} and \SI{14.74}{\kelvin}. It starts in the upper well with a kick of \num{e-3} thicknesses, runs
\SI{0.12}{\second}, and is sampled at \SI{25}{\kilo\hertz}. A copy of the reference model starts with its first-mode
displacement \num{e-3} thicknesses away from the truth's. The copy is corrected with the truth's displacements at the
nine published extraction points. It takes two discrete-gradient steps per time sample. The correction of
Equation~\eqref{eq:nudging} acts like springs and dampers that tie the copy to the truth at the nine points. The gain
$k_p$ is the spring stiffness, acting on the displacement mismatch. The gain $k_d$ is the damping coefficient, acting on
the mismatch in its rate. The tracking runs use $k_p = \SI{e5}{\newton\per\meter}$ and
$k_d = \SI{5e3}{\newton\second\per\meter}$. The gain study uses a \SI{0.2}{\second} member at
the same condition.

The energy budget of Corollary~\ref{cor:observer} closes under assimilation, with the observer's work reported as a
separate term. The worst per-step defect, divided by the larger of the step's two energies, is $3\times10^{-13}$.

\paragraph{Observability} The nine published points observe all seven modes. The observability Gramian $\mathbf{M}_\Phi$ of Equation~\eqref{eq:nudging} has a condition number, the ratio of its largest to its smallest eigenvalue, of $45$. The least visible combination of modes therefore registers
about $\sqrt{45} \approx 6.7$ times more weakly in amplitude than the most visible one. Amplitude scales as the square
root of a Gramian eigenvalue. No mode has a node at all nine points, and recovering the weakest one does not
demand a large amplification of noise. That was not assured in advance, because the points were chosen for imaging
rather than for observability.

\paragraph{Gains} The open-loop exponent is the largest Lyapunov exponent of the error
equation~\eqref{eq:error} with both gains set to zero. On the trajectory of this twin experiment it is
\SI{285}{\per\second}. The \SI{349}{\per\second} of Section~\ref{sec:rc19-chaos} comes from another trajectory, at
\SI{51.1}{\kilo\pascal} with the chaotic case's thermal field. The difference lies within the single-trajectory scatter
measured there.

Both gains act on the unit-mass modal coordinates. Position feedback alone does not contract the error. With $k_d = 0$, no position gain up to
$k_p = 10^{14}$ did. Rate feedback does, but more of it is not better. The smallest position gain tested is
$k_p = 10^{2}$, and at both rate gains below the error already contracts there. Both closed-loop exponents are
therefore quoted at $k_p = 10^{2}$. The exponent is \SI{-10.6}{\per\second} at $k_d = 2\times10^{3}$ and only
\SI{-2.6}{\per\second} at $k_d = 2\times10^{4}$. A negative exponent means the error shrinks rather than grows. With
feedback the copy therefore stays on the truth for as long as observations arrive, instead of leaving it after a
horizon.

A one-mode argument explains why. At $k_p = 10^{2}$ the position feedback is negligible. The largest eigenvalue of
$k_p \mathbf{M}_\Phi$ is $10^{2} \times 195 = 2.0\times10^{4}$. That is under one percent of the flat panel's smallest
squared natural frequency $\omega_1^2 = 2.7\times10^{6}\,\mathrm{s}^{-2}$. The error equation~\eqref{eq:error}
therefore keeps the panel's own stiffness $\nabla^2 U + \mathbf{A}$. Consider one overdamped mode, and let
$K_\delta$ be its stiffness. The roots of $s_r^2 + k_d s_r + K_\delta = 0$ are then real, and the slower root is about
$-K_\delta/k_d$. Raising $k_d$ therefore slows the decay. The measured rates differ by a factor of $4.1$ rather than the
predicted ten. The one-mode picture is therefore qualitative.

\paragraph{Tracking against observation rate} Without observations the copy leaves the truth after about
$\ln(0.5/10^{-3})/(\SI{285}{\per\second}) \approx \SI{22}{\milli\second}$. The late tracking error is taken over the
second half of the \SI{0.12}{\second} run, from $60$ to \SI{120}{\milli\second}, which lies entirely past that horizon.
It is the median of the Euclidean norm of the modal displacement error divided by $w_s$, in thicknesses.
Table~\ref{tab:assimilation} and Figure~\ref{fig:observability} give it.

\begin{figure}[htbp]
  \centering
  \includegraphics[width=\textwidth]{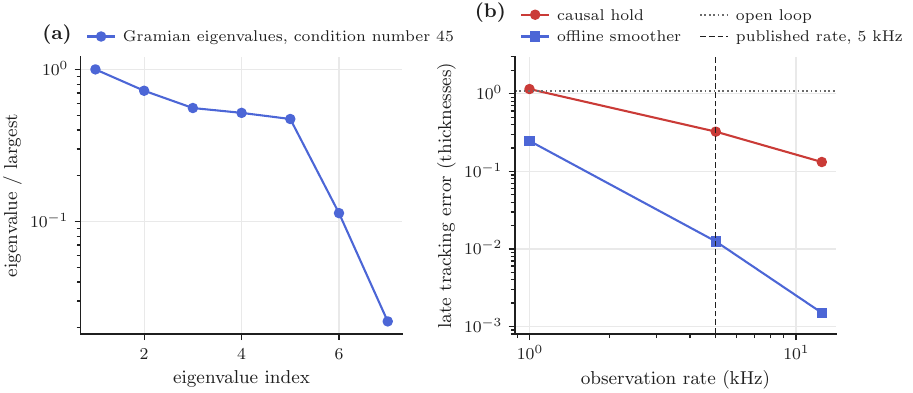}
  \caption{Assimilation through the nine published points. (a) The eigenvalues of the observability Gramian $\mathbf{M}_\Phi$,
  relative to the largest; their ratio is the condition number. (b) The late tracking error of the twin experiment against
  observation rate, with the causal hold and the offline smoother. The dotted line is the open-loop error, and the dashed
  line the published \SI{5}{\kilo\hertz} rate.}
  \label{fig:observability}
\end{figure}
 Open loop it
saturates at $1.09$ thicknesses, which is the size of the response itself. With the causal hold, the error falls linearly with the observation interval and is insensitive to the gain over three
decades. That is the signature of the hold's reconstruction, not of the information in the data. With the offline smoother it falls roughly as the square of the interval. This rate is the order of the smoother's reconstruction. At the experiment's \SI{5}{\kilo\hertz} measurement rate, the smoother holds the state to
$0.0125$ thicknesses, twenty-six times smaller than the hold's $0.325$.

\begin{table}[htbp]
  \centering
  \caption{Late tracking error against observation rate, in panel thicknesses, in a twin experiment on the reference model,
  observed at the nine published points. The experiment's measurement rate is \SI{5}{\kilo\hertz}. Open loop the
  error is $1.09$. Samples per cycle divide the observation rate by \SI{240}{\hertz}, the reference model's mean
  dominant frequency on the chaotic case, \SI{239.6}{\hertz}, rounded.}
  \label{tab:assimilation}
  \begin{tabular}{lccc}
    \toprule
    Observation rate & Samples per cycle & Causal hold & Offline smoother \\
    \midrule
    \SI{1}{\kilo\hertz}                     & $4.2$  & $1.151$ & $0.247$ \\
    \SI{5}{\kilo\hertz}, the published rate & $20.8$ & $0.325$ & $\mathbf{0.0125}$ \\
    \SI{12.5}{\kilo\hertz}                  & $52$   & $0.132$ & $0.0015$ \\
    \bottomrule
  \end{tabular}
\end{table}

\paragraph{Noise} At \SI{5}{\kilo\hertz} with the smoother the error is $0.0125$ thicknesses without noise, $0.0130$ with
one micron of noise on each reading, and $0.0352$ with ten. The noise is Gaussian, independent at each point and each
sample. Ten microns is $0.016$ thicknesses of the \SI{0.635}{\milli\meter} panel. Once noise dominates, the error settles at
$0.0352/0.016 \approx 2.2$ times the noise level. That is what a linear observer with a well-conditioned Gramian should give.

This result is a best case, not a validation. It assumes an exact model and observations generated by that model. The
published layout and rate, with an offline smoother, hold the reference model's state to $0.0130$ thicknesses with
\SI{1}{\micro\meter} of measurement noise. With \SI{10}{\micro\meter} the error is $0.0352$. Measured data would add the model's
own error to these values. The result also states what future data must provide for assimilation: displacement time
histories at the nine points, at \SI{5}{\kilo\hertz}, with noise near one micron. The public package holds no response time series. The published statistics come from Brouwer et
al.~\cite{brouwer2021} and Piccolo Serafim et al.~\cite{piccoloserafim2025}.

\subsection{The open comparison, and the measurement that closes it}
\label{sec:rc19-open}

One comparison with the experiment remains open on the attached-shock case. The measured deflection at three-quarter
chord on the centerline has a standard deviation of $0.943$ thicknesses. We digitized it from a figure
of~\cite{piccoloserafim2025} that plots the mean deflection and the lines one standard deviation above and below it. At
three-quarter chord those three lines read $+0.096$, $+1.035$ and $-0.850$ thicknesses. Half the spread between the
outer two is $(1.035 + 0.850)/2 = 0.943$.

The reference model is driven by a random surface pressure at the measured fluctuation level. That level is
$0.0096\,q_\infty$ in root-mean-square, where $q_\infty$ is the freestream dynamic pressure. It is the Kulite reading
beneath the incident shock in Fig.~16(b) of~\cite{brouwer2023pressure}. The random pressure is Gaussian, flat from
\SI{20}{} to \SI{2000}{\hertz}, and correlated over the boundary-layer thickness of \SI{8.6}{\milli\meter}. Each run
discards its first \SI{0.25}{\second} to remove the decay from the initial state and keeps the next \SI{2}{\second}.
Over three random seeds, the model's deflection at the same station has a standard deviation of $0.057$ to $0.062$
thicknesses. That is about a sixteenth of the measured $0.943$. Figure~\ref{fig:attached} shows the mean shape and this
deviation.

The published statistics do not say how the panel moves. Switching between two buckled shapes, a limit cycle, and
broadband motion about one shape all reproduce the published mean and deviation. The factor of sixteen therefore does
not depend on which of these motions the panel performs. No learned model was built for this case. The gap lies in the
reference model, upstream of any learning.

Every parametric candidate cause inside the physically admissible range was eliminated by measurement
(Table~\ref{tab:eliminated}). One structural mechanism remains, and it is stated below.
One of them, a learned conservative correction to the potential, was registered as a prediction, run, and falsified.
The aerodynamic enrichment of Section~\ref{sec:rc19-aero} is bounded with the others. Amplitude scales as the coefficients to the power $0.44$. Closing a sixteenfold deficit would therefore need the
coefficients scaled by $16^{1/0.44} \approx 10^{2.7}$. No enrichment of piston theory provides that.

\begin{table}[htbp]
  \centering
  \caption{Candidate causes of the fluctuation deficit on the attached-shock case, and how each was eliminated.}
  \label{tab:eliminated}
  \small
  \begin{tabular}{p{0.38\textwidth}p{0.52\textwidth}}
    \toprule
    Candidate cause & How it was eliminated \\
    \midrule
    Turbulent forcing level & bounded by the facility's own unsteady measurements~\cite{brouwer2023pressure}; the response is linear in it and
      reaches a third of the measurement at five times the measured level \\
    Temperature & bounded above by the adiabatic wall \\
    Edge stiffness, membrane stiffness, and the modal basis & swept across the admissible range; none closes the
      deficit \\
    Shock position, and its dependence on the panel's deflection & each thickness of deflection moves the measured shock
      footprint \SI{0.88}{\milli\meter} upstream; a self-consistent solve with that shift still finds one equilibrium \\
    Shock strength, and the same dependence & closed by shock-expansion theory on the deflected panel \\
    A learned conservative correction to the potential & registered, run, and falsified \\
    The enriched aerodynamic model & amplitude elasticity of $0.44$ to the piston coefficients
      (Section~\ref{sec:rc19-aero}) \\
        \bottomrule
  \end{tabular}
\end{table}

\begin{figure}[htbp]
  \centering
  \includegraphics[width=\textwidth]{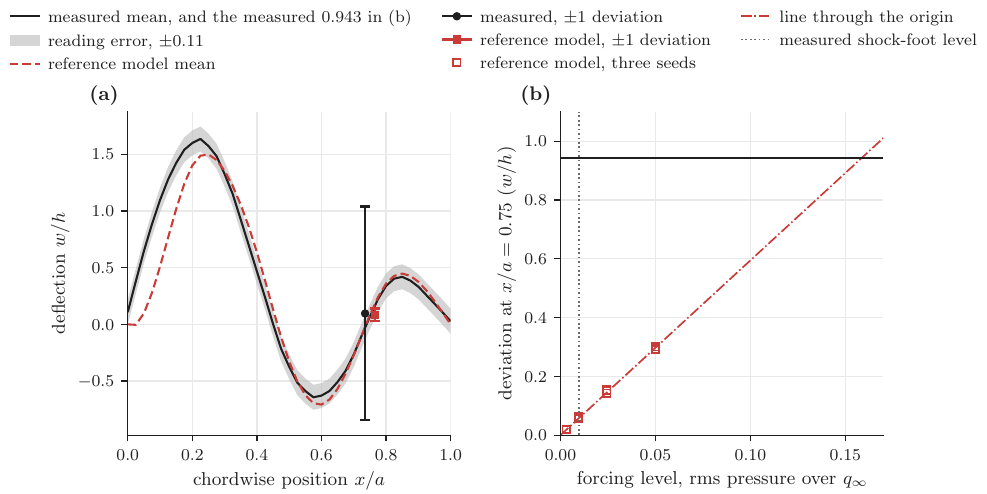}
  \caption{The attached-shock case, with the reference model against the measurement digitized
  from~\cite{piccoloserafim2025}. The reference model uses the measured static and thermal fields, seven measured modes
  and clamped edges. (a) The chordwise mean deflection on the centerline at the measured forcing level. The shaded band
  is the reading error of the digitized profile. The bars at three-quarter chord span plus and minus one standard
  deviation. (b) The model's standard deviation at three-quarter chord against the forcing level, which is the
  root-mean-square surface pressure divided by $q_\infty$. Each level has three random seeds. The deviation grows
  linearly with the forcing level and reaches the measured $0.943$ only near sixteen times the measured level. The
  dotted line marks that measured level. Each run lasts \SI{2.25}{\second}, and its first \SI{0.25}{\second} is
  discarded.}
  \label{fig:attached}
\end{figure}

What remains is structural rather than a parameter. The reference model is loaded by the pressure field measured on the rigid
panel, held fixed as the panel deflects. The real field depends on the panel's deformed shape, and that dependence changes
the effective aeroelastic stiffness. With the fixed field the reference model reproduces the measured chordwise mean shape to a root-mean-square
deviation of $0.22$ thicknesses. The profile digitized from~\cite{piccoloserafim2025} carries a reading error of about $0.11$. 
The mean is therefore not where it fails. Testing the mechanism needs either time-resolved surface data or a
computation on the deformed geometry. The public package carries neither. 
The same limitation shows in the reduced-order model of Brouwer et al., which needed adjusted inputs to place its own regimes.

Two items of data would close this comparison. The first is the time series behind the
published statistics at the nine points, or a statement that none was retained. Section~\ref{sec:rc19-assimilation} shows
that the published rate and a plausible noise level suffice to use it. The second is a computational record on the
deformed geometry of the attached-shock case, which would test the remaining mechanism directly.
\FloatBarrier

\section{HyMAX: the cantilever plate at Mach 5.85}
\label{sec:hymax}

The HyMAX model is the RC-19 construction applied to a different configuration. The plate is a cantilever rather than a clamped panel, at Mach $5.85$ rather than $1.92$. An oblique shock impinges on the plate, which the RC-19 training cases lack, and the inflow varies in time. What changed is the reference model and its generalized forces. The energy identity (Proposition~\ref{prop:identity}), the learned dissipation (Lemma~\ref{lem:spd}), the fitted force term, the objective and the integrator are unchanged from Section~\ref{sec:methods}. No term
was added to the objective.

Table~\ref{tab:hymax} and Figure~\ref{fig:hymax} give the relative root-mean-square error $e_{\mathrm{rms}}$ of
Equation~\eqref{eq:hymax-rms}, under the protocol of Section~\ref{sec:setup-hymax}. Over the early window, the first
\SI{1.9}{\milli\second}, it is \SI{0.31}{\percent}. Over the whole \SI{64.9}{\milli\second} rollout it is \SI{3.67}{\percent}, and that is the result to read. The error grows linearly rather than
exponentially, because the HyMAX response is not chaotic (Figure~\ref{fig:hymax-modes}a). The fitted force term
represents the time-varying aerodynamic operators only on average, and its small residual acts at the first-mode
frequency. On a mode damped at $\zeta = 0.01$ that residual builds up over cycles. The build-up stays linear for times
short against $1/(\zeta\omega_1) \approx \SI{190}{\milli\second}$. The trained model follows the reference model's trailing edge through about five first-mode cycles, $\SI{64.9}{\milli\second} \times \SI{83}{\hertz} \approx 5.4$, to within four percent, with nothing fed back.

\begin{table}[htbp]
  \centering
  \caption{HyMAX: relative root-mean-square error of the trailing-edge displacement in free rollout, against the reference model record
  member, for the trained model after $263$ epochs.}
  \label{tab:hymax}
  \begin{tabular}{lcc}
    \toprule
    Window & Duration & Relative error \\
    \midrule
    First $48$ time samples & \SI{1.9}{\milli\second} & \SI{0.31}{\percent} \\
    Full rollout & \SI{64.9}{\milli\second} & \SI{3.67}{\percent} \\
    \bottomrule
  \end{tabular}
\end{table}

\begin{figure}[htbp]
  \centering
  \includegraphics[width=\textwidth]{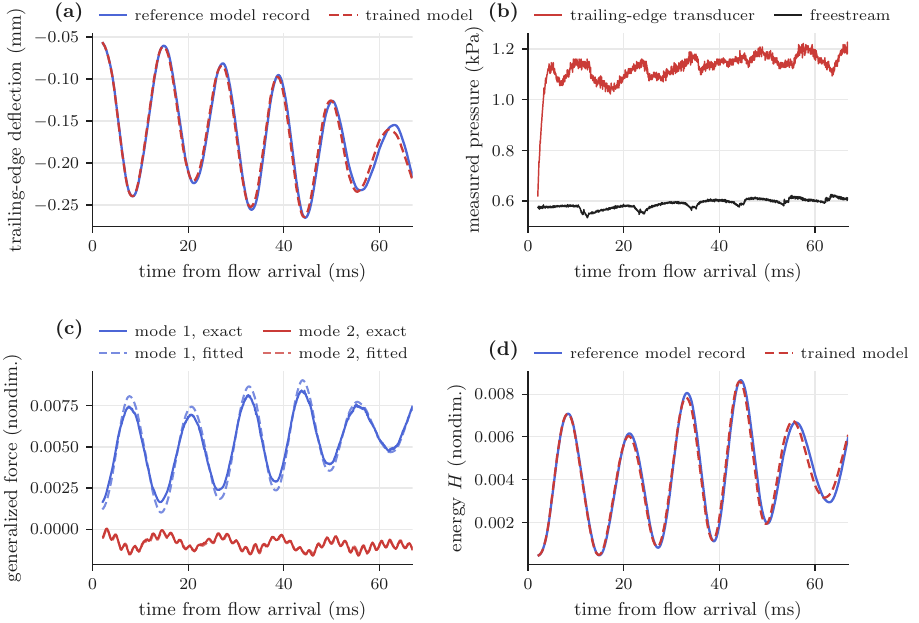}
  \caption{HyMAX in free rollout from one recorded state. (a) The trailing-edge deflection of the trained model against the
  reference model record member it is graded on. (b) The measured freestream and trailing-edge pressures that drive
  both. (c) The state-dependent generalized force of each mode, exact and fitted, evaluated on the record's states. (d)
  The energy along the rollout.}
  \label{fig:hymax}
\end{figure}

The error that remains is the one Section~\ref{sec:methods-head} anticipates. The reference model's aerodynamic operators follow
the measured dynamic pressure in time, and a fitted term with constant coefficients can represent them only on average. Over the record's states, the fitted generalized forces correlate with the reference model's exact forces at a Pearson coefficient of $0.998$, over both modes pooled. Figure~\ref{fig:hymax}(c) shows the
residual. The fitted first-mode force follows the exact one with a slightly larger swing, and the second mode's matches
closely.

Two further results concern the number of retained modes. A third mode does not reduce the error. Two and three modes give \SI{0.307}{}
and \SI{0.294}{\percent} over the first \SI{1.9}{\milli\second}, and \SI{3.67}{\percent} over the full rollout.
Figure~\ref{fig:hymax-modes} shows the two errors over the rollout. A six-mode basis is not admissible.
With six modes the state has twelve components, but the whitened coordinate of Section~\ref{sec:methods-identity} keeps only eleven independent directions.
The sixth mode, near \SI{7}{\kilo\hertz}, is resolved at the record's \SI{25}{\kilo\hertz} rate but carries no representable energy under the measured drive.
The map back to physical coordinates is then not invertible, and the basis is rejected. The largest usable basis is therefore determined by the record rather than assumed. The model reports what its data do not support,
instead of fitting a direction the data do not contain.

\begin{figure}[htbp]
  \centering
  \includegraphics[width=\textwidth]{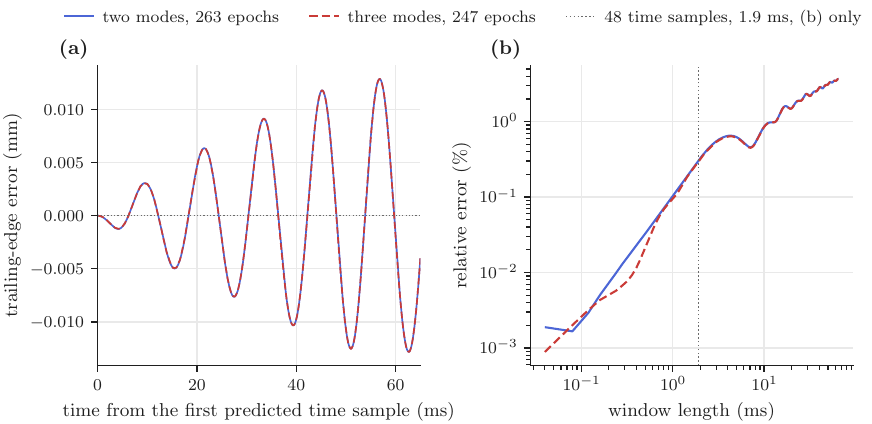}
  \caption{HyMAX trained models with two and three retained modes, trained for $263$ and $247$ epochs, in free rollout from
  the time sample of Figure~\ref{fig:hymax}. (a) The trailing-edge displacement error, model minus reference model record. (b)
  The relative root-mean-square error over a window from the first predicted time sample, against the window's length. The
  dotted line marks $48$ time samples, \SI{1.9}{\milli\second}.}
  \label{fig:hymax-modes}
\end{figure}

These results are graded against the HyMAX reference model, not against the experiment. The reference model's first frequency is the measured \SI{83}{\hertz} by construction. Its second, \SI{520}{\hertz}, lies within \SI{2}{\percent} of the \SI{530.4}{\hertz} computed by Xu and Dowell~\cite{xu2025}.
Its damping ratio, $\zeta = 0.01$, is a setting rather than a measurement.
\FloatBarrier

\section{Preregistered predictions and their outcomes}
\label{sec:predictions}

Table~\ref{tab:predictions} lists the predictions registered for the studies this report presents, in the order
registered, with their outcomes. A prediction is broken when any registered condition fails. Four of nine broke: A1, R2,
R3 and T5. All are retained as registered, and no threshold was moved. Twenty-one earlier predictions, made
while the reference models were developed, concern superseded configurations and are not listed.

\begin{table}[htbp]
  \centering
  \caption{The predictions registered for the RC-19 studies of this report, and their outcomes. Row A1 tests a
  correction to the reference model on the attached-shock case. Rows R1 to R3 test the reference model, and rows T1 to T5
  the trained model.}
  \label{tab:predictions}
  \small
  \begin{tabular}{p{0.04\textwidth}p{0.36\textwidth}p{0.38\textwidth}p{0.08\textwidth}}
    \toprule
    & Registered prediction & Outcome & Verdict \\
    \midrule
        A1 & A learned conservative correction to the potential, fitted on the attached-shock case, transfers unchanged:
      it reproduces the two no-shock regimes and their frequency within \SI{15}{\percent}, at under half the size of the
      original terms & size test passed; neither regime reproduced; the chaotic case turned periodic & broken \\
    R1 & Regime at the simulation operating conditions of Brouwer et al.: periodic at \SI{52.7}{\kilo\pascal}, \SI{8.8}{\kelvin}; chaotic at
      \SI{51.1}{\kilo\pascal}, \SI{13.6}{\kelvin}; both wells & as predicted in every repeat & held \\
    R2 & Scaling the aerodynamic operators by $1.054$ leaves the regime unchanged and moves the frequency less than \SI{2}{\percent}
      & regime unchanged; median frequency $+3.23$ and $-3.00$~\SI{}{\percent} in the two wells at the experimental operating condition & broken \\
    R3 & Mean frequency equivalent within \SI{2}{\percent}, and spectral change within sampling, both wells
      & frequency within; spectral band error $3.1$ and $3.7$ times the sampling level & broken \\
    T1 & Energy identity on the trained weights at most $10^{-12}$ & $2.2\times10^{-15}$ & held \\
    T2 & Periodic simulation operating condition: 16 of 16 periodic; frequency within one bin and amplitude within \SI{5}{\percent} of the reference model
      & 16 of 16; \SI{253.9}{\hertz}, zero spread; $0.7164$ against $0.7161$ & held \\
    T3 & Chaotic simulation operating condition: 16 of 16 chaotic on the trained and on the reference model with the same thermal field & 16 of 16 on each; medians
      $249.0$ and \SI{239.3}{\hertz} & held \\
    T4 & Ensemble: mean frequency within one bin of the reference model; \SI{113}{\hertz} absent in 30 of 30; amplitude within
      \SI{5}{\percent} & $238.4$ against \SI{237.2}{\hertz}; absent in 30 of 30; $0.760$ against $0.752$ & held \\
    T5 & Lyapunov exponent: sign agrees at six temperatures; within \SI{15}{\percent} where positive & sign at all six;
      \SI{17.6}{\percent} at \SI{8.7}{\kelvin} & broken \\
    \bottomrule
  \end{tabular}
\end{table}

\paragraph{R2 broke because its condition sat below its own noise} Every frequency the estimator returns lies on the
package's \SI{4.88}{\hertz} spectral grid. Two percent of \SI{236}{\hertz} is \SI{4.7}{\hertz}, less than one grid
interval. The two wells moved by the same \SI{7.3}{\hertz} in opposite directions. A change in the aerodynamic coefficients does not do that,
since it stiffens or softens both wells in the same sense. Within one repeat set, the frequency range, largest minus smallest, was
$39$ to \SI{44}{\hertz}, against a pooled standard deviation of \SI{16.3}{\hertz}. The condition placed a two percent
tolerance on a quantity whose standard deviation across repeats was \SI{16.3}{\hertz}, about seven percent. Its
estimator is quantized at \SI{4.88}{\hertz}, about two percent. It
was registered without checking it against the noise of the quantity it tested. The replacement R3 changed the
statistic, not the tolerance, and sized its sample from the measured spread before the run.

\paragraph{R3 broke on one of its two clauses, and that is its result} Its mean-frequency clause held and its
spectral clause did not. Section~\ref{sec:rc19-aero} reports both. Together they bound the local-coefficient enrichment of Brouwer et al.: it moves the amplitude by \SI{2.4}{\percent}, not the mean frequency.

\paragraph{T5 broke at one temperature} The follow-up measurement of Section~\ref{sec:rc19-chaos}, eight trajectories per model at \SI{8.7}{\kelvin}, was registered before it ran.
It found the \SI{17.6}{\percent} difference inside the scatter of single-trajectory exponents. That diagnosis
does not re-grade the prediction.

\paragraph{What the four trained-model predictions that held license} In T2 to T4 the trained model was run outside its training temperature, with the generalized-$\alpha$ step it was trained with. It returns the reference
model's regimes, ensemble frequency and amplitude, and, like the reference model, shows no \SI{113}{\hertz} peak. These results may therefore be attributed to the trained
model. The fitted force term reproduces the reference model's non-conservative forces to a relative root-mean-square error of
$3.6\times10^{-7}$. These results therefore
measure the fitted structure-preserving model and its integrator. They do not show that training learned the dynamics.

\section{Discussion}
\label{sec:discussion}

\paragraph{One generated record suffices when the physics is in the model's form} Neither challenge problem offers an ensemble, and
neither offers a time-resolved response. The models of this report needed neither, because almost nothing in them is
learned from data. The RC-19 model has $35$ trained numbers and $119$ fitted ones, against $2464$ fixed numbers in its
potential (Section~\ref{sec:methods-head}). The energy identity holds for any weights, hence at the random initialization and at every epoch. The optimizer therefore never had to find
the physics. What the record determines is the non-conservative part. On RC-19 a ridge-regularized least-squares fit determines it to a relative error of $3.6\times10^{-7}$ before the first epoch.
On HyMAX the same fit is limited by the time variation of the aerodynamic operators. There the fitted and exact forces have a Pearson correlation of $0.998$, over both modes pooled.

\paragraph{What training did here} On these records, training changed little. On RC-19 the fitted term absorbed the reference model's damping.
The trained dissipation's eigenvalues settled at $2\times10^{-8}$ to $10^{-7}$, against $0.010$ to $0.020$ for that damping. That is the expected outcome when the reference model's non-conservative forces are linear in the state. We report it as such,
not as a success of the optimizer. The
trained model's value in this report is that its results do not depend on the reference model being available. It runs at
conditions it was not trained at, from its trained parameters alone, with its energy identity and positive-definite dissipation intact. Whether the trained
dissipation earns its place has to be tested on a record whose dissipation is not linear, which a measured response
would provide.

\paragraph{Structure tells the user what the data cannot support} The HyMAX construction rejected a six-mode basis, because the measured drive does not
excite the sixth mode. A penalty formulation could fit the unsupported direction and still report a small loss. The RC-19 reference model reports the work the flow does on the panel, and the work the assimilation correction does on it. Both are terms
of the discrete energy budget, which closes to roundoff. In both cases the built-in structure is more useful as a diagnostic than as a constraint. A failure appears at the step where it occurs, as a budget
term, and not later as a drift.

\paragraph{Statistics, chaos and assimilation} A chaotic panel cannot be forecast trajectory by trajectory beyond a horizon
of tens of milliseconds at practical starting accuracy. Even an exact model computed in double precision loses the
trajectory after about \SI{100}{\milli\second}. The comparisons that are reproducible are statistics. Every ensemble mean of this report is quoted with its member count and kick interval, because
Section~\ref{sec:rc19-agreement} shows that the mean depends on that interval.
Past the horizon, only measurement can keep a model on the realized trajectory. The twin experiment of Section~\ref{sec:rc19-assimilation} shows that the nine published points, sampled at \SI{5}{\kilo\hertz}, would suffice for offline reconstruction.
That requires the time series behind the published statistics, which the package does not hold.

\paragraph{The preregistered predictions} Each of the four broken predictions of Table~\ref{tab:predictions} carries a
lesson. A1 showed that a correction fitted at one condition is a calibration, not physics. R2 showed that an
acceptance condition must be sized against the noise of its statistic before it is registered. T5 showed that a
single-trajectory Lyapunov exponent at \SI{8.7}{\kelvin} is a range, not a value. That also qualifies the quoted rise of the exponent with temperature, which uses one trajectory per temperature.
A preregistration in which every prediction held would have been less informative and less credible.

\section{Limitations}
\label{sec:limitations}

The limitations of this work follow mostly from the data the two public packages contain. Both packages were assembled
to grade computations against statistics, not to train or to track a model. We describe first what they hold and what
they lack. We then state what each gap implies for the results, and close with the limits of scope that are our own.

\paragraph{The data available} The RC-19 response is reported as statistics at nine points in the published
papers~\cite{brouwer2021,piccoloserafim2025}. These are the mean, the standard deviation and spectra. The public package
defines the points but does not hold the statistics. It holds no time history, and each case is a single run with no repeats. Because the
response is chaotic, a single statistic carries sampling spread that one run cannot estimate. The package gives no
material constants. The plate's modulus, density and thermal expansion are therefore taken from a materials database~\cite{matweb4140}. The temperature field
rests on two thermocouples. One is on the panel and the other is on the frame. The panel thermocouple sits in the thermal
gradient near the frame~\cite{brouwer2021}. The field's shape is therefore inferred rather than measured. Surface pressures were
measured on a rigid panel only, never on the deflecting one. The HyMAX package holds one run of surface pressure at four
points. Its one structural frequency, \SI{83}{\hertz}, comes from~\cite{currao2019}. It holds no displacement, no damping estimate, and no measure of the pressure
behind the plate or of the shock's footprint on it.

\paragraph{What the gaps imply for the learned models} With no time-resolved response, the record each model is trained
on is generated by the reference model rather than measured. The comparisons of the learned model are therefore
self-consistency, and the comparisons with experiment are made through statistics alone. A second consequence follows
from the reference model's form. Its non-conservative forces are linear in the state, and a least-squares fit recovers
them to a relative root-mean-square error of $3.6\times10^{-7}$ before training begins. The accuracy in Table~\ref{tab:self} is thus set by that fit,
and training changes it little. Whether the learned dissipation adds accuracy on its own needs a record whose dissipation is not linear. A
measured record would be one. The assimilation result of Section~\ref{sec:rc19-assimilation}
has the same origin. With no measured time series to assimilate, it is a twin experiment. It establishes what the
published sampling rate and layout can support, and it is not agreement with the experiment.

On HyMAX the gaps are wider. The back pressure and shock position were not measured. The record therefore spans three back pressures and two shock positions.
The damping ratio, $\zeta = 0.01$, is a setting, and the HyMAX results are graded against that case's own reference model. The
package states that a measured trailing-edge displacement will be provided for validation~\cite{aepw4hswg}. It is
not used here, and
Section~\ref{sec:hymax} contains no comparison with experiment.

\paragraph{Limits of the comparisons with experiment} Four RC-19 comparisons fall short, and each traces to an input
rather than to the method. At the periodic simulation operating condition, our models' frequency, \SI{253.9}{\hertz}, sits \SI{7.6}{\percent} above the
measured \SI{236}{\hertz}. On the chaotic case, the reference model's ensemble mean of \SI{239.6}{\hertz} sits
\SI{6.8}{\percent} below the measured \SI{257}{\hertz}.
Neither our models nor that of Brouwer et al.\ returns the measured periodic regime at the measured inputs. This points to the temperature field the package
cannot resolve. The reference model also uses the Mach number of Brouwer et al., $1.94$. The package's computed
value is $1.92$. At the same stagnation conditions, the freestream static pressure rises as the Mach number falls. The
dynamic pressure would then rise by about one percent, and the piston coefficient $2q_\infty/\beta$ by about
\SI{2.5}{\percent}. On the attached-shock case, the reference model's fluctuation level
is about sixteen times below the measurement, and no learned model was built for that case. Section~\ref{sec:rc19-open} leaves the
dependence of the pressure field on the deflected shape as the one untested cause. The package measures that field on a rigid panel
only.

Two statistical limits apply to every comparison. An ensemble mean depends on the interval of initial perturbations it
sampled. Each mean in this report is therefore quoted with its interval. Likewise, a Lyapunov exponent from one trajectory is a
range, not a value. The trend over temperature was computed from one trajectory per temperature. At \SI{8.7}{\kelvin}, the eight
trajectories per model of Section~\ref{sec:rc19-chaos} span ranges wider than the difference between the models.

\paragraph{Limits of scope} We did not measure the trained model's cost per time step, and no comparison of cost against computational
fluid dynamics was made. No speedup is claimed, and none is implied.

\section{Conclusions and outlook}
\label{sec:conclusions}

We built structure-preserving learned reduced-order models of both AePW-4 High-Speed Working Group challenge problems
with one construction. Each model was built from the data package that the working group distributes for its case. The
energy identity holds to roundoff on the trained weights. The learned dissipation cannot supply energy at any parameter
value. The discrete-gradient step closes the energy budget at every step, including under assimilation.

On RC-19, inviscid shock relations applied to the package's wedge geometry and schlieren shock angle predict a shock
pressure ratio of $1.521$. The measured value comes from the pressure-sensitive paint that Brouwer et
al.~\cite{brouwer2023pressure} applied to the rigid panel. The ratio of its centerline pressures behind and ahead of the
shock is also $1.521$. The trained model's thirty-member ensemble places the mean response frequency within one percent
of the \SI{236}{\hertz} measured by Brouwer et al.~\cite{brouwer2021}. At two conditions taken from their simulations,
the trained model returns the regimes of their model in every repeat. It does so without the \SI{113}{\hertz} peak of
their model, which the experiment does not show. Brouwer et al.\ enrich piston theory with local coefficients. On the
reference model, that enrichment changes the fluctuation amplitude by only $0.44$ percent per percent change in the
coefficients.

The chaos of the panel limits any model to a tracking horizon of tens of milliseconds. In a twin experiment,
assimilation at the nine published points and the published \SI{5}{\kilo\hertz} rate holds the state past that horizon.
The tracking error then stays near one percent of a thickness. On HyMAX the trained model follows its reference model's
trailing edge for \SI{65}{\milli\second} with nothing fed back. Its relative error over that rollout is
\SI{3.67}{\percent}.

The open items are data rather than method. The time series behind the RC-19 statistics would allow the experiment
itself to be assimilated, and trajectories rather than statistics to be graded. They would also allow training on a
measured record. The fitted force term and the learned dissipation would then absorb what the reference physics misses.

A deployed digital twin would run the learned model with the assimilation force on live measurements. Three steps
separate the twin experiment of Section~\ref{sec:rc19-assimilation} from that use. First, real-time tracking must use
the causal hold, which holds the state to $0.325$ thicknesses at \SI{5}{\kilo\hertz} against the smoother's $0.0125$. A
twin would therefore need a higher observation rate or a better causal reconstruction. Second, the correction must
absorb model error as well as state error. The twin experiment used an exact model and did not test that. Third, the
correction has run only on the reference model. With the learned model, the observer's work in the energy budget would
measure what the model misses.

Further data would close the remaining comparisons. A computation or a measurement on the deformed geometry of the
attached-shock case would test the one mechanism left for its fluctuation deficit. That mechanism is the dependence of
the pressure field on the deflected shape. The HyMAX trailing-edge displacement would turn that case's
self-consistency into experimental agreement. On the side of the method, two records are missing. Testing what the
trained dissipation can carry needs a record whose non-conservative forces are not linear in the state. Testing the
temperature dependence of the fitted term needs a record that spans temperature.

The construction is not specific to these panels. Equation~\eqref{eq:reference model-H} is a Ritz--Galerkin model of a von
K\'arm\'an plate, and a finite element model can supply its operators directly~\cite{mignolet2013review}. Fitting then sets
every term whose form the physics fixes. Training is reserved for the remaining terms. With finite element records across
material and thermal parameters, the conditioned model could be tested on held-out parameter sets, registered in advance.
That test requires a modal basis valid across the parameter range, and data that excite every retained mode.

\section*{Reproducibility}
\label{sec:repro}

\paragraph{Data availability} The challenge-problem packages are distributed by the AePW-4 High-Speed Working
Group~\cite{aepw4hswg}. The records of this report are generated from them by the reference models of
Section~\ref{sec:methods}. Inquiries may be directed to the corresponding author.

\section*{References}
\bibliographystyle{elsarticle-num}
\bibliography{spmlhswg}

\begin{thebibliography}{10}
\expandafter\ifx\csname url\endcsname\relax
  \def\url#1{\texttt{#1}}\fi
\expandafter\ifx\csname urlprefix\endcsname\relax\def\urlprefix{URL }\fi
\expandafter\ifx\csname href\endcsname\relax
  \def\href#1#2{#2} \def\path#1{#1}\fi

\bibitem{aepw4hswg}
{NASA Engineering and Safety Center Academy}, Fourth aeroelastic prediction
  workshop: High speed working group,
  \url{https://nescacademy.nasa.gov/workshops/AePW4/public/wg/high_speed},
  accessed 3 September 2026.

\bibitem{brouwer2023pressure}
K.~R. Brouwer, R.~A. Perez, T.~J. Beberniss, S.~M. Spottswood, Surface pressure
  measurements and predictions in shock-dominated flows, in: AIAA AVIATION 2023
  Forum, San Diego, CA, 2023, aIAA Paper 2023-4134.
\newblock \href {http://dx.doi.org/10.2514/6.2023-4134}
  {\path{doi:10.2514/6.2023-4134}}.

\bibitem{brouwer2021}
K.~R. Brouwer, R.~A. Perez, T.~J. Beberniss, S.~M. Spottswood, D.~A. Ehrhardt,
  R.~Wiebe, Investigation of aeroelastic instabilities for a thin panel in
  turbulent flow, Nonlinear Dynamics 104 (2021) 3323--3346.
\newblock \href {http://dx.doi.org/10.1007/s11071-021-06571-4}
  {\path{doi:10.1007/s11071-021-06571-4}}.

\bibitem{dowell1970review}
E.~H. Dowell, Panel flutter: a review of the aeroelastic stability of plates
  and shells, AIAA Journal 8~(3) (1970) 385--399.
\newblock \href {http://dx.doi.org/10.2514/3.5680} {\path{doi:10.2514/3.5680}}.

\bibitem{mei1999review}
C.~Mei, K.~Abdel-Motagaly, R.~Chen, Review of nonlinear panel flutter at
  supersonic and hypersonic speeds, Applied Mechanics Reviews 52~(10) (1999)
  321--332.
\newblock \href {http://dx.doi.org/10.1115/1.3098919}
  {\path{doi:10.1115/1.3098919}}.

\bibitem{mcnamara2011hypersonic}
J.~J. McNamara, P.~P. Friedmann, Aeroelastic and aerothermoelastic analysis in
  hypersonic flow: past, present, and future, AIAA Journal 49~(6) (2011)
  1089--1122.
\newblock \href {http://dx.doi.org/10.2514/1.J050882}
  {\path{doi:10.2514/1.J050882}}.

\bibitem{dowell1982chaotic}
E.~H. Dowell, Flutter of a buckled plate as an example of chaotic motion of a
  deterministic autonomous system, Journal of Sound and Vibration 85~(3) (1982)
  333--344.
\newblock \href {http://dx.doi.org/10.1016/0022-460X(82)90259-0}
  {\path{doi:10.1016/0022-460X(82)90259-0}}.

\bibitem{brouwer2025aepw4}
K.~Brouwer, {AePW-4} high-speed working group: An overview of recent progress
  and future directions, AIAA AVIATION Forum, Aeroelastic Prediction Workshop
  session, presentation, 22 July 2025. Distribution A, AFRL-2025-3295 (2025).

\bibitem{piccoloserafim2025}
L.~Piccolo~Serafim, E.~Dowell, K.~Brouwer, Correlation of mathematical model
  with {AFRL} {RC-19} aerothermoelastic experiment including an impinging
  shock, AIAA AVIATION Forum, Las Vegas, NV, presentation, 22 July 2025 (2025).

\bibitem{aguilo2026spml}
M.~A. Aguil\'o, Structure-preserving reduced-order forecasting of transient,
  parametric flows with neural {W}hitney forms and a hybridizable neural
  integrator, technical report, Morphorm LLC; arXiv submission in preparation
  (2026).

\bibitem{aguilo2026spmlfsi}
M.~A. Aguil\'o, A structure-preserving mortar coupling of reduced-order fluid
  and structure forecasters for fluid--structure interaction, technical report,
  Morphorm LLC; arXiv submission in preparation (2026).

\bibitem{aguilo2026spmlopen}
M.~A. Aguil\'o, Structure-preserving reduced-order forecasting of open flows
  with prescribed and learned boundary fluxes, technical report, Morphorm LLC;
  arXiv submission in preparation (2026).

\bibitem{schuirmann1987tost}
D.~J. Schuirmann, A comparison of the two one-sided tests procedure and the
  power approach for assessing the equivalence of average bioavailability,
  Journal of Pharmacokinetics and Biopharmaceutics 15~(6) (1987) 657--680.
\newblock \href {http://dx.doi.org/10.1007/BF01068419}
  {\path{doi:10.1007/BF01068419}}.

\bibitem{spottswood2019exploring}
S.~M. Spottswood, T.~J. Beberniss, T.~G. Eason, R.~A. Perez, J.~M. Donbar,
  D.~A. Ehrhardt, Z.~B. Riley, Exploring the response of a thin, flexible panel
  to shock-turbulent boundary-layer interactions, Journal of Sound and
  Vibration 443 (2019) 74--89.
\newblock \href {http://dx.doi.org/10.1016/j.jsv.2018.11.035}
  {\path{doi:10.1016/j.jsv.2018.11.035}}.

\bibitem{brouwer2021experiments}
K.~R. Brouwer, R.~A. Perez, T.~J. Beberniss, S.~M. Spottswood, D.~A. Ehrhardt,
  Experiments on a thin panel excited by turbulent flow and
  shock/boundary-layer interactions, AIAA Journal 59~(7) (2021) 2737--2752.
\newblock \href {http://dx.doi.org/10.2514/1.J060114}
  {\path{doi:10.2514/1.J060114}}.

\bibitem{currao2019}
G.~Currao, A.~Neely, C.~Kennell, S.~Gai, D.~Buttsworth, Hypersonic
  fluid--structure interaction on a cantilevered plate with shock impingement,
  AIAA Journal 57~(11) (2019) 4819--4834.
\newblock \href {http://dx.doi.org/10.2514/1.J058375}
  {\path{doi:10.2514/1.J058375}}.

\bibitem{currao2021oscillating}
G.~M.~D. Currao, L.~P. McQuellin, A.~J. Neely, S.~L. Gai, S.~O'Byrne,
  F.~Zander, D.~R. Buttsworth, J.~J. McNamara, I.~Jahn, Hypersonic oscillating
  shock-wave/boundary-layer interaction on a flat plate, AIAA Journal 59~(3)
  (2021) 940--959.
\newblock \href {http://dx.doi.org/10.2514/1.J059590}
  {\path{doi:10.2514/1.J059590}}.

\bibitem{ames1953naca1135}
{Ames Research Staff}, Equations, tables, and charts for compressible flow,
  NACA Report 1135, National Advisory Committee for Aeronautics (1953).

\bibitem{xu2025}
H.~Xu, E.~Dowell, Computational modeling of the {HyMAX} experiment in support
  of a new aeroelastic experimental design, AIAA AVIATION Forum, Aeroelastic
  Prediction Workshop session, presentation, 22 July 2025 (2025).

\bibitem{lighthill1953oscillating}
M.~J. Lighthill, Oscillating airfoils at high {M}ach number, Journal of the
  Aeronautical Sciences 20~(6) (1953) 402--406.
\newblock \href {http://dx.doi.org/10.2514/8.2657} {\path{doi:10.2514/8.2657}}.

\bibitem{ashley1956piston}
H.~Ashley, G.~Zartarian, Piston theory: a new aerodynamic tool for the
  aeroelastician, Journal of the Aeronautical Sciences 23~(12) (1956)
  1109--1118.
\newblock \href {http://dx.doi.org/10.2514/8.3740} {\path{doi:10.2514/8.3740}}.

\bibitem{dowell1975plates}
E.~H. Dowell, Aeroelasticity of Plates and Shells, Noordhoff International,
  Leyden, 1975.

\bibitem{mignolet2013review}
M.~P. Mignolet, A.~Przekop, S.~A. Rizzi, S.~M. Spottswood, A review of
  indirect/non-intrusive reduced order modeling of nonlinear geometric
  structures, Journal of Sound and Vibration 332~(10) (2013) 2437--2460.
\newblock \href {http://dx.doi.org/10.1016/j.jsv.2012.10.017}
  {\path{doi:10.1016/j.jsv.2012.10.017}}.

\bibitem{brouwer2019enriched}
K.~R. Brouwer, J.~J. McNamara, Enriched piston theory for expedient aeroelastic
  loads prediction in the presence of shock impingements, AIAA Journal 57~(3)
  (2019) 1288--1302.
\newblock \href {http://dx.doi.org/10.2514/1.J057595}
  {\path{doi:10.2514/1.J057595}}.

\bibitem{ganti2025panel}
H.~Ganti, L.~Bravo, A.~Ghosal, P.~Khare, Mutual interactions between a thin
  flexible panel and supersonic flows, arXiv preprint arXiv:2501.08875.

\bibitem{raissi2019pinn}
M.~Raissi, P.~Perdikaris, G.~E. Karniadakis, Physics-informed neural networks:
  A deep learning framework for solving forward and inverse problems involving
  nonlinear partial differential equations, Journal of Computational Physics
  378 (2019) 686--707.
\newblock \href {http://dx.doi.org/10.1016/j.jcp.2018.10.045}
  {\path{doi:10.1016/j.jcp.2018.10.045}}.

\bibitem{greydanus2019hnn}
S.~Greydanus, M.~Dzamba, J.~Yosinski, Hamiltonian neural networks, in: Advances
  in Neural Information Processing Systems (NeurIPS), Vol.~32, 2019, pp.
  15353--15363.

\bibitem{vanderSchaftJeltsema2014}
A.~van~der Schaft, D.~Jeltsema, Port-{H}amiltonian systems theory: an
  introductory overview, Vol.~1 of Foundations and Trends in Systems and
  Control, Now Publishers, 2014.

\bibitem{desai2021porthnn}
S.~A. Desai, M.~Mattheakis, D.~Sondak, P.~Protopapas, S.~J. Roberts,
  Port-{H}amiltonian neural networks for learning explicit time-dependent
  dynamical systems, Physical Review E 104~(3) (2021) 034312.
\newblock \href {http://dx.doi.org/10.1103/PhysRevE.104.034312}
  {\path{doi:10.1103/PhysRevE.104.034312}}.

\bibitem{quispel2008avf}
G.~R.~W. Quispel, D.~I. McLaren, A new class of energy-preserving numerical
  integration methods, Journal of Physics A: Mathematical and Theoretical
  41~(4) (2008) 045206.
\newblock \href {http://dx.doi.org/10.1088/1751-8113/41/4/045206}
  {\path{doi:10.1088/1751-8113/41/4/045206}}.

\bibitem{gonzalez1996discrete}
O.~Gonzalez, Time integration and discrete {H}amiltonian systems, Journal of
  Nonlinear Science 6~(5) (1996) 449--467.
\newblock \href {http://dx.doi.org/10.1007/BF02440162}
  {\path{doi:10.1007/BF02440162}}.

\bibitem{mclachlan1999discrete}
R.~I. McLachlan, G.~R.~W. Quispel, N.~Robidoux, Geometric integration using
  discrete gradients, Philosophical Transactions of the Royal Society A
  357~(1754) (1999) 1021--1045.
\newblock \href {http://dx.doi.org/10.1098/rsta.1999.0363}
  {\path{doi:10.1098/rsta.1999.0363}}.

\bibitem{lowdin1950nonorthogonality}
P.-O. L{\"o}wdin, On the non-orthogonality problem connected with the use of
  atomic wave functions in the theory of molecules and crystals, The Journal of
  Chemical Physics 18~(3) (1950) 365--375.
\newblock \href {http://dx.doi.org/10.1063/1.1747632}
  {\path{doi:10.1063/1.1747632}}.

\bibitem{matweb4140}
{MatWeb}, {AISI} 4140 steel, annealed at 815\,{\textdegree}{C}
  (1500\,{\textdegree}{F}), furnace cooled 11\,{\textdegree}{C}/hour to
  665\,{\textdegree}{C}, air cooled, 25~mm round, material property data sheet,
  \url{https://www.matweb.com/search/DataSheet.aspx?MatGUID=7b75475aa1bc41618788f63c6500d36b},
  density, modulus, Poisson's ratio, linear thermal expansion
  (0--100\,{\textdegree}{C}) and thermal conductivity (100\,{\textdegree}{C}).

\bibitem{white2006viscous}
F.~M. White, Viscous Fluid Flow, 3rd Edition, McGraw-Hill, New York, 2006.

\bibitem{incropera2007heat}
F.~P. Incropera, D.~P. DeWitt, T.~L. Bergman, A.~S. Lavine, Fundamentals of
  Heat and Mass Transfer, 6th Edition, Wiley, Hoboken, NJ, 2007.

\bibitem{eckert1956}
E.~R.~G. Eckert, Engineering relations for heat transfer and friction in
  high-velocity laminar and turbulent boundary-layer flow over surfaces with
  constant pressure and temperature, Transactions of the ASME 78~(6) (1956)
  1273--1283.
\newblock \href {http://dx.doi.org/10.1115/1.4014011}
  {\path{doi:10.1115/1.4014011}}.

\bibitem{matweb6061t6}
{MatWeb}, Aluminum 6061-{T6}; 6061-{T651}, material property data sheet,
  \url{https://www.matweb.com/search/datasheet.aspx?matguid=b8d536e0b9b54bd7b69e4124d8f1d20a},
  cited by the HyMAX challenge-problem package as the plate's material
  specification.

\bibitem{blevins1979formulas}
R.~D. Blevins, Formulas for Natural Frequency and Mode Shape, Van Nostrand
  Reinhold, New York, 1979.

\bibitem{hoerl1970ridge}
A.~E. Hoerl, R.~W. Kennard, Ridge regression: biased estimation for
  nonorthogonal problems, Technometrics 12~(1) (1970) 55--67.
\newblock \href {http://dx.doi.org/10.1080/00401706.1970.10488634}
  {\path{doi:10.1080/00401706.1970.10488634}}.

\bibitem{chungHulbert1993}
J.~Chung, G.~M. Hulbert, A time integration algorithm for structural dynamics
  with improved numerical dissipation: the generalized-$\alpha$ method, Journal
  of Applied Mechanics 60~(2) (1993) 371--375.
\newblock \href {http://dx.doi.org/10.1115/1.2900803}
  {\path{doi:10.1115/1.2900803}}.

\bibitem{jansenWhitingHulbert2000}
K.~E. Jansen, C.~H. Whiting, G.~M. Hulbert, A generalized-$\alpha$ method for
  integrating the filtered {N}avier--{S}tokes equations with a stabilized
  finite element method, Computer Methods in Applied Mechanics and Engineering
  190~(3--4) (2000) 305--319.
\newblock \href {http://dx.doi.org/10.1016/S0045-7825(00)00203-6}
  {\path{doi:10.1016/S0045-7825(00)00203-6}}.

\bibitem{bai2019deq}
S.~Bai, J.~Z. Kolter, V.~Koltun, Deep equilibrium models, in: Advances in
  Neural Information Processing Systems, Vol.~32, 2019.

\bibitem{kingma2015adam}
D.~P. Kingma, J.~Ba, {Adam}: A method for stochastic optimization, in:
  International Conference on Learning Representations (ICLR), 2015,
  arXiv:1412.6980 [cs.LG].

\bibitem{loshchilov2017sgdr}
I.~Loshchilov, F.~Hutter, {SGDR}: Stochastic gradient descent with warm
  restarts, in: International Conference on Learning Representations, 2017.

\bibitem{hoke1976nudging}
J.~E. Hoke, R.~A. Anthes, The initialization of numerical models by a
  dynamic-initialization technique, Monthly Weather Review 104~(12) (1976)
  1551--1556.
\newblock \href
  {http://dx.doi.org/10.1175/1520-0493(1976)104<1551:TIONMB>2.0.CO;2}
  {\path{doi:10.1175/1520-0493(1976)104<1551:TIONMB>2.0.CO;2}}.

\bibitem{luenberger1964observer}
D.~G. Luenberger, Observing the state of a linear system, IEEE Transactions on
  Military Electronics 8~(2) (1964) 74--80.
\newblock \href {http://dx.doi.org/10.1109/TME.1964.4323124}
  {\path{doi:10.1109/TME.1964.4323124}}.

\bibitem{benettin1980lyapunov}
G.~Benettin, L.~Galgani, A.~Giorgilli, J.-M. Strelcyn, Lyapunov characteristic
  exponents for smooth dynamical systems and for {H}amiltonian systems; a
  method for computing all of them. {P}art 1: Theory, Meccanica 15~(1) (1980)
  9--20.
\newblock \href {http://dx.doi.org/10.1007/BF02128236}
  {\path{doi:10.1007/BF02128236}}.

\bibitem{kalnay2003assimilation}
E.~Kalnay, Atmospheric Modeling, Data Assimilation and Predictability,
  Cambridge University Press, Cambridge, 2003.

\bibitem{nosek2018preregistration}
B.~A. Nosek, C.~R. Ebersole, A.~C. DeHaven, D.~T. Mellor, The preregistration
  revolution, Proceedings of the National Academy of Sciences 115~(11) (2018)
  2600--2606.
\newblock \href {http://dx.doi.org/10.1073/pnas.1708274114}
  {\path{doi:10.1073/pnas.1708274114}}.

\bibitem{anderson2007fundamentals}
J.~D. Anderson, Jr., Fundamentals of Aerodynamics, 4th Edition, McGraw-Hill,
  New York, 2007.

\bibitem{mann1947test}
H.~B. Mann, D.~R. Whitney, On a test of whether one of two random variables is
  stochastically larger than the other, The Annals of Mathematical Statistics
  18~(1) (1947) 50--60.
\newblock \href {http://dx.doi.org/10.1214/aoms/1177730491}
  {\path{doi:10.1214/aoms/1177730491}}.

\end{thebibliography}

\end{document}